\documentclass[12pt,a4paper]{article}

\usepackage{geometry}
\usepackage{graphicx}
\usepackage{amsmath}
\usepackage{natbib}

\usepackage{amssymb}
\usepackage{amsthm}
\usepackage{psfrag,epsf}
\usepackage{enumerate}
\usepackage{hyperref}    

\usepackage{booktabs}
\usepackage{multirow}
\usepackage{adjustbox}
\usepackage{makecell}
\usepackage{tikz}
\usetikzlibrary{arrows,shapes.arrows,shapes.geometric, shapes.multipart,backgrounds,decorations.pathmorphing,positioning,fit,automata}
\usepackage{caption}

\usepackage{texab}      

\newtheorem{definition}{Definition}
\newtheorem{lemma}{Lemma}
\newtheorem{assumption}{Assumption}
\newtheorem{theorem}{Theorem}
\newtheorem{Coro}{Corollary}
\newtheorem{example}{Example}
\newtheorem{remark}{Remark}

\newcommand{\supp}{\text{Supp}}
\newcommand{\floor}[1]{\lfloor #1 \rfloor}

\usepackage[perpage]{footmisc} 

\newcommand{\jinzhou}[1]{\footnote{{  {Jinzhou: #1}}}}

\title{Estimation and Inference of Extremal Quantile Treatment Effects for Heavy-Tailed Distributions}
\author{David Deuber\thanks{These authors contributed equally to this work.}
	\and Jinzhou Li\footnotemark[1]
	\and Sebastian Engelke 
	\and Marloes H. Maathuis}
\date{\today}

\begin{document}

\maketitle

\begin{abstract}
Causal inference for extreme events has many potential applications in fields such as climate science, medicine and economics. We study the extremal quantile treatment effect of a binary treatment on a continuous, heavy-tailed outcome. Existing methods are limited to the case where the quantile of interest is within the range of the observations. For applications in risk assessment, however, the most relevant cases relate to extremal quantiles that go beyond the data range. We introduce an estimator of the extremal quantile treatment effect that relies on asymptotic tail approximation,
and use a new causal Hill estimator for the extreme value indices of potential outcome distributions. 
We establish asymptotic normality of the estimators 
and propose a consistent variance estimator to achieve valid statistical inference. 
We illustrate the performance of our method in simulation studies, and apply it to a real data set to estimate the extremal quantile treatment effect of college education on wage.
\end{abstract}


\section{Introduction}

Quantifying causal effects of binary treatments on extreme events is an important problem in many fields of research.
Examples include the effect of anthropogenic forcing on extreme precipitation,
the effect of smoking on low birth weights,
and the effect of education on high wages \citep[e.g.,][]{madakumbura2021anthropogenic, dessi2018exposure, heckman2018returns}.
The quantile treatment effect (QTE) \citep[e.g.,][]{doksum1974empirical, lehmann1975nonparametrics},
which is based on the potential outcome framework,
quantifies such causal effects.

Formally, for a binary treatment $D\in\{0,1\}$ and an outcome $Y \in \R$, let $Y(0)$ and $Y(1)$ denote the potential outcomes of $Y$ under treatment $D=0$ and $D=1$ respectively.
The fundamental problem of causal inference is that for each sample unit, only one of the outcomes can be observed, namely the one under the given treatment.
The observed response is $Y = Y(1)D + Y(0)(1-D)$ (i.e., we make the stable unit treatment value assumption).
Causal effects of $D$ on $Y$ can be defined in various ways according to different targets of interest. 
An example is the often used average treatment effect,
which is defined as $\E[Y(1)] - \E[Y(0)]$. 
In this paper, the causal effect at extremely high/low quantiles is our main interest.
Let the $\tau$-QTE be defined as
\begin{equation}\label{QTE_intro}
\delta(\tau) := q_1(\tau) - q_0(\tau),
\end{equation}
where $\tau \in (0,1)$ and $q_j(\tau) := \inf \{t \in \R: F_j(t) \geq \tau \}$ denotes the $\tau$-quantile of the potential outcome $Y(j)$, and $F_j(t)$ denotes its distribution function. 
We will treat the case where $\tau$ is very close to $0$ or $1$.


The QTE is generally not identifiable from observational data without making additional assumptions.
A commonly made assumption, called the unconfoundedness assumption \citep[e.g.,][]{rosenbaum1983central}, is that $(Y(1), Y(0))  \indep D \mid  X$  for some set of observed covariates $X \in \R^r$.
Under this assumption, the propensity score defined as $\Pi(x) := P(D=1 \mid X=x)$ can be used to adjust for confounding in the binary treatment setting, 
and the QTE is identifiable.

Along these lines,
\cite{firpo2007efficient} introduced an adjusted quantile estimator $\wh q_j(\tau)$ (see \eqref{eq:firpo_estimator}), 
defined as the minimizer of the inverse propensity score weighted empirical quantile loss,
to estimate the $\tau$-QTE in~\eqref{QTE_intro} for a fixed quantile level $\tau \in (0,1)$.
He also established the asymptotic normality of this estimator, allowing statistical inference.

For extreme events, the interest lies in the $\tau$-QTE for $\tau$ close to $0$ or $1$.
Considering the lower tail,
we allow the quantile $\tau=\tau_n $ to converge to zero as the sample size $n $ tends to infinity. 
In particular, we distinguish between three cases based on different rates of $\tau_n$,
and we call them: (a) intermediate: if $\tau_n \to 0$ and $n \tau_n \to \infty$;
(b) moderately extreme: if $\tau_n \to 0$ and $n \tau_n \to d > 0$; and
(c) extreme: if $\tau_n \to 0$ and $n \tau_n \to 0$.
Here $n\tau_n$ is the effective sample size or the expected number of observations below the $\tau_n$-quantile in a sample of size $n$.
We note that the cases $n \tau_n \to d > 0$ and $n \tau_n \to 0$ are sometimes referred to as ``extreme" and ``very extreme" \citep[][]{chernozhukov2005extremal, chernozhukov2016extremal, zhang2018extremal}.

The asymptotic results of $\wh q_j(\tau)$ with fixed $\tau$ in \cite{firpo2007efficient} no longer hold in the framework of changing levels $\tau_n$,
and \cite{zhang2018extremal} first established the asymptotic theory for this estimator in this framework.
Specifically, in the intermediate case, \cite{zhang2018extremal} showed that the estimator is asymptotically normal and suggested a valid full-sample bootstrap confidence interval. 
For the moderately extreme case, he showed that the limiting distribution of the estimator is no longer Gaussian, and proposed a $b$ out of $n$ bootstrap for valid inference.

For many applications, 
the most relevant quantiles are often those that go beyond the range of the data. 
For instance, in attribution studies, climate scientists investigate the causal effect of anthropogenic influences on climate extremes such as heavy precipitation. The quantiles of interest for such extreme events typically go far beyond the range of historical recordings and therefore extreme value extrapolation is required \citep[e.g.,][]{EASTERLING201617, van_Oldenborgh_2017}.
Formally this corresponds to the case of extreme (rather than intermediate or moderately extreme) $\tau_n$-QTE where $n \tau_n \to 0$. 


From now on, we use the notation $p_n$ to denote levels where $p_n \to 0$ and $n p_n \to d \ge 0$, which includes the extreme case, and we use $\tau_n$ to denote only the intermediate levels where $\tau_n \to 0$ and $n \tau_n \to \infty$.
To the best of our knowledge, there is no existing method for the estimation and inference of the $p_n$-QTE where $n p_n \to 0$.
In particular, since the effective sample size below the $p_n$-quantile tends to zero, the estimators $\wh q_j(p_n)$ based on empirical quantile loss are no longer applicable.



In this paper, we focus on heavy-tailed distributions, which have polynomially decaying tail probabilities and are thus heavier than Gaussian. Heavy tails are often encountered in risk analysis applications and many works therefore study this distribution class \citep[e.g.,][]{mat2004,wang2012estimation, athey2021semiparametric, xu2022}.
We propose a new quantile estimator $\wh Q_j(p_n)$ for $q_j(p_n)$ based on parametric tail approximations from extreme value theory \citep[][]{de2007extreme}, which enables us to extrapolate from intermediate to extreme quantiles. 
Indeed, based on the theory of regular variation, we have
\begin{equation}\label{approx}
q_j(p_n) \approx q_j(\tau_n) \left( {\tau_n}/{p_n} \right )^{\gamma_j},\quad j=0,1,
\end{equation}
where $\gamma_j > 0$ is the extreme value index of the potential outcome $Y(j)$.
Figure~\ref{fig:prediction} in Section~\ref{sec:EVT} illustrates the advantage of using extrapolation as in~\eqref{approx} for estimation of extreme quantiles compared to empirical estimates.

Our proposed estimator $\wh Q_j(p_n)$ (see \eqref{ExtremeQuantileEstimator} in Section~\ref{sec:method_asymptotics}) is obtained by plugging in
the estimator $\wh q_j(\tau_n)$ (see \eqref{eq:firpo_estimator}) from \cite{firpo2007efficient} for intermediate quantiles
and a newly proposed causal Hill estimator $\wh \gamma_j^H$ (see \eqref{hill_est}) for the extreme value index based on inverse propensity score weighting. We use the Hill type instead of the Pickands type estimator for the extreme value index because the latter is known to suffer from high variance in the heavy-tailed case,
but we would also like to note that the Hill type estimator is not invariant to location shift while the Pickands estimator is.
The final $p_n$-QTE estimator is $\wh\delta(p_n) = \wh Q_1(p_n) - \wh Q_0(p_n)$,
which we call the extremal QTE estimator.
Beyond point estimation, 
we establish the asymptotic normality of this estimator.
In particular, inspired by \cite{zhang2018extremal} and \cite{chernozhukov2011inference}, we propose a new normalizing factor for the extrapolation quantile estimators of two treatment groups, to deal with the problem that they may have different convergence rates.

The asymptotic variance of the extremal QTE estimator is unknown, and we propose a technically tractable variance estimator.
We prove that this estimator is consistent under an additional assumption (Assumption~\ref{ass:squared_hill_tozero}). Even when this assumption does not hold, this estimator is conservative, in the sense that it is still consistent to some quantity that is larger than the true variance. Thus it can be used to construct asymptotically honest confidence intervals for the extremal QTE.

The extremal QTE estimator can be used for estimation and inference of 
moderately extreme and extreme quantiles for heavy-tailed distributions. It thus provides an alternative to \cite{zhang2018extremal} for moderately extreme QTEs, and a first method for extreme QTEs.

Our approach requires additional assumptions when compared to \cite{zhang2018extremal},
including most importantly the second-order regular variation, which is fairly standard in extreme value theory \citep[e.g.,][]{de2007extreme}.
It is the price we need to pay in order to go to more extreme quantiles than \cite{zhang2018extremal}.

The topic of causality for extreme events is receiving increasing interest. 
The line of work by \cite{gis2018, gis2018a, gnecco2021causal} and \cite{mha2020} define structural causal models and investigate causal relationships in the setting where several variables are simultaneously extreme, and they focus on learning the unknown causal structure. 
In climate science, there is a large body of literature on attribution studies where the effect of climate change on weather extremes is analyzed \citep[e.g.,][]{han2016, EASTERLING201617, van_Oldenborgh_2017, nav2018, naveau2020statistical}. 
These methods focus on model-based data where interventions on, say, carbon dioxide emissions, are possible, and no adjustment for confounding is required. 
\cite{jan2022} propose a method to quantify the causal effect of London cycling superhighways on extreme
traffic congestion, but no theoretical analysis of this method is done.
Our method adds to this growing literature and provides a theoretically justified approach for estimation and inference of extremal QTEs in the presence of confounding variables.

The paper is structured as follows. 
In Section~\ref{sec:preliminaries}, we review some key concepts from extreme value theory and
the $\tau$-QTE estimator.
In Section~\ref{sec:method_asymptotics}, we propose the causal Hill extreme value index estimator and the extremal QTE estimator, and show their asymptotic normality.
Furthermore, we propose a variance estimator and prove that it is consistent under an additional assumption, and that it is conservative otherwise.
In Section~\ref{sec:simulations}, we present the finite sample behavior of our proposed extremal QTE estimator in different simulation settings, and compare it to existing methods.
We apply our methodology to a real data set about college education and wage in Section~\ref{sec:application}. All proofs, more technical details and additional simulations can be found in the Supplementary Material. Any equations, theorems, etc.,~from the Supplementary Material are referred to starting with a letter A--F corresponding to the respective section in that document.

To be in line with the literature on  extreme value theory, we focus on extremal QTEs in the upper tail, that is, $\delta(1-p_n)$. Results for the lower tail can be derived similarly.

\section{Preliminaries} \label{sec:preliminaries}

\subsection{Extreme Value Theory} \label{sec:EVT}

We are interested in extreme quantiles with level $p_n$ where $n p_n \to d \geq 0$. Empirical estimates of extreme quantiles with $d=0$ become highly biased and classical asymptotic theory no longer applies (see Figure~\ref{fig:prediction}). Extreme value theory studies methods for quantile extrapolation that result in more accurate estimates of extremal quantiles \citep[e.g.,][]{de2007extreme}. This theory relies on the following mild assumption on a distribution that guarantees that the tail can be well approximated in a parametric way. 
Let the random variable $Y$ have distribution $F$ and quantile function $q(\cdot) = F^\leftarrow(\cdot)$,
where $f^{\leftarrow}(x) := \inf\{ y \in \R: f(y) \ge x \}$ is the left-continuous inverse of a non-decreasing function $f$.


\begin{definition}(cf. \cite{de2007extreme})
	The distribution $F$ is in the max-domain of attraction of a generalized extreme value distribution if
	there exist $\gamma \in \R$ and
	sequences of constants $a_n >0$ and $ b_n \in \R $,
	$n =1,2, \ldots$,
	such that
	\begin{equation}\label{doa_cond}
	\lim_{n \to \infty} F^n(a_n x + b_n) = \exp(-(1+\gamma x)^{-1/\gamma})
	\end{equation}
	for all $x$ such that $1+\gamma x > 0$.
	For $\gamma=0$, the right hand side is interpreted as $\exp(-e^{-x})$.
	The parameter $\gamma$ is called the extreme value index (EVI).
	\label{def:domain_of_attraction}
\end{definition}
This condition is mild as it is satisfied for most standard distributions, for example, the normal, Student-$t$ and beta distributions. For a complete characterization of the max-domain of attraction of the three regimes ($\gamma<0$, $\gamma=0$, $\gamma > 0$), we refer to \cite{res2008} and \cite{emb1997}.
We focus on the heavy-tailed case where $\gamma > 0$,
that is, distributions $F$ with regularly varying tails $1 - F(x) = L(x) x^{-1/\gamma}$, where $L$ is a slowly varying function at $\infty$.
The regular variation of the tail of $F$ can be equivalently expressed in
terms of the tail quantile function $U := \left( 1/(1-F) \right)^{\leftarrow}$. The max-domain of attraction condition~\eqref{doa_cond} for $\gamma > 0$ is equivalent to 
\begin{equation}
\quad \lim_{t \to \infty} \frac{U(tx)}{U(t)} = x^\gamma, \quad \forall x > 0.
\label{eq:first_order_u}
\end{equation}
Relation~\eqref{eq:first_order_u} and the fact that the quantile function $q = F^\leftarrow$ can be expressed as
$q(\tau) = U( 1/(1- \tau) )$ for $\tau \in (0,1)$, imply that for large enough $n$,
\begin{equation} \label{ExtrapolationApprox}
q(1-p_n) \approx q(1-\tau_n) \left( {\tau_n}/{p_n} \right )^\gamma,
\end{equation}
where $\tau_n$ is a sequence such that $\tau_n \to 0$ and $\tau_n>p_n $. 
By using an intermediate sequence $\tau_n$, 
relation~\eqref{ExtrapolationApprox} allows us to extrapolate from intermediate to extreme quantiles.
We illustrate the benefit of using the extreme value extrapolation~\eqref{ExtrapolationApprox}
in Figure~\ref{fig:prediction}. It can be seen that the empirical estimates are strongly biased for extreme quantiles when the effective sample size $np_n< 1$, while the extrapolation allows for approximately unbiased estimates.
\begin{figure}[h!]
	\centering
	\includegraphics[scale=.6]{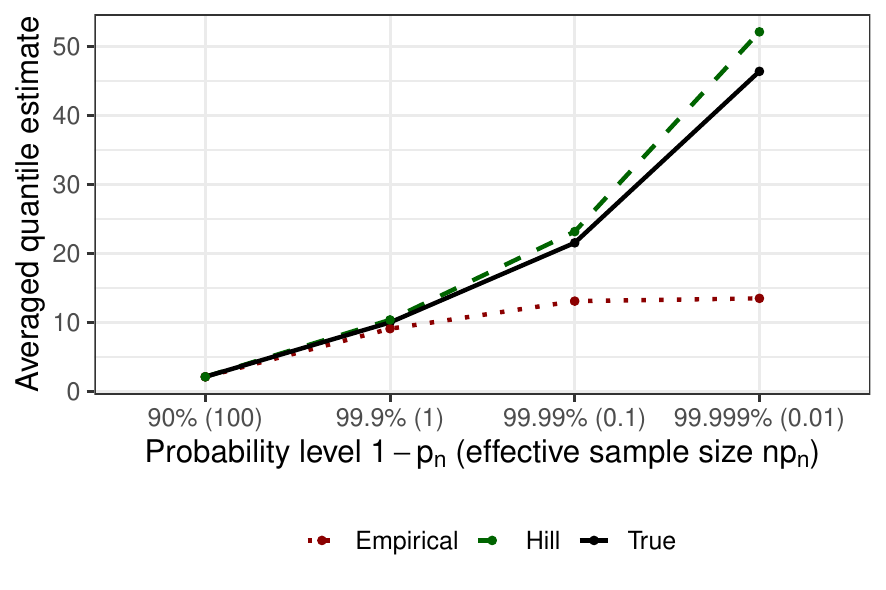}
	\caption[]
	{Averaged quantile estimates (over 1000 repetitions) at different probability levels $1-p_n$ fitted on $n=1000$ i.i.d.\ samples
	from a Fr\'echet distribution with shape $3$ (i.e., $\gamma = 1/3$), location $0$ and scale $1$. 
	The corresponding effective sample size $n p_n$ is shown with parentheses. 
	The solid black line denotes the true quantiles,
	the dotted red line denotes the empirical estimator, 
	and the dotted green line denotes the extrapolation method~\eqref{ExtrapolationApprox}, where the intermediate quantile level is $1-\tau_n = 90\%$
	and the EVI $\gamma$ is estimated by the Hill estimator \citep{hill1975simple}.} 
	\label{fig:prediction}
\end{figure}

To analyze the asymptotic distribution of the
extrapolated quantiles, usually a second-order condition is assumed.
\begin{definition}(cf. \cite{de2007extreme}) \label{def:second-order regular variation}
	The function $U$ is of second-order regular variation
	with 
	first-order parameter \mbox{$\gamma > 0$}
	and second-order parameter $\rho \le 0$
	if there exists
	a positive or negative function $A$
	with $\lim_{t \to \infty} A(t) = 0$
	such that for all $x > 0$
	\[
	\lim_{t \to \infty} \frac{U(tx)/U(t) - x^\gamma}{A(t)}
	= x^\gamma \frac{x^\rho - 1}{\rho}.
	\]
	The function $A$ is called the second-order auxiliary function. For $\rho=0$, $\frac{x^\rho - 1}{\rho}$ is defined as $ \lim_{\rho \to 0} \frac{x^\rho - 1}{\rho} = \log(x)$.
	\label{def:second_order}
\end{definition}
We note that Definition~\ref{def:second-order regular variation}
is stronger than Definition~\ref{def:domain_of_attraction}:
if $U$ is of second-order regular variation with first-order parameter $\gamma > 0$, then $F$ 
satisfies the max-domain of attraction condition with extreme value \mbox{index $\gamma$}. Second-order regular variation is satisfied by many popular distributions, and it is often possible to derive the function $A$ and the value of $\rho$ explicitly; see \cite{alves2007note} for details and examples. 

\subsection{Estimation of the $\tau$-QTE}  \label{sec:estQTEFirpo}




The estimation of $\tau$-QTE from \cite{firpo2007efficient} relies on the propensity score $\Pi(x)$,
which is used to adjust for confounding in the binary treatment setting.
In particular, \citet[][Corollary 1]{firpo2007efficient} shows that the QTE is identifiable under the following assumptions.
\begin{assumption} \label{ass:1}
	\
	\begin{enumerate}
		\item[i)] $(Y(1), Y(0)) \ \indep \ D \ \mid \ X$.
		\item[ii)] $X$ has compact support $\supp(X)$ 
		and there exists $c>0$ such that $c < \Pi(x) < 1-c$ for all $x \in \supp(X)$.
	\end{enumerate}
\end{assumption}
Assumption~\ref{ass:1}~$i)$ is the unconfoundedness condition and Assumption~\ref{ass:1}~$ii)$ is called the common support assumption. 
Both assumptions are fairly standard in causal inference literature, and we impose them throughout this paper.


The propensity score $\Pi(x)$ is generally unknown in practice and needs to be estimated.
For $n$ independent copies $(Y_i, D_i, X_i)_{i=1}^n$ of $(Y, D, X)$, we follow \cite{hirano2003efficient}, \cite{firpo2007efficient} and \cite{zhang2018extremal} and use the nonparametric sieve method to obtain an estimate $\wh \Pi(x)$ (see Section~\ref{appendix:estPropensity} in the Supplementary Material for more details).
Based on inverse propensity score weighting,
\cite{firpo2007efficient} proposed 
the following estimators for the quantiles of the potential outcomes $Y(1)$ and $Y(0)$:
\begin{equation}
\begin{aligned}
\wh q _1(\tau) &:= \argmin_{q\in \R} \Sn \frac{D_i}{\wh \Pi (X_i) }
(Y_i - q) (\tau - \Ind{Y_i \le q}),  
\\
\wh q_0(\tau) &:= \argmin_{q\in \R} \Sn \frac{1-D_i}{1-\wh \Pi (X_i) }
(Y_i - q) (\tau - \Ind{Y_i \le q}).
\label{eq:firpo_estimator} 
\end{aligned}
\end{equation}
The $\tau$-QTE is then estimated by $\wh q _1(\tau) - \wh q _0(\tau)$.

Under Assumption~\ref{ass:tail} below and the two regularity Assumptions~\ref{ass:sieve} and \ref{ass:zhangTechnical} in Section~\ref{appendix:sieve} of the Supplementary Material,
\citet[Theorem 3.1][]{zhang2018extremal} showed 
that for the intermediate quantile index $\tau_n$ (i.e., $\tau_n \to 0$, $n \tau_n \to \infty$),
the $\tau_n$-QTE estimator $\wh q_1(1-\tau_n) - \wh q_0 (1-\tau_n)$ is asymptotic normal.
For the moderately extreme case (i.e., $\tau_n \to 0$ and $n \tau_n \to d > 0$), however, \cite{zhang2018extremal} showed that the limiting distribution is no longer Gaussian.

\begin{assumption}\label{ass:tail}(Regularity conditions on the potential outcome distributions)
	~\\
	For $j=0,1$:
	\begin{enumerate}
		\item[i)] $Y(j)$ and $Y(j) \mid X$ are continuously distributed with densities $f_j$ and $f_{j \mid X}$, respectively.
		\item[ii)] The density $f_j$ of $Y(j)$ is  monotone in its upper tail.
		\item[iii)] The distribution function $F_j$ of $Y(j)$ belongs to the max-domain of attraction of a generalized extreme value distribution
		with extreme value index $\gamma_j$ (see Definition~\ref{def:domain_of_attraction}). 
	\end{enumerate}
\end{assumption}

\section{Extremal Quantile Treatment Effect Estimation for Heavy-Tailed Models}
\label{sec:method_asymptotics}


We now go beyond the intermediate and moderately extreme cases and propose an extremal QTE estimator based on quantile extrapolation,
that can be used in moderately
extreme and extreme cases (i.e., $p_n \to 0$, $n p_n \to d \geq 0$). We also derive its asymptotic normality and propose an asymptotic variance estimator which enables us to construct a confidence interval for the extremal QTE.

\subsection{Extremal QTE Estimator}

Our extremal QTE estimator is based on the quantile extrapolation approach (see approximation \eqref{ExtrapolationApprox}),
so estimators for the EVIs of the potential outcome distributions are required.
In the classical setting, the Hill estimator introduced by \cite{hill1975simple} is a common choice for heavy-tailed cases.
In the potential outcome setting, however, the classical Hill estimator is not appropriate due to confounding.
Therefore, we propose the following causal Hill estimators that adjust for confounding via 
the estimated propensity score $\wh{\Pi}$ (see \eqref{SieveEstimatorPropensityScore}):
\begin{equation} 
\begin{aligned}
\wh \gamma_1^H &:= \frac{1}{n\tau_n} \Sn (\log(Y_i) - \log(\wh q_1(1-\tau_n)))
\frac{D_i}{\wh \Pi(X_i)} \Ind{Y_i > \wh q_1(1-\tau_n)}, \\ 
\wh \gamma_0^H &:= \frac{1}{n\tau_n} \Sn (\log(Y_i) - \log(\wh q_0(1-\tau_n))) 
\frac{1-D_i}{1-\wh\Pi(X_i)} \Ind{Y_i > \wh q_0(1-\tau_n)}.
\label{hill_est} 
\end{aligned}
\end{equation}

\cite{zhang2018extremal} also proposed two EVI estimators for short-tailed distributions in his supplementary material,
one is the Pickands type estimator and one is the Hill type estimator. 
Rather than extrapolation, his goal of using the EVI estimator is to estimate the $0$-th QTE, that is, the lower endpoint of the distribution. 
His Hill type estimator is moment-based, which also uses inverse propensity score weighting, but with the true propensity score.
In the simulations in Section~\ref{sec:simulations}, we implement the quantile extrapolation approach with his Pickands type estimator and compare it to our proposed estimator~\eqref{hill_est}.

Building on the causal Hill estimators $\wh \gamma_j^H$ in~\eqref{hill_est} and the intermediate quantile estimator $\wh q_j(1-\tau_n)$ in~\eqref{eq:firpo_estimator}, 
we propose the following quantile extrapolation estimator
\begin{align}\label{ExtremeQuantileEstimator}
\wh{Q}_j(1-p_n) := 
\wh q_j(1-\tau_n) \left( \frac{\tau_n}{p_n} \right)^{\wh \gamma_j^H},
\end{align}
for $j=0,1$, and the final extremal QTE estimator is defined as the difference of the extrapolated quantiles:
\begin{align} \label{ExtremeQTEestimator}
\wh{\delta}(1-p_n) := \wh{Q}_1(1-p_n) - \wh{Q}_0(1-p_n).
\end{align}
For simplicity, we use the same intermediate level $\tau_n$ for both extrapolation estimators $\wh Q_1(1-p_n)$ and $\wh Q_0(1-p_n)$ in this paper,
but in principle, one can use different intermediate levels for each potential outcome.
The latter might be advantageous if the potential outcome distributions have very different tail behaviors, 
or if there is a severe imbalance between treated and non-treated samples.

To obtain the extremal QTE estimator \eqref{ExtremeQTEestimator}, 
the intermediate level $\tau_n$ (or $k$ if we consider $\tau_n = k/n$) needs to be chosen.
The optimal choice of the $\tau_n$ depends on the underlying data distribution and is difficult in practice,
and there is a bias-variance trade-off. 
Specifically,
if $\tau_n$ is too small, then the effective sample size $n \tau_n$ is small
and this will lead to high variance.
On the other hand, if $\tau_n$ is too large, then the assumptions of extreme value theory may not hold because we are no longer in the tail of the distribution, and this will lead to high bias.
In practice, there are some commonly used approaches for choosing $\tau_n$.
The simplest one is to set $\tau_n$ to some reasonable fixed value based on the background knowledge of the concrete problem.
One can also plot the estimates depending on different $\tau_n$ and then select $\tau_n$ in the first stable region of the plot \citep[e.g.,][]{resnick2007heavy}.
There are also adaptive methods to approximate the optimal $\tau_n$ \citep[e.g.,][]{drees1998selecting, boucheron2015tail}.

\subsection{Asymptotic Properties}

The main theoretical result in this subsection is Theorem~\ref{thm:hill_extreme_qte}, which shows the asymptotic normality of the extremal QTE estimator $\wh{\delta}(1-p_n)$.
The major steps to prove this theorem are the following: (1) showing that the asymptotic behavior of the quantile extrapolation estimator $\wh{Q}_j(1-p_n)$ depends only on the asymptotic distribution of the causal Hill estimator $\wh \gamma_j^H$ (see Lemma~\ref{lemma:hill_extreme_quantiles}); 
(2) deriving the asymptotic distribution of $\wh \gamma_j^H$ (see Theorem~\ref{thm:hill_normality}), which builds on an asymptotic linearity result (see Theorem~\ref{thm:hill});
(3) introducing a new normalizing factor $\wh {\beta}_n$ (see formula~\eqref{NormalizingFactorQTE}) when deriving the asymptotic distribution of $\wh{\delta}(1-p_n)$ to account for the issue that $\wh{Q}_1(1-p_n)$ and $\wh{Q}_0(1-p_n)$ may have different convergence rates.


\subsubsection{Asymptotic Properties of the Causal Hill Estimator} \label{sec:hill}
We first present a result showing that under the same conditions as the Theorem 3.1 of \cite{zhang2018extremal}, our proposed causal Hill estimator is consistent.
\begin{lemma}
	\label{lemma:hill_consistency}
	Suppose that Assumptions~\ref{ass:1}, \ref{ass:tail},
	\ref{ass:sieve} and \ref{ass:zhangTechnical} hold, 
	and $n \tau_n  \to \infty$ and $\tau_n \to 0$.
	If for $j=0,1$, the extreme value index $\gamma_j>0$,
	then 
	\[
	\wh \gamma_j^H \tendsto{P} \gamma_j.
	\]
      \end{lemma}
      
To obtain asymptotic normality of the causal Hill estimator, we require Assumption~\ref{ass:second_order_regular_var} below,
which is a  second-order regular variation assumption.
This assumption is standard to obtain asymptotic normality results for heavy-tailed distributions, and it is satisfied by most heavy-tailed distributions such as the Student-$t$ and the Fr\'echet distribution \citep[e.g.,][]{de2007extreme}.
\begin{assumption}\label{ass:second_order_regular_var} 
	~\\
	For $j=0,1$,
	the tail function $U_j = (1/(1-F_j))^\leftarrow$ of $Y(j)$ is of 
	second-order regular variation (see Definition~\ref{def:second_order}) with extreme value index $\gamma_j>0$,
	second-order auxiliary function $A_j$ and 
	second-order parameter $\rho_j \leq 0$. 
\end{assumption}

To guarantee that the estimated propensity score is compatible with the causal Hill estimator, 
we also require Assumption~\ref{ass:hill_sieve} below, which is a regularity assumption on the sieve estimator. 
This assumption is of similar type as Assumption~\ref{ass:sieve} which is used in the Theorem 3.1 of \cite{zhang2018extremal}.
\begin{assumption} \label{ass:hill_sieve}
	~\\
	For $j=0,1$,
	$\ERW{ 
		\log \left( \frac{\tau_n}{1-F_j(Y(j))} \right)  \Ind{ Y(j) > q_j(1-\tau_n)} \mid X = x}$
	is $t$-times continuously differentiable in $x$ with all derivatives
	bounded by some $N_n$ uniformly over $\supp(X)$.
	Here $\sqrt{\frac{n}{\tau_n}} N_n  h_n^{- t/2r} \to 0$ 
	as $n\to \infty$, where $h_n$ is the number of sieve bases 
	and $r$ is the dimension of $X$.
\end{assumption}

Before showing asymptotic normality of the causal Hill estimator,
we show an important asymptotic linearity result under the above two extra assumptions.
\begin{theorem}
	Suppose that Assumptions~\ref{ass:1} -- \ref{ass:hill_sieve},
	\ref{ass:sieve} and \ref{ass:zhangTechnical} hold, 
	and $k = n \tau_n \to \infty$ and $\tau_n=k/n \to 0$.
	If for $j=0,1$, $\sqrt{k} A_j(n/k) \to \lambda_j \in \R$ and the extreme value index $\gamma_j>0$,
	then
	\[
	\sqrt{k} (\wh \gamma^H_j - \gamma_j) = \frac{\lambda_j}{1-\rho_j} + \frac{1}{\sqrt{n}} \Sn \left( \psi_{i,j,n} -\gamma_j\phi_{i,j,n} \right) + o_p(1),
	\]
	where
	\begin{align*}
	\psi_{i,1,n} &:= \frac{1}{\sqrt{\tau_n}}  \left(
	\frac{ D_i}{\Pi(X_i)} S_{i,1,n}  
	- \gamma_1 \tau_n   -
	\frac{\ERW{S_{i,1,n} \mid  X_i}}{\Pi(X_i)} (D_i - \Pi(X_i))
	\right), \\
	\psi_{i,0,n} &:= \frac{1}{\sqrt{\tau_n}}  \left(
	\frac{ 1-D_i}{1-\Pi(X_i)} S_{i,0,n}  
	-  \gamma_0 \tau_n   +
	\frac{\ERW{S_{i,0,n} \mid  X_i}}{1-\Pi(X_i)} (D_i - \Pi(X_i))
	\right),
	\\
	\phi_{i,1,n} &:= \frac{1}{\sqrt{ \tau_n}} \left( \frac{D_i}{\Pi(X_i)} T_{i,1,n} 
	- \frac{\ERW{T_{i,1,n}\mid X_i}}{\Pi(X_i)} (D_i - \Pi(X_i)) \right),
	\\
	\phi_{i,0,n} &:= \frac{1}{\sqrt{\tau_n}} \left( \frac{1-D_i}{1-\Pi(X_i)} T_{i,0,n} 
	+ \frac{\ERW{T_{i,0,n}\mid X_i}}{1-\Pi(X_i)} (D_i - \Pi(X_i)) \right)
	\end{align*}
	with  
	$T_{i,j,n} :=  \Ind{Y_i(j) > q_j(1-\tau_n)}- \tau_n$
	and 
	$S_{i,j,n} := 
	\gamma_j \log \left( \frac{\tau_n}{1-F_j(Y_i(j))} \right)  \Ind{ Y_i(j) > q_j(1-\tau_n)}$.
	\label{thm:hill}
\end{theorem}

\begin{remark} \label{rmk:hill_information_negligible}
	$\phi_{i,j,n}$ is the influence function for
	the intermediate quantile estimator $\wh q_j(1-\tau_n)$ which also appears in \cite{zhang2018extremal} (see also Theorem~\ref{thm:interquantile}), and
	$\psi_{i,j,n}$ is the new influence function appearing in our result, which has a similar form as $\phi_{i,j,n}$.
	The factors $\ERW{S_{i,j,n}\mid X_i}$ and $\ERW{T_{i,j,n}\mid X_i}$ in the influence functions
	correspond to the information gain from the nonparametric estimation of the propensity score.
	\cite{zhang2018extremal} made the observation that
	since $P(Y_i(j) > q_j(1-\tau_n) \mid X_i)$ is of order $O_p(\tau_n)$,
	the term with factor $\ERW{T_{i,j,n}\mid X_i}$ in $\phi_{i,j,n}$ is negligible under
	suitable integrability conditions.
	The same holds for the information gain term in $\psi_{i,j,n}$ 
	because $\ERW{S_{i,j,n}\mid X_i}$ is also of order $O_p(\tau_n)$ (see Lemma~\ref{lemma:s_order}).
	We discuss this in more detail when considering variance estimation in Section~\ref{sec:variance}.
\end{remark}

Given the asymptotic linearity result in Theorem~\ref{thm:hill} and the fact that the influence functions depend on the sample size $n$, 
we will use a triangular array central limit theorem to obtain asymptotic normality.
This requires that the covariance matrix $\Sigma_n$ of the random vectors
$(\psi_{i,1,n}, \psi_{i,0,n}, \phi_{i,1,n}, \phi_{i,0,n})$ converge (see Assumptions~\ref{ass:zhangTechnical} and \ref{ass:hill_Technical}).
The extra Assumption~\ref{ass:hill_Technical} is of similar flavor as Assumption~\ref{ass:zhangTechnical} used in Theorem 3.1 of \cite{zhang2018extremal}.
In fact, we can show that the sequence $\Sigma_n$ is bounded by using similar arguments as in the proof of Theorem~\ref{thm:hill_normality}.
Its convergence is therefore mostly a technical condition.



We now give the asymptotic normality result of the causal Hill estimator.
\begin{theorem}
	Suppose that Assumptions~\ref{ass:1} -- \ref{ass:hill_sieve}
	and \ref{ass:sieve} -- \ref{ass:hill_Technical} hold,
	and $k = n \tau_n \to \infty$ and $k/n \to 0$.
	If for $j=0,1$, $\sqrt{k} A_j(n/k) \to \lambda_j \in \R$ and
	the extreme value index $\gamma_j>0$,
	then 
	\[
	\sqrt{k} (\wh \gamma_1^H - \gamma_1, \wh \gamma_0^H - \gamma_0)
	\tendsto{D} \N( \mu_{\gamma}, B \Sigma B^T), 
	\]
	where $\mu_{\gamma} = \left(\frac{\lambda_1}{1-\rho_1}, \frac{\lambda_0}{1-\rho_0} \right)^T$,
	$B =  \left(
	\begin{matrix}                                
	1 & 0 & -\gamma_1 & 0 \\
	0 & 1 & 0 & -\gamma_0 \\
	\end{matrix}
	\right)$
	and $\Sigma$ defined as in \eqref{covariance_true}. 
	\label{thm:hill_normality}
\end{theorem}


\subsubsection{Asymptotic Properties of the Extremal QTE Estimator}

We now study the asymptotic properties of the quantile extrapolation estimator $\wh Q_j(1-p_n)$ in~\eqref{ExtremeQuantileEstimator}. We first give the following lemma which shows that the asymptotic behavior of $\wh Q_j(1-p_n)$ only depends on the 
asymptotic distribution of the EVI estimator.
\begin{lemma}
    Suppose that Assumptions~\ref{ass:1} -- \ref{ass:hill_sieve}
    and \ref{ass:sieve} -- \ref{ass:hill_Technical} hold,
    and $k = n \tau_n \to \infty$, $k/n \to 0$, $n p_n = o(k)$, and $\log(n p_n) = o(\sqrt{k})$.
    If for $j=0,1$, $\sqrt{k} A_j(n/k) \to \lambda_j \in \R$ and
    the extreme value index $\gamma_j>0$,
    then 
    \[
     \frac{\sqrt{k}}{\log(\tau_n/p_n)} \left( \frac{\wh{Q}_j(1-p_n) }{ q_j(1-p_n)} -1  \right)
     = \sqrt{k} (\wh \gamma_j^H - \gamma_j) + o_p(1).
    \]
    In particular, the above implies that $ \wh{Q}_j(1-p_n)- q_j(1-p_n)  \tendsto{P}=0 $ for $j=0,1$.
    \label{lemma:hill_extreme_quantiles}
\end{lemma}


Lemma~\ref{lemma:hill_extreme_quantiles} 
(and the main result Theorem~\ref{thm:hill_extreme_qte})
allows that $n p_n \to 0$, but it cannot converge to zero arbitrarily fast as $\log(n p_n) = o(\sqrt{k})$. 
This is reasonable since it means that there are limitations on how far the extrapolation can be pushed.
The other rate condition $n p_n = o(k)$ is also natural since we are interested in the case where $p_n$ converges to $0$ much faster than $\tau_n$.
These rate conditions are standard \citep[see, e.g.,][]{de2007extreme}.

We already know from Theorem~\ref{thm:hill} that the causal Hill estimator is asymptotically normal.
Thus, to show asymptotic normality of the extremal QTE estimator $\wh{\delta}(1-p_n)$ in \eqref{ExtremeQTEestimator},
the only remaining difficulty is that $\wh{Q}_1(1-p_n)$ and $\wh{Q}_0(1-p_n)$ can have
different normalizing factors and convergence rates. 
This is problematic if the ratio of the normalizing factors oscillates.
\cite{zhang2018extremal} encountered the same issue, 
and he followed the idea from 
\cite{chernozhukov2011inference}
to construct a feasible normalizing factor under the assumption that the ratio of normalizing factors converges.
We proceed similarly. Specifically, based on Lemma \ref{lemma:hill_extreme_quantiles}, we introduce the following normalizing factor
\begin{align}
    \label{NormalizingFactorQTE}
	\wh{\beta}_{n} := \frac{\sqrt{k}}{\log(\tau_n/p_n) \max\{ \wh Q_1(1-p_n), \wh Q_0(1-p_n) \}},
\end{align}
and make the following assumption.
\begin{assumption} \label{ass:normal_conv} 
	
	\[
	\frac{q_1(1-\tau)}{q_0(1-\tau)}
	\to \kappa \in [0, +\infty]
	\text{ as } \tau \to 0.
	\]
\end{assumption}
Assumption~\ref{ass:normal_conv} states that the tails of the potential outcome distributions are either comparable or that one of them is heavier than the other. This is a fairly standard assumption satisfied by many models. 

Using the normalizing factor \eqref{NormalizingFactorQTE}, we can show asymptotic normality of the 
extremal QTE estimator.
\begin{theorem} \label{thm:hill_extreme_qte}
	Suppose that Assumptions~\ref{ass:1} -- \ref{ass:normal_conv}
	and \ref{ass:sieve} -- \ref{ass:hill_Technical} hold,
	and $k = n \tau_n \to \infty$, $k/n \to 0$,  $n p_n = o(k)$ and $\log(n p_n) = o(\sqrt{k})$.
	If for $j=0,1$, $\sqrt{k} A_j(n/k) \to \lambda_j \in \R$ and
	the extreme value index $\gamma_j>0$, then
	\[
	\wh{\beta}_n \left( \wh{\delta}(1-p_n) - \delta(1-p_n) \right)
	\tendsto{D} \N( \mu, \sigma^2), 
	\]
	where $\mu = v_{\kappa}^T w_{\lambda, \rho}$ and $\sigma^2 = v_{\kappa}^T B \Sigma B^T v_{\kappa}$ 
	with $B$ and $\Sigma$ defined as in Theorem~\ref{thm:hill_normality}, and
	$$ v_\kappa = \left(
	\begin{matrix}
	\min\{1, \kappa\} \\
	-\min\{1, {1}/{\kappa} \} \\
	\end{matrix}
	\right), \qquad w_{\lambda, \rho} = \left(
	\begin{matrix}
	\lambda_1/(1-\rho_1) \\
	\lambda_0/(1-\rho_0)  \\
	\end{matrix}
      \right).
      $$
\end{theorem}

Due to the asymptotic bias of $\widehat\gamma_j^H$ (see Theorem~\ref{thm:hill_normality}), there is also an asymptotic bias of the extremal QTE estimator $\widehat{\delta}(1-p_n)$, which affects the validity of our later proposed confidence interval (see \eqref{ConfidenceInterval}). This asymptotic bias equals $0$ if $\sqrt{k} A_j(n/k) \to 0$. Recall that $\lim_{t \to \infty} A_j(t) = 0$ by the second-order regular variation assumption (see Assumption~\ref{ass:second_order_regular_var}), so $\sqrt{k} A_j(n/k) \to 0$ holds if $k$ grows not too fast. The rate at which $A_j(n/k)$ tends to zero is unknown, and we therefore advise being conservative in the sense that one should choose $k$ (or equivalently, $\tau_n$) rather small in practice to ensure that $\sqrt{k} A_j(n/k) \to 0$, and hence the asymptotic bias is negligible.

\subsection{Variance Estimation and Confidence Intervals} \label{sec:variance}
 
In order to conduct statistical inference based on Theorem~\ref{thm:hill_extreme_qte}, a consistent estimator of the asymptotic variance $\sigma^2$ is needed.
The main difficulty in estimating $\sigma^2$ lies in estimating the corresponding matrix $\Sigma$, 
which is the limit of the covariance matrices of the influence functions. 
Recall from Remark~\ref{rmk:hill_information_negligible} that these influence functions contain terms describing the information gain from nonparametric estimation of the propensity score. 
\cite{firpo2007efficient} encountered a similar issue and he proposed a nonparametric regression approach to estimate the contribution of the information gain to the variance.
In this paper, however, we will not go in this direction. Instead, we show that under suitable assumptions, the information gain for the proposed extremal QTE estimator is actually negligible, which can simplify the covariance matrix needed to be estimated, and thus a simpler and computational cheaper method can be proposed.
Specifically, we require the following assumption:
\begin{assumption} \label{ass:squared_hill_tozero} 
	~\\
	For $j=0,1$ 
	\begin{align*}
	\frac{1}{\tau_n} \ERW{ P(Y(j) > q_j(1-\tau_n) \mid  X)^2}  \to 0,
	\quad \text{and} \quad
	\frac{1}{\tau_n} \ERW{ \ERW{S_{j,n}\mid X}^2}  \to 0,
	\end{align*}
	where
	\[
	S_{j,n} := 
	\gamma_j \log\left(\frac{\tau_n}{1-F_j(Y(j))}\right) \Ind{Y(j) >  q_j(1-\tau_n)} .
	\]
\end{assumption}

Lemma~\ref{lemma:s_order} in the Supplementary Material shows that both $P(Y(j) > q_j(1-\tau_n) \mid  X)$ and
$\ERW{S_{j,n}\mid X}$ are of order $O_p(\tau_n)$.
Hence Assumption~\ref{ass:squared_hill_tozero} holds under suitable integrability conditions, and Section~\ref{appendix:example} in the Supplementary Material presents a concrete example where it is satisfied.
We propose the following variance estimator
\begin{align}\label{var_estimator}
	\wh \sigma^2 := \wh v_{\kappa}^T \wh B \wh \Sigma \wh B^T \wh v_{\kappa},
\end{align}
where 
\begin{equation*}
	\wh B :=  \left(
	\begin{matrix}                                
		1 & 0 & -\wh \gamma_1^H & 0 \\
		0 & 1 & 0 & -\wh \gamma_0^H \\
	\end{matrix}
	\right)
	\text{ and }\
	\wh v_{\kappa} := \left(
	\begin{matrix}
		\min\{1, \wh \kappa \} \\
		-\min\{1, \frac{1}{\wh \kappa} \}
	\end{matrix}
	\right)
	\text{ with }
	\
	\wh{\kappa} := 
	\frac{\wh Q_1(1-p_n)}{\wh Q_0(1-p_n)},
\end{equation*}
and $\wh \Sigma$ is defined in \eqref{covariance_est_ass6} and its entries are estimated using inverse propensity score weighting; see Section~\ref{appendix:variance} of the Supplementary Material for the for details.
The following result shows that this estimator is consistent.


\begin{theorem} \label{thm:hill_variance_estimation}
	Suppose that Assumptions~\ref{ass:1} -- \ref{ass:hill_Technical} hold,
	and $k = n \tau_n \to \infty$, $k/n \to 0$,  $n p_n = o(k)$ and $\log(n p_n) = o(\sqrt{k})$.
	If for $j=0,1$, $\sqrt{k} A_j(n/k) \to \lambda_j \in \R$ and
	the extreme value index $\gamma_j>0$,
	then
	\[
	\wh \sigma^2 \tendsto{P} \sigma^2, \qquad n\to\infty,
	\]
	where $\sigma^2$ is the asymptotic variance of the extremal QTE estimator in Theorem~\ref{thm:hill_extreme_qte} and $\wh \sigma^2$ is defined by~\eqref{var_estimator}.
\end{theorem}

Based on 
Theorem~\ref{thm:hill_extreme_qte} and the variance estimator \eqref{var_estimator},
we propose the following approximate $(1-\alpha)-$confidence interval of the extremal QTE:
\begin{align}\label{ConfidenceInterval}
\left[\wh \delta(1-p_n) \pm z_{(1-\alpha/2)} \frac{\wh \sigma}{\wh \beta_n} \right],
\end{align}
where $\wh \beta_n$ is the normalization constant in~\eqref{NormalizingFactorQTE}
and $z_{(1-\alpha/2)}$ is the $(1-\alpha/2)$-quantile of the standard normal distribution.

Since Assumption~\ref{ass:squared_hill_tozero} only affects the information gain terms in the covariance matrix $\Sigma$ and these terms reduce the variance,
even Assumption \ref{ass:squared_hill_tozero} does not hold, 
the variance estimator is conservative, in the sense that it is still consistent to some quantity $\wt\sigma^2$
(see \eqref{variance_conservative} and the proof of Theorem~\ref{thm:hill_variance_estimation}) that is
larger than the true variance $\sigma^2$; see Lemma~\ref{lemma:hill_conservative} below. Therefore, even if Assumption \ref{ass:squared_hill_tozero} does not hold, it is safe to use our estimator $\wh \sigma^2$.
\begin{lemma} \label{lemma:hill_conservative}
	Let $\sigma^2$ be the asymptotic variance of the extremal QTE estimator in Theorem~\ref{thm:hill_extreme_qte} and $\wt\sigma^2$ be defined in \eqref{variance_conservative},
	then $\wt\sigma^2 \geq \sigma^2$.
\end{lemma}

\section{Simulations} \label{sec:simulations}

We conducted simulations to examine the finite sample behavior of 
our proposed extremal QTE estimator \eqref{ExtremeQTEestimator} 
and the related confidence interval \eqref{ConfidenceInterval}, and to compare them to other methods. 
All simulations were carried out in R and the code is available at Github. 

\subsection{Simulation Set-up}
Throughout the simulation study, we consider univariate covariate $X$ as \cite{zhang2018extremal}.
Specifically, let $X$ and $U$ be uniformly distributed random variables on $[0,1]$ and assign the treatment by $D = \Ind{U \le \Pi(X)}$ with propensity score $\Pi(x) = 0.5x^2 + 0.25$. We generate the outcomes from the following three models:
\begin{align*}
\text{H}_1:
\begin{cases}
Y(1) & = 5 S \cdot (1+X) \\
Y(0) &= S \cdot (1+X) 
\end{cases}
\quad
\text{H}_2:
\begin{cases}
Y(1) & = C_2 \cdot  \exp(X) \\
Y(0) &= C_3 \cdot \exp(X) 
\end{cases}
\quad
\text{H}_3:
\begin{cases}
Y(1) & = P_{1.75+X, 2} \\ 
Y(0) &= P_{1.75+5X, 1} 
\end{cases}
\end{align*}
where $S$ follows a Student-t distribution with $3$ degrees of freedom,
$C_s$ is Fr\'echet distributed with shape parameter $s$, location $0$ and scale $1$,
and  $P_{a, b}$ is Pareto distributed with shape parameter $a$ and scale $b$.

The EVIs of the potential outcome distributions are $\gamma_1=\gamma_0 = 1/3$ for model $\text{H}_1$  and $\gamma_1 = 1/2$ and $\gamma_0 = 1/3$ for model $\text{H}_2$. 
For model $\text{H}_3$, a small calculation yields $\gamma_1 = \gamma_0 = 4/7$.
Models $\text{H}_2$ and $\text{H}_3$ are more heavy-tailed than $\text{H}_1$.

We consider data sets with sample size $n\in \{1000, 2000, 5000\}$ and aim to estimate the $(1-p_n)$-QTE with $p_n \in \{5/n, 1/n, 5/(n \log n)\}$.
Throughout, the target coverage for the confidence intervals is $90\%$ (i.e., $\alpha=0.1$).
For all bootstrap based methods, we use $1000$ bootstrapped data sets.
The empirical squared error and coverage are calculated based on $1000$ sampled data sets.

\subsection{Implemented Methods}
For point estimation of the extremal QTE, we compare the squared errors of three methods:
\begin{itemize}
	\item Firpo--Zhang estimator: the non-extrapolated, empirical QTE estimator $\wh q_1(1-p_n) - \wh q_0(1-p_n)$,
	where the quantile estimators are defined by~\eqref{eq:firpo_estimator}. It was proposed by \cite{firpo2007efficient} and further studied by \cite{zhang2018extremal}.
	\item Extrapolation with a causal Pickands estimator: 
	\begin{align} \label{extrapolation_pickands}
	\wh q_1(1-\tau_n) (\tau_n/p_n)^{\wh \gamma_1^P}- \wh q_0(1-\tau_n) (\tau_n/p_n)^{\wh \gamma_0^P},
	\end{align}
	where
	\[
	\wh{\gamma}_j^P = \frac{1}{\log(2) }
	\log \left( \frac{ \hat{q}_j(1-  \tau_n) - \hat{q}_j(1-2 \tau_n)}{ \hat{q}_j(1-2 \tau_n) -  \hat{q}_j(1-4\tau_n)} \right), \quad j=0,1.
	\]
	This estimator is based on quantile extrapolation with the causal Pickands EVI estimator $\wh{\gamma}_j^P$ proposed in the supplementary materials of \cite{zhang2018extremal}.
	\item Extremal QTE estimator (see~\eqref{ExtremeQTEestimator}): our proposed estimator based on quantile extrapolation with the causal Hill EVI estimator. 
\end{itemize}

For the confidence interval of the extremal QTE, we compare the empirical coverages of the following four methods:
\begin{itemize}
	\item Zhang: the $b$ out of $n$ bootstrap confidence interval proposed by \cite{zhang2018extremal} that builds on the Firpo--Zhang estimator. We use the ``with replacement" version as \cite{zhang2018extremal} suggested in his paper. 
	Its tuning parameters are described in Section~\ref{appendix:tuningZhang} of the Supplementary Material.
	\item BS Pickands: a bootstrap based method with the bootstrap confidence interval
	\begin{align}\label{simu:bootstrap-interval}
		\left[\wh \delta'(1-p_n) \pm z_{(1-\alpha/2)} \wh \sigma_* \right],
	\end{align}
	where $\wh \delta'(1-p_n)$ is the point estimate \eqref{extrapolation_pickands} of the extremal QTE based on the full sample and $\wh \sigma_*$ is the estimated standard deviation of this estimate via the non-parametric bootstrap.
	\item BS Hill: a non-parametric bootstrap based method as \eqref{simu:bootstrap-interval}, 
	but using  \eqref{ExtremeQTEestimator} to obtain the point estimate.
	\item Extremal QTE CI: our proposed confidence interval \eqref{ConfidenceInterval}.
\end{itemize}

For the intermediate quantile level $\tau_n= k/n$ of the extrapolation based methods, we use $k=n^{0.65}$. 
This value guarantees that all rate assumptions about $k, n$ and $p_n$ in Theorems~\ref{thm:hill_extreme_qte} and \ref{thm:hill_variance_estimation} are satisfied.
An additional sensitivity analysis can be found in Section~\ref{appendix:kdependency} of the Supplementary Material.
We note that $k=n^{0.65}$ may not be optimal in all settings, and we use it in our simulations mostly for convenience.
Choosing the optimal data-dependent $\tau_n$ is a difficult problem in extreme value theory.
In practice, we recommend choosing it by plotting the estimates using different $\tau_n$ and selecting the $\tau_n$ in the first stable region of the plot (e.g., Resnick 2007). Please also see the real data application in Section~\ref{sec:application} for an illustration of this approach.


For the size of sieve basis functions $h_n$ in the propensity score estimation, 
we use $h_n = \floor{2 n^{1/11}}$.
Note that this choice is only for the case of univariate $X$, and please see Section~\ref{appendix:sieve} of the Supplementary Material for some justifications about this choice.
Specifically, we have $h_{1000}=h_{2000} = 3$ and  $h_{5000}=4$.
In practice, people may choose the sieve basis functions according to their specific problem and use model selection methods such as cross-validation. Please see the real data application in Section~\ref{sec:application} for an illustration.






\subsection{Simulation Results}
The squared errors of the point estimates are shown in Figure~\ref{fig:squared_error}.
We see that our proposed extremal QTE estimator generally performs better than the other two methods. 
In particular, it exhibits the lowest mean squared error (MSE) over almost all settings.
This is especially true for the more heavy-tailed models $\text{H}_2$ and $\text{H}_3$, in which our method greatly outperforms the others.
The extrapolation based method using the Pickands EVI estimator has the worst performance,
which is not surprising as the Pickands estimator is known to suffer from high variance in heavy-tailed settings.
This also indicates that choosing a suitable EVI estimator is crucial for the extrapolation based method. 
\begin{figure}
	\centering
	\includegraphics[scale=0.34]{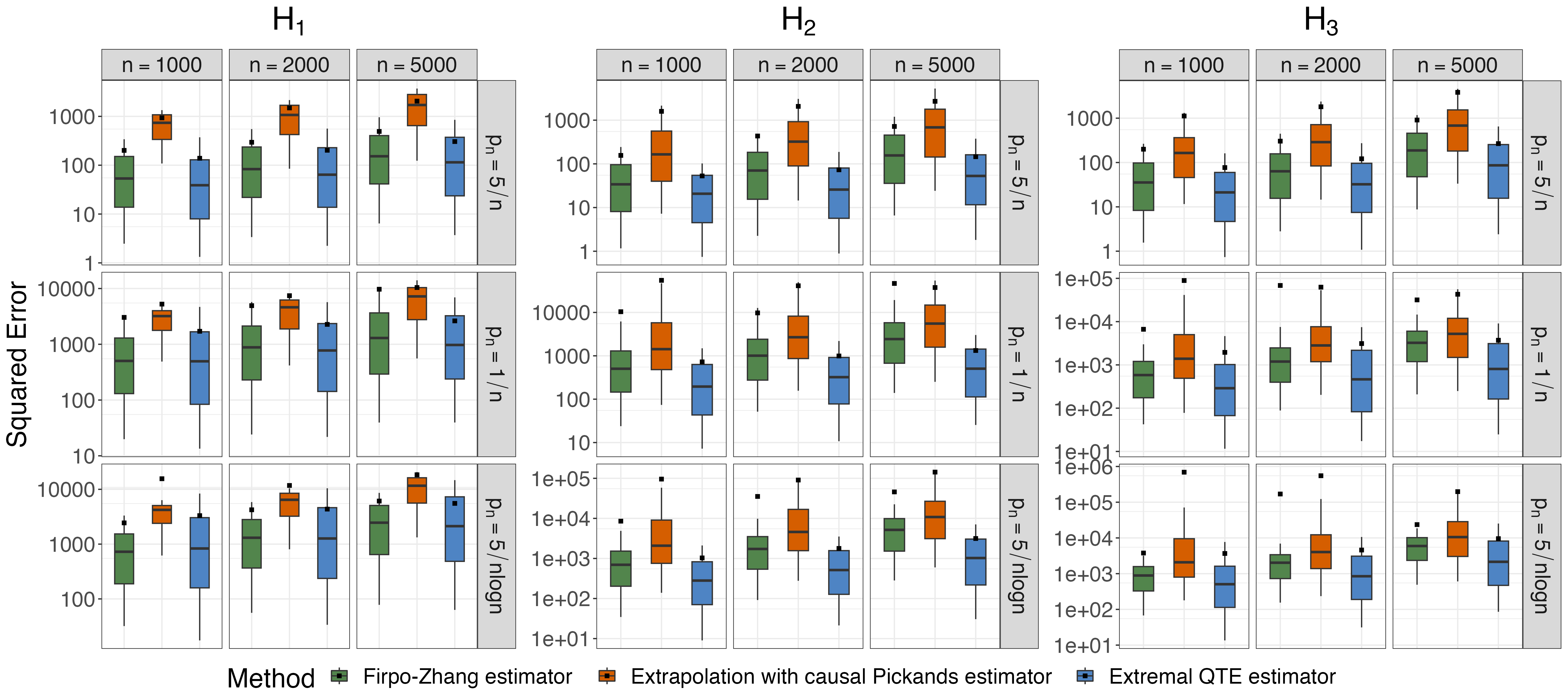}
	\caption[]
	{Box plots of the squared errors of different point estimators for the extremal QTE.
		The whiskers of the box plots correspond to the 0.1 and 0.9
		quantiles, the black horizontal line in the box corresponds to the median,
		and the square indicates the mean.
		Please note that the log-scale is used for the $y$-axis.} 
	\label{fig:squared_error}
\end{figure}

\begin{figure}
    \centering
    \includegraphics[scale=0.34]{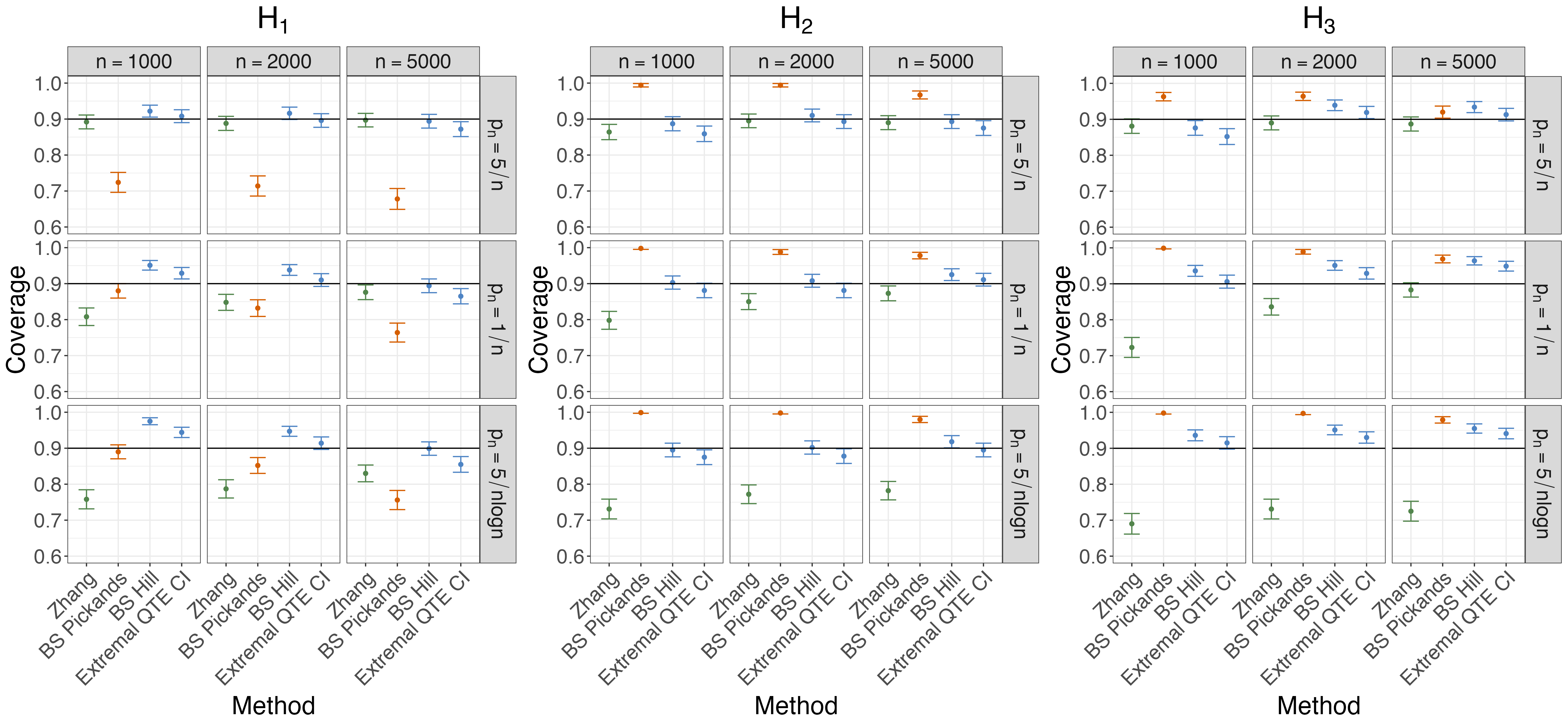}
    \caption[]
    {Coverage of different methods to construct confidence intervals. 
    	The target coverage is $90\%$ and is indicated by the solid horizontal black lines. 
    	The dots indicate the empirical coverage over $1000$ simulations,
    	and the error bars indicate an approximate normal based $95\%$-confidence interval for the true coverage over $1000$ simulations. 
    }
    \label{fig:coverage}
  \end{figure}

Figure~\ref{fig:coverage} compares the empirical coverage of the different confidence intervals.
We see that ``Zhang"
performs quite well for the not so extreme quantile level $p_n=5/n$, 
but that it can undercover largely when the quantile index becomes more extreme. 
Such results were expected as this method is designed for the moderately extreme case where $n p_n \to d > 0$, and not for the extreme case where $n p_n \to 0$.
The bias of the Firpo--Zhang estimator in the extreme case could be a reason for this undercoverage.
We see that ``BS Pickands" may suffer from both undercoverage (e.g., setting $p_n=5/n$ of $\text{H}_1$) and overcoverage (e.g., setting $p_n=5/n$ of $\text{H}_2$).
In comparison, ``BS Hill" and our proposed confidence interval ``Extremal QTE CI" perform better,
with empirical coverage close to the nominal level. 
The performance of ``BS Hill" shows that the bootstrap based method may be valid, and it would be interesting to formalize this in future research.
Compared to ``BS Hill", ``Extremal QTE CI" has  computational advantage, as it does not require bootstrapping the data.

The asymptotic normality result in Theorem~\ref{thm:hill_extreme_qte} is also confirmed by the normal Q-Q plot in Section~\ref{appendix:qqplot} of the Supplementary Material.

\section{Extremal Quantile Treatment Effect of College Education on Wages} \label{sec:application}

The causal effect of education on wage has been studied extensively in the literature \citep[e.g.,][]{ card1995geographic, heckman1998instrumental, card1999causal, messinis2013returns, heckman2018returns}.
It is well-known that wage exhibits heavy-tailed behavior, 
and there can be considerable confounding between education and wage \citep[e.g.,][]{griliches1977estimating, heckman2006effects}.
In this section, we apply our method to obtain point estimates and confidence intervals of the extremal QTE of college education on wage in the upper tail of the distribution,
where we focus on the 0.99, 0.995, 0.997 and 0.999-QTEs.
As comparison,
we also implement the Firpo--Zhang estimator and the $b$ out of $n$ bootstrap confidence interval of \cite{zhang2018extremal}.

We use data from the National Longitudinal Survey of Youth (NLSY79).
It consists of a representative sample of young Americans who were between 14 and 21 years old at the time of the first interview in 1979, and 
contains a wide range of information about education, adult income, parental background, test scores and behavioral measures of the study participants. 
In particular, we use the NLSY79 data analyzed by \cite{heckman2006effects, heckman2018returns}, which is available from \url{https://www.journals.uchicago.edu/doi/suppl/10.1086/698760}. 
This data set consists of male participants who finished their education before the age of $30$ and were not in the military.
We only consider participants that graduated from high school.
The same (or similar) data set was also used in other literature \cite[e.g.,][]{brand2010benefits, cheng2021heterogeneous, zhou2022attendance}.

The outcome $Y$ is the hourly wage (in US dollar) at age $30$, 
and the treatment $D$ equals $0$ if the person did not receive any college education, and $1$ otherwise.
For the covariates $X$ used for propensity score estimation, 
we follow \cite{heckman2018returns} and consider 
race, region of residence in 1979, urban status in 1979, broken home statue, age in 1979, number of siblings, family income in 1979, education (highest grade completed) of father and mother, scores from the Armed Services Vocational Aptitude Battery (ASVAB) test, and GPAs from 9th grade core subjects (language, math, science and social science). 
This leads to $19$ covariates in total as some of the variables are categorical.
We omit samples with missing values, leading to a data set with $n=805$ samples, among which $432$ $(53.7\%)$ went to college. 
In this data set, the 0.99, 0.995, 0.997 and 0.999 quantiles of hourly wage are 50.47,  60.72, 85.50 and 154.73 US dollar, respectively. 

For the propensity score estimation, considering that there are $19$ covariates and the sample size is $805$, 
we refrain from using too many high-order terms in the sieve method to avoid overfitting.
In particular, we consider two approaches.
The first approach, which we refer to as PROP1, uses only linear terms, leading to sieve basis functions $H_{h_n}(x) = (H_{h_n,j}(x))_{j=1,\ldots,h_n} = (1, x_1, \dots, x_{19})$ with $h_n=20$. This approach is equivalent to logistic regression, 
which is widely used in practice for the propensity score estimation, and is the default option in many packages \citep{olmos2015propensity}. 
The second approach, which we refer to as PROP2, allows second-order terms and uses model selection to avoid overfitting.
Specifically, we first apply a model selection procedure on the $19$ covariates.
Then, we do another round of model selection, allowing only first- and second- order terms of all covariates that were selected in the first step. Both model selection steps are implemented using the R package \textit{glmulti} \citep{JSSv034i12} with Akaike's information criteria and a genetic search algorithm.
The resulting model of PROP2 can be found in Section~\ref{appendix:applicationPROP2} of the Supplementary Material.
We use the same estimated propensity scores for all methods,
and we mention that with a larger sample size, one may consider to use more higher-order terms in the sieve method for the propensity score estimation.


To select the tuning parameter $k$ for our method, here we use the approach of plotting the estimated EVI and QTE versus $k$ and then choose $k$ from the first stable region \citep[e.g.,][]{resnick2007heavy}.
The corresponding plots can be found in Section~\ref{appendix:applicationPlots} of the Supplementary Material.
Based on these plots, we choose $k=85$.
Note that the choice used in Section~\ref{sec:simulations} leads to $k=805^{0.65} \approx 77$, resulting in similar confidence intervals as using $k=85$.

\begin{figure}[h!]
	\centering
	\includegraphics[width=8cm, height=7cm]{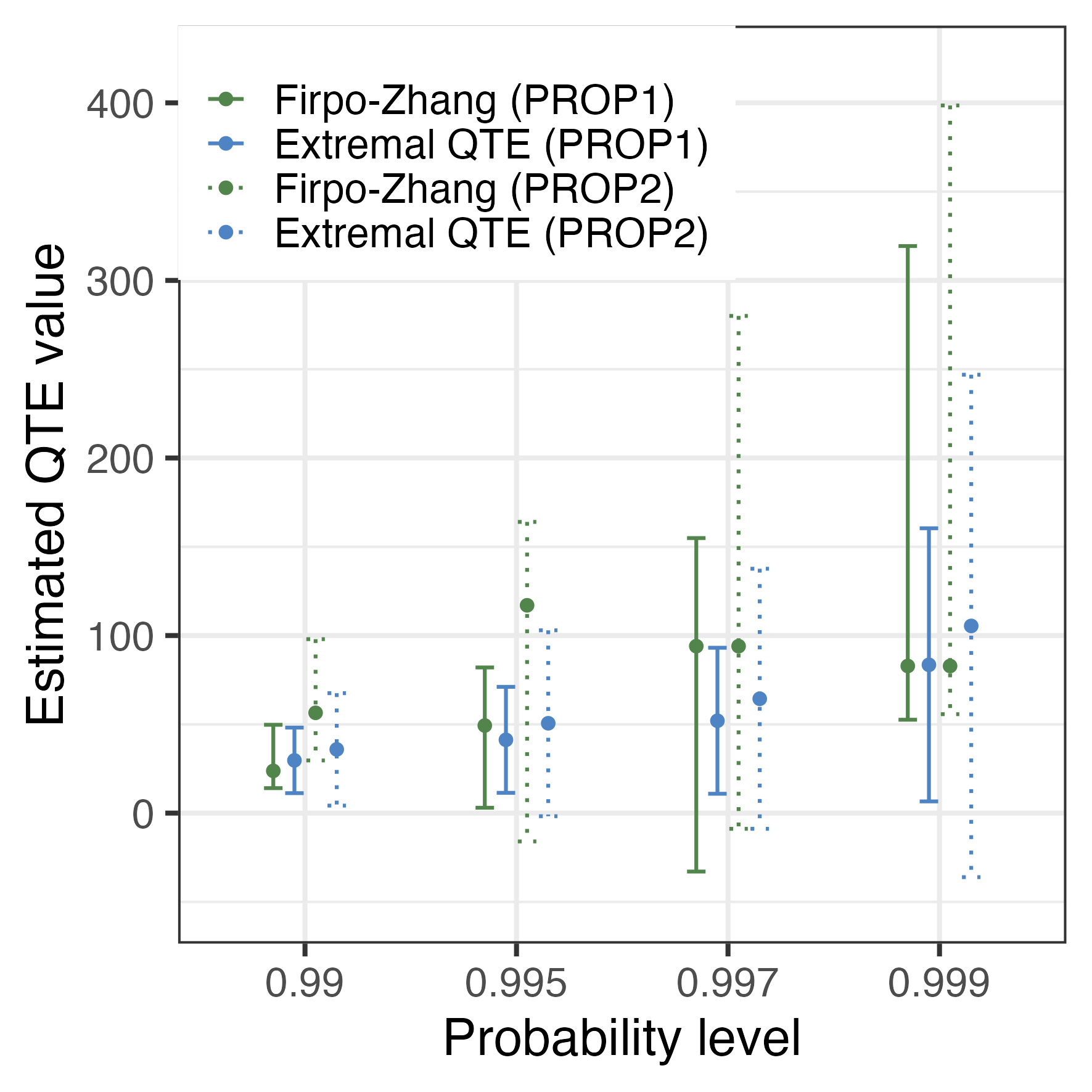}
	\caption[]
	{
    Point estimates and $90\%$-confidence intervals of the extremal QTEs of college education on wage for different quantile levels. 
    Solid and dashed lines denote methods with 
    estimated propensity scores using PROP1 and PROP2, respectively.}
	\label{fig:college_extremal_qte_all}
\end{figure}

Figure~\ref{fig:college_extremal_qte_all} presents the results of the different methods.
Considering the point estimate, we see that both Firpo--Zhang and our method give positive QTE estimates, 
but the corresponding values are quite different in some cases.
In particular, the estimated values of Firpo--Zhang are not monotonically increasing when the quantile levels become more extreme, 
whereas ours are monotonic for these data.
Such monotonicity implies that college education would have a stronger effect on wages for higher quantiles, which seems possible.



The confidence intervals of both methods mostly lie on the positive part, showing strong evidence that the QTEs are positive.
The intervals of our method are considerably narrower than Zhang's $b$ out of $n$ bootstrap intervals for the 0.997- and 0.999-QTEs, a clear advantage of our methodology.
We also observe that the propensity scores estimated by the second approach lead to wider confidence intervals in all cases.
 

At last, we would like to note that the unconfoundedness assumption can not be verified in practice. For example, one may suspect that cognitive ability is not explicitly controlled for in this data example, so the unconfoundedness assumption may not hold. But since we have controlled for many important related covariates, we think that the unconfoundedness assumption is still a suitable approximation. For more discussion we refer to \cite{brand2010benefits}.

\section{Discussion}

We propose a method to estimate the extremal QTE of a binary treatment on a continuous outcome for heavy-tailed distributions under the unconfoundedness assumption. 
Our method, which we call the extremal QTE estimator, builds on the quantile extrapolation approach from extreme value theory.
We use the inverse propensity score weighted intermediate quantile estimates of \cite{firpo2007efficient} 
and our newly proposed causal Hill estimator to extrapolate to extreme quantiles.
We show the asymptotic normality of the causal Hill estimator and the extremal QTE estimator.
In particular, asymptotic normality of the extremal QTE estimator holds for extremal $(1-p_n)$-QTEs, where $n p_n$ may converge to $0$. This is particularly important since it represents a common setting in risk assessment where the quantities of interest are beyond the range of the data.
To the best of our knowledge, our approach is the first that achieves this.
We also develop an estimator for the asymptotic variance which is consistent under suitable assumptions.
This enables us to construct confidence intervals for the extremal QTEs.
Simulations show that our method generally performs well. 



As mentioned before, there is an asymptotic bias term of our proposed extremal QTE estimator $\wh{\delta}(1-p_n)$ (see Theorem~\ref{thm:hill_extreme_qte}), which is due to the asymptotic bias of the causal Hill estimator $\gamma_j^H$. In this paper, we suggest choosing a sufficiently small $k$ so that the asymptotic bias is negligible.
It would be interesting and desired to formally propose bias-corrected versions of $\gamma_j^H$ and $\wh{\delta}(1-p_n)$ in future research.

One potential issue of introducing inverse propensity score weighting to EVI estimators is that it complicates the estimation of the asymptotic variance, making statistical inference difficult.
Bootstrap methods can be useful in practice for constructing confidence intervals for QTEs, as our simulation results suggest.
It is important to study the theoretical validity of such bootstrap based methods in future research.
In particular, it would be interesting to investigate whether the bootstrap based methods are valid even without Assumption~\ref{ass:squared_hill_tozero}.

The proposed method is just the first step towards the goal of doing causal inference for extremes.
The considered causal inference setting is the most common and simplest one: binary treatment with assumed unconfoundedness.
We believe that it is possible to extend it in many ways to fit a range of applications.
For example, it would be interesting to generalize our method to categorical and continuous treatment settings
by using the generalized propensity score \citep{imbens2000role}. 
One may also consider extending our method to other causal inference settings, such as instrumental variable settings that allow for some types of confounding.
At last, quantile extrapolation is not limited to heavy-tailed distributions,
and it would be interesting to extend our proposed extremal QTE estimator to other settings 
where the potential outcome distributions may have lighter tails. 


\section*{Acknowledgments}
Sebastian Engelke was supported by an Eccellenza grant of the Swiss National Science Foundation.

\bibliographystyle{apalike}
\bibliography{extremeQTE}

\addtocontents{toc}{\vspace{.5\baselineskip}}
\addtocontents{toc}{\protect\setcounter{tocdepth}{1}}
\appendix
\newpage

The supplementary material consists of the following seven sections.
\begin{description}
        \item[A] Details of the estimated propensity score using sieve method
	\item[B] Regularity assumptions for sieve estimation and the central limit theorem
	\item[C] Examples satisfying Assumption~\ref{ass:squared_hill_tozero}
	\item[D] Details of Variance Estimation
	\item[E] Supplementary material for simulations
	\item[F] Supplementary material for real application
	\item[G] Proofs
\end{description}


\section{Details of the estimated propensity score using sieve method} \label{appendix:estPropensity}

Suppose that we observe $n$ independent copies $(Y_i, D_i, X_i)_{i=1}^n$ of $(Y, D, X)$.
The main idea of the sieve method is to approximate the logit of the propensity score by a linear combination of sieve basis functions, and then estimate the propensity score by
\begin{equation} \label{SieveEstimatorPropensityScore}
\wh \Pi(x) := \frac{1}{1  + \exp \{ - H_{h_n}(x)^T \wh \pi_n\}}
\end{equation}
where $H_{h_n} = (H_{h_n,j})_{j=1,\ldots,h_n}: \R^r \to \R^{h_n}$ is a vector consisting of sieve basis functions, 
and
\begin{equation}
\wh \pi_n := \argmax_{\pi \in \R^{h_n}} \Sn 
D_i \log L(H_{h_n}(X_i)^T\pi) + (1-D_i) \log(1- L(H_{h_n}(X_i)^T\pi) ),
\label{eq:pihat}
\end{equation}
and  $L(a) := 1/(1+e^{-a})$ is the sigmoid function.

Let $H_{h_n} = (H_{h_n,j})_{j=1,\ldots,h_n}: \R^r \to \R^{h_n}$ be a vector consisting of $h_n$ sieve basis functions.
Following \cite{hirano2003efficient} and \cite{firpo2007efficient},
we use polynomials as the sieve basis functions in this paper.
In particular, we require that $H_{h_n,1} = 1$ and for all $m$ such that $h_n > (m+1)^r$, the span of $H_{h_n}$ contains all polynomials up to order $m$.
For an illustration purpose, some possible examples for $H_{h_n}$ are
$H_{h_n}(x) = (1, x_1, x_2, \dots, x_r)$
or
$H_{h_n}(x) = (1, x_1, x_2, \dots, x_r, x_1^2, x_2^2, \dots,x_r^2)$.
The crucial point in sieve estimation is that the dimension of the sieve space $h_n$ grows to infinity at an appropriate speed with the sample
size $n$. 
In other words, with larger sample size, one may consider a more complex model for the estimation.

\section{Regularity assumptions for sieve estimation and the central limit theorem} \label{appendix:sieve}



For sieve estimation, we require certain regularity assumptions:
\begin{assumption} \label{ass:sieve}
	\
	\begin{enumerate}
		\item[i)] $X$ is continuous and has density $f_X$ such that $\exists c' > 0: c' < f_X(x) < \frac{1}{c'},$
		\mbox{$\forall x \in \supp(X)$}. 
		\item[ii)] $\Pi(x)$ is $s$-times continuously differentiable with 
		all the derivatives bounded, where $s \ge 4r$ and $r$ denotes the dimension of $X$.
		\item[iii)] $\E[\tau_n - \Ind{Y(j) > q_j (1 -  \tau_n )} | x]$  is $u$-times
		continuously differentiable in $x$  with all derivatives bounded by some $M_n$
		uniformly over $\supp(X)$, where $u \in \mathbb{N}$.
		\item[iv)] Let $\zeta(h_n) = \sup_{x \in \supp(X)} \Vert H_{h_n}(x) \Vert$\footnote{For any vector or matrix $A$, the norm $\Vert A \Vert = \sqrt{\text{tr}(A^TA)}$.}.  We assume 
		$\frac{\zeta(h_n)^2 h_n}{\sqrt{n}} \to 0, \frac{ \tau_n \zeta(h_n)^{10} h_n}{n} \to 0 $,
		$n \tau_n \zeta(h_n)^6 h_n^{-s/r} \to 0$ and $\frac{n M_n}{\tau_n h_n^{u/r}} \to 0$.
	\end{enumerate}
\end{assumption}
In Assumption \ref{ass:sieve}, i) and ii) are standard assumptions for sieve estimation. iii) and iv) were introduced by \cite{zhang2018extremal} for the intermediate quantile estimation, and we refer to \cite{zhang2018extremal} for more discussions about these two assumptions.


\cite{newey1994asymptotic} showed that if $H_{h_n}$ consists of orthonormal polynomials, then $\zeta(h_n) = O(h_n)$. In this case, if $h_n = \floor{c_2 n^{c_1}}$ for some positive constants $c_1, c_2$, Assumption \ref{ass:sieve} $iv)$ is equivalent to $ c_1 < \frac{1}{6}$, $ \tau_n n^{11c_1 - 1 } \to 0$, $ \tau_n n^{c_1(6-s/r) + 1} \to 0$ and $M_n n^{1-c_1 u/r}  / \tau_n \to 0$. In particular, since $\tau_n \to 0$ and $n \tau_n \to \infty$, this assumption holds if we have $c_1\le\frac{1}{11}$ and sufficient smoothness.


Below we present the regularity assumptions for the central limit theorem.
\begin{assumption} \label{ass:zhangTechnical}
	~\\
	There exist real numbers $H_1, H_0, H_{10}$ such that
	\begin{align*}
		\frac{1}{\tau_n} \E \biggl[ \frac{P(Y(1) > q_1(1- \tau_n) \mid X)}{\Pi(X)}
		- \frac{1-\Pi(X)}{\Pi(X)} P(Y(1) > q_1(1- \tau_n)\mid X)^2 
		\biggr]
		&\to H_1 \\
		\frac{1}{\tau_n} \E \biggl[ \frac{P(Y(0) > q_0(1- \tau_n) \mid X)}{1-\Pi(X)}
		- \frac{\Pi(X)}{1-\Pi(X)} P(Y(0) > q_0(1- \tau_n)\mid X)^2
		\biggr]
		&\to H_0 \\
		\frac{1}{\tau_n} \ERW{ P(Y(1) > q_1(1-\tau_n)\mid X) P(Y(0) > q_0(1-\tau_n)\mid X) } 
		&\to H_{10}. 
	\end{align*}
	
\end{assumption}

\begin{assumption} \label{ass:hill_Technical}
	~\\
	There exist real numbers $G_1, G_0, G_{10}$ and $J_{1}, J_{0}, J_{10}, J_{01}$ such that
	\begin{align*}
		\frac{1}{\tau_n}  \E \bigg[ \frac{1}{\Pi(X)} \ERW{ S_{1,n}^2 \mid X}  
		-   \frac{1-\Pi(X)}{\Pi(X)}
		\ERW{S_{1,n}\mid X}^2
		\bigg]
		& \to G_1
		\\
		\frac{1}{\tau_n}  \E \bigg[ \frac{1}{1-\Pi(X)} \ERW{ S_{0,n}^2 \mid X}  
		-   \frac{\Pi(X)}{1-\Pi(X)}
		\ERW{S_{0,n}\mid X}^2
		\bigg]
		& \to G_0
		\\ 
		\frac{1}{\tau_n}   \E \bigg[  \ERW{ 
			S_{1,n} \mid X}  
		\ERW{S_{0,n}\mid X}
		\bigg] & \to G_{10}
		\\
		\frac{1}{\tau_n}  \E \bigg[ 
		\frac{1}{\Pi(X)} \ERW{ S_{1,n} \mid X} 
		- \frac{1-\Pi(X)}{\Pi(X)} \ERW{ S_{1,n} \mid X}  P(Y(1) > q_1(1-\tau_n)\mid X) 
		\bigg]
		& \to J_{1}
		\\ 
		\frac{1}{\tau_n}  \E \bigg[
		\frac{1}{1-\Pi(X)} \ERW{ S_{0,n} \mid X}
		- \frac{\Pi(X) }{1-\Pi(X)} \ERW{ S_{0,n} \mid X} P(Y(0) > q_0(1-\tau_n)\mid X) \bigg)
		\bigg]
		& \to J_{0}
		\\
		\frac{1}{\tau_n}  \E \bigg[ \ERW{ 
			S_{1,n}\mid X}  
		P(Y(0) > q_0(1-\tau_n)\mid X) 
		\bigg] & \to J_{10}
		\\
		\frac{1}{\tau_n}  \E \bigg[ \ERW{ 
			S_{0,n}\mid X}  
		P(Y(1) > q_1(1-\tau_n)\mid X) 
		\bigg] & \to J_{01},
	\end{align*}
	where
	\[
	S_{j,n} = 
	\gamma_j \log \left( \frac{\tau_n}{1-F_j(Y(j))} \right) \Ind{Y(j) >  q_j(1-\tau_n)}. 
	\]
\end{assumption}

\section{Examples satisfying Assumption~\ref{ass:squared_hill_tozero}}
\label{appendix:example}
Assumption~\ref{ass:squared_hill_tozero} is sarisfied for 
the following random scale model in Example~\ref{example:random_scale}.
\begin{example} \label{example:random_scale}
	Let $X \in \R^r$ be a random vector and
	let $h_j: \R^r \to \R_+$, $j=0,1$, be functions satisfying 
	$ C_1 \le h_j(X) \le C_2$ almost surely for some real numbers $C_1, C_2$ such that $0 < C_1 \le C_2 < \infty$.
	Let $\varepsilon_0, \varepsilon_1$ be random variables independent of $X$ and let potential outcome
	\[
	Y(j) = h_j(X) \cdot \varepsilon_j. 
	\]
	Let $F_{\varepsilon_j}$ be the distribution function of $\varepsilon_j$.
	Suppose that $U_{\varepsilon_j} := (1/(1-F_{\varepsilon_j}))^\leftarrow$ satisfies the second-order regular variation condition with extreme
	value index $\gamma_j>0$ and second-order parameter $\rho_j \leq 0$. 
	Then $U_j = (1/(1-F_j))^\leftarrow$ satisfies the second-order regular variation condition with the same parameters $\gamma_j$, $\rho_j$,
	and Assumption~\ref{ass:squared_hill_tozero} is met.
\end{example}

We give two concrete examples for the distribution of $\epsilon_j$ in Example~\ref{example:random_scale}.
\begin{example}
	The first example is the  Student $t$-distribution.
	Let $F_\nu$ be the CDF of the Student $t$-distribution with $\nu$ degrees of freedom. 
	It is known that $F _\nu$ satisfies the second-order regular-variation condition with first-order parameter $-\nu$ and
	second-order parameter $\rho = -2$ (see e.g. \mbox{Example 3 in \cite{hua2011second})}.
	Therefore, by Theorem 2.3.9 in \cite{de2007extreme},
	$U_\nu:= (1/(1-F_\nu))^\leftarrow$ is 
	of second-order regular variation with extreme value index $\gamma=1/\nu$
	and second-order parameter $\rho=-2$.
	One special case is the Cauchy distribution for $\nu=1$.
	\label{example:student}
\end{example}

\begin{example}
	Another example is the Fr\'echet model $F_\gamma(x) = \exp(-x^{-1/\gamma})\Ind{x\ge 0}$
	for $\gamma > 0$.
	It is known that  $F_\gamma$ satisfies the second-order regular variation condition with first order parameter $\gamma$ and second-order parameter
	$\rho=-1$ (see e.g. Example 4.2 in \cite{alves2007note}).
	Thus, by Theorem 2.3.9 in \cite{de2007extreme},
	$U_\gamma := (1/(1-F_\gamma))^\leftarrow$ is 
	of second-order regular variation with extreme value index $\gamma$
	and second-order parameter $\rho=-1$.
	\label{example:frechet}
\end{example}

We prove the following Lemma~\ref{lemma:scale_heavy_tailed}. The claim of Example~\ref{example:random_scale} then follows from setting $t=q_j(1-\tau_n)$.
\begin{lemma}
	Let $X \in \R^r$ be a random vector and let $h: \R^r \to \R_+$ be a function satisfying 
	$ C_1 \le h(X) \le C_2$ almost surely for some real numbers $C_1, C_2$ such that $0 < C_1 \le C_2 < \infty$.
	Let $\varepsilon$ be a random variable such that $\varepsilon \indep X$ and let $Z = h(X) \cdot \varepsilon$.
	Denote the CDFs of $Z$ and $\varepsilon$ by $F$ and $\wt F$, respectively.
	Suppose that $\wt U = (1/(1-\wt F))^\leftarrow$ is of second-order regular variation with extreme value index $\gamma > 0$ and second-order parameter $\rho \le 0$.
	Then, $U = (1/(1-F))^\leftarrow$ is of second-order regular variation with extreme value index $\gamma$ and second-order parameter $\rho$.
	In addition, for $t \to \infty$, we have
	\begin{align} \label{ex_claim1}
		\frac{\ERW{P(Z > t \mid X)^2}}{P(Z > t)} \to 0
	\end{align}
	and
	\begin{align} \label{ex_claim2}
		\frac{1}{1-F(t)} \ERW{ \ERW{ \log\left(\frac{1-F(t)}{1-F(Z)}\right) \Ind{Z > t}\bigg|  X }^2}  \to 0 .
	\end{align}
	\label{lemma:scale_heavy_tailed}
\end{lemma}

\begin{proof}[Proof of Lemma~\ref{lemma:scale_heavy_tailed}]
	First, since replacing $\varepsilon$ by $\varepsilon \Ind{\varepsilon \ge 0}$ does not change
	the CDF of $\varepsilon$ or $Z$ on $(0, +\infty)$, we can assume without loss of generality that $\varepsilon \ge 0$.
	
	By Theorem 2.3.9 in \cite{de2007extreme}, $\wt U$ being
	second-order regular variation with extreme value index $\gamma > 0$ and second-order parameter $\rho \le 0$
	is equivalent to
	$1-\wt F$ being second-order regular variation with first-order parameter $-1/\gamma$ and second-order parameter $\rho$.
	Therefore, equation (25) in \cite{hua2011second}
	implies that $1-F$ is of second-order regular variation with parameters $-1/\gamma$ and $\rho$, 
	given the condition that there exists $\delta>0$ such that $\ERW{ h(X)^{1/\gamma - \rho + \delta}} < \infty$, 
	The required condition holds in our case because
	$\ERW{ h(X)^{1/\gamma - \rho + \delta}} \le C_2^{1/\gamma - \rho + \delta} < \infty$ for all $\delta>0$. 
	Theorem 2.3.9 in \cite{de2007extreme} then implies that $U$ is second-order regularly varying with parameters $\gamma$ and $\rho$.
	
	Now we prove claim \eqref{ex_claim1}.
	
	First, since $\varepsilon \indep X$, we have $F(t) = P(\varepsilon \le t/h(X)) = \ERW{ \wt F (t / h (X))} $, and thus for all $t > 0$,
	\begin{equation*}
		\wt F(t/C_2) \le F(t) \le \wt F(t/C_1)
	\end{equation*}
	and
	\[
	1 = \frac{1-\wt F(t/C_2 )}{1-\wt F(t/C_2)}
	\le 
	\frac{ 1- \wt F (t/C_2)}{1-F(t)} \le 
	\frac{1-\wt F(t/C_2 )}{1-\wt F(t/C_1)}.
	\]
	Because $1-\wt F$ is of second-order regular variation 
	with first-order parameter $-1/\gamma$, we have
	\[
	\frac{1-\wt F(t/C_2 )}{1-\wt F(t/C_1)} 
	=
	\frac{1-\wt F(t/C_2 )}{1-\wt F(t)} 
	\frac{1-\wt F(t)}{1-\wt F(t/C_1)} \tendsto{t \to \infty} 
	\left( \frac{C_2}{C_1} \right) ^{1/\gamma},
	\]
	which implies
	\begin{equation}
		\frac{ 1- \wt F (t/C_2)}{1-F(t)} = 
		O(1).
		\label{eq:ratio_wtF_F}
	\end{equation}
	Since $P(Z > t \mid X) = 1- \wt F(t/h(X)) \le 1-\wt F(t/C_2)$, we have that for $t \to \infty$,
	\[
	\frac{\ERW{P(Z > t \mid X)^2}}{P(Z > t)} 
	\le 
	\frac{(1-\wt F(t/C_2))^2}{1-F(t)}
	= O(1-F(t)) = o(1).
	\]
	
	Now we prove claim \eqref{ex_claim2}.
	
	We have that almost surely,
	\begin{equation}\label{secondClaim}
		\begin{aligned}
			0 \le \ERW{ \log\left(\frac{1-F(t)}{1-F(Z)}\right) \Ind{Z > t}\bigg|  X } 
			&   \le 
			\ERW{ \log\left(\frac{1-F(t)}{1-F(\varepsilon\cdot C_2)}\right) \Ind{\varepsilon \cdot C_2 > t} } 
			\\ & \le 
			\ERW{ \log\left(\frac{1-F(t)}{1-\wt F(\varepsilon  \cdot C_2 / C_1)}\right) \Ind{\varepsilon \cdot C_2 > t} } 
			\\ & = 
			\ERW{ \log\left(\frac{1-F(t)}{1-\wt F(\varepsilon)}\right) \Ind{\varepsilon \cdot C_2 > t} } 
			\\ & \quad + 
			\ERW{ \log\left(\frac{1-\wt F(\varepsilon)}{1-\wt F(\varepsilon\cdot C_2 / C_1)}\right) \Ind{\varepsilon \cdot C_2 > t} } .
		\end{aligned}
	\end{equation}
	For the first term, since $\wt F(\varepsilon)$ is uniformly distributed, 
	we have that $\log\left(\frac{1}{1-\wt F(\varepsilon)}\right)$ follows a standard exponential distribution.
	Thus
	\begin{align*}
		\ERW{ \log\left(\frac{1-F(t)}{1-\wt F(\varepsilon)}\right) \Ind{\varepsilon \cdot C_2 > t} } 
		& = \int_{-\log(1-\wt F(t/C_2))}^\infty (z + \log(1-F(t)))e^{-z} \ dz 
		\\ & =
		(1-F(t)) \int_{\log((1-F(t))/(1-\wt F(t/C_2)))}^\infty z 
		e^{-z} \ dz .
	\end{align*}
	Based on \eqref{eq:ratio_wtF_F}, we have that $\log((1-F(t))/(1-\wt F(t/C_2)))$ is of order $O(1)$, and thus
	\[
	\ERW{ \log\left(\frac{1-F(t)}{1-\wt F(\varepsilon)}\right) \Ind{\varepsilon\cdot C_2 > t} } = O(1-F(t)).
	\]
	For the second term, we have
	\begin{align*}
		&\ERW{ \log\left(\frac{1-\wt F(\varepsilon)}{1-\wt F(\varepsilon\cdot C_2/C_1)}\right) \Ind{\varepsilon \cdot C_2 > t} } 
		\le
		(1 - \wt F(t/ C_2))
		\sup_{q > t/C_2 } \log\left(\frac{1-\wt F(q)}{1-\wt F(q\cdot C_2 / C_1)}\right)
	\end{align*}
	Since $1-\wt F$ is of second-order regular variation, we have that
	$ 1 \le \frac{1-\wt F(q)}{1-\wt F(q\cdot C_2/C_1)} = O(1)$,
	and it follows that
	\[
	\sup_{q > t /C_2} \log\left(\frac{1-\wt F(q)}{1-\wt F(q\cdot C_2/C_1)}\right)
	= O(1).
	\]
	Combining this with (\ref{eq:ratio_wtF_F})
	yields
	\[
	\ERW{ \log\left(\frac{1-\wt F(\varepsilon)}{1-\wt F(\varepsilon\cdot C_2/C_1)}\right) \Ind{\varepsilon/C_2 > t} }  = O(1-F(t)). 
	\]
	
	Therefore, by taking squares and expectations on both sides of \eqref{secondClaim}, and using the above two results, we have 
	\begin{align*}
		0 \le \frac{1}{1-F(t)} \ERW{ \ERW{ \log\left(\frac{1-F(t)}{1-F(Z)}\right) \Ind{Z > t}\bigg|  X }^2} 
		\le \frac{1}{1-F(t)} O((1-F(t))^2)
		= o(1),
	\end{align*}
	which proves claim \eqref{ex_claim2}.
\end{proof}


\section{Details of Variance Estimation} \label{appendix:variance}

The true variance in Theorem~\ref{thm:hill_extreme_qte} is $\sigma^2 = v_{\kappa}^T B \Sigma  B^T v_{\kappa}$, where we denote the true covariance matrix
\begin{align}\label{covariance_true}
	\Sigma :=  \left(
	\begin{matrix}                                
		G_{1} &  G_{10} & J_{1} & J_{10} \\
		G_{10} &  G_{0} & J_{01} & J_{0} \\
		J_{1} & J_{01} & H_{1} & H_{10} \\
		J_{10} & J_{0} & H_{10} & H_{0} \\
	\end{matrix}
	\right),
\end{align}
where $H_1, H_0, H_{10}, H_{10}$ are defined according to
Assumption~\ref{ass:zhangTechnical}. and
$G_1, G_0, G_{10}$, $J_1$, $J_0$, $J_{10}$, $J_{01}$ are defined according to
Assumption~\ref{ass:hill_Technical}.

Let 
\begin{align}\label{variance_conservative}
	\wt\sigma^2 := v_{\kappa}^T B \wt \Sigma  B^T v_{\kappa}
\end{align}
with the simplified covariance matrix
\begin{align}\label{covariance_simp}
	\wt \Sigma :=  \left(
	\begin{matrix}                                
		\wt G_{1} &  0 & \wt J_{1} & 0 \\
		0 &  \wt G_{0} & 0 & \wt J_{0} \\
		\wt J_{1} & 0 & \wt H_{1} & 0 \\
		0 &  \wt J_{0} &  0 & \wt H_{0}
	\end{matrix}
	\right)
\end{align}
where the entries are defined as the following limits:
\begin{align*}
	\frac{1}{\tau_n} \ERW{ \frac{P(Y(1) > q_1(1- \tau_n) \mid X)}{\Pi(X)}}
	& \to \wt H_1 \\
	\frac{1}{\tau_n} \ERW{\frac{P(Y(0) > q_0(1- \tau_n) \mid X)}{1-\Pi(X)}}
	& \to \wt H_0 \\
	\frac{1}{\tau_n}  \ERW{ \frac{1}{\Pi(X)} \ERW{ S_{1,n}^2 \mid X } } & \to \wt G_1
	\\
	\frac{1}{\tau_n}  \ERW{ \frac{1 }{1-\Pi(X)} \ERW{ S_{0,n}^2 \mid X} } &\to \wt G_0
	\\ 
	\frac{1}{\tau_n}  \ERW { \frac{1}{\Pi(X)} \ERW{ S_{1,n} \mid X } }
	& \to\wt  J_{1}
	\\ 
	\frac{1}{\tau_n}  \ERW{ \frac{1 }{1-\Pi(X)} \ERW{ S_{0,n} \mid X} } 
	& \to \wt J_{0}.
\end{align*}

Under Assumption \ref{ass:squared_hill_tozero}, the true covariance matrix $\Sigma$ is simplified to $\wt \Sigma$ (see Lemma~\ref{lemma:easy_sigma_hill}),
which leads to the estimator
\begin{equation}\label{covariance_est_ass6}
	\begin{aligned} 
	\wh \Sigma :=  \left(
	\begin{matrix}                                
		\wh G_{1} &  0 & \wh J_{1} & 0 \\
		0 &  \wh G_{0} & 0  & \wh J_{0} \\
		\wh J_{1} & 0 & \wh H_{1} & 0 \\
		0 & \wh J_{0} & 0 & \wh H_{0} \\
	\end{matrix}
	\right)
 \end{aligned}
\end{equation}
with entries
\begin{equation}\label{covariance_est_terms}
	\begin{aligned} 
		\wh{H}_1 &:= \frac{1}{k} \Sn \frac{D_i}{\wh{\Pi}(X_i)^2} \Ind{Y_i > \wh{q}_1(1- \tau_n)} \\ 
		\wh{H}_0 &:= \frac{1}{k} \Sn \frac{1-D_i}{(1-\wh{\Pi}(X_i))^2} \Ind{Y_i > \wh{q}_0(1- \tau_n)} \\
		\wh G_1 & := 
		\frac{1}{k} \Sn (\log(Y_i) - \log(\wh q_1(1-\tau_n))^2 
		\frac{D_i}{\wh \Pi(X_i)^2} \Ind{Y_i > \wh q_1(1-\tau_n)} \\
		\wh G_0 & := 
		\frac{1}{k} \Sn (\log(Y_i) - \log(\wh q_0(1-\tau_n))^2 
		\frac{1-D_i}{ (1- \wh\Pi(X_i))^2} \Ind{Y_i > \wh q_0(1-\tau_n)} \\
		\wh J_{1} & := 
		\frac{1}{k} \Sn (\log(Y_i) - \log(\wh q_1(1-\tau_n)) 
		\frac{D_i}{\wh \Pi(X_i)^2} \Ind{Y_i > \wh q_1(1-\tau_n)} \\
		\wh J_{0} & := 
		\frac{1}{k} \Sn (\log(Y_i) - \log(\wh q_0(1-\tau_n)) 
		\frac{1-D_i}{ (1- \wh\Pi(X_i))^2} \Ind{Y_i > \wh q_0(1-\tau_n)},
	\end{aligned}
\end{equation}
where $\wh q_j(1-\tau_n)$ is the estimator defined by \eqref{eq:firpo_estimator}, and
$\wh \Pi$ is the estimated propensity score in~\eqref{SieveEstimatorPropensityScore}.
Finally, the estimator of the variance is given by 
\begin{align*}
	\wh \sigma^2 := \wh v_{\kappa}^T \wh B \wh \Sigma \wh B^T \wh v_{\kappa},
\end{align*}
where 
\begin{equation*}
	\wh B :=  \left(
	\begin{matrix}                                
		1 & 0 & -\wh \gamma_1^H & 0 \\
		0 & 1 & 0 & -\wh \gamma_0^H \\
	\end{matrix}
	\right)
	\text{ and }\
	\wh v_{\kappa} := \left(
	\begin{matrix}
		\min\{1, \wh \kappa \} \\
		-\min\{1, \frac{1}{\wh \kappa} \}
	\end{matrix}
	\right)
	\text{ with }
	\
	\wh{\kappa} := 
	\frac{\wh Q_1(1-p_n)}{\wh Q_0(1-p_n)}.
\end{equation*}



\section{Supplementary material for simulations}

\subsection{Tuning parameters of Zhang's $b$ out of $n$ bootstrap} \label{appendix:tuningZhang}
For the tuning parameters of the $b$ out of $n$ bootstrap of \cite{zhang2018extremal}, we use the same values as suggested in the paper. Specifically, for the subsample size $b$, we follow the formula suggested in Section 5.5 of \cite{zhang2018extremal}:
\[
b = \left\lfloor 0.4n - \frac{1}{7}(n-300)^{+} - \frac{2.3}{28}(n-1000)^{+} - \frac{7}{40} \left(1-\frac{\log(5000)}{\log(n)}\right)  (n-5000)^{+} \right\rfloor,
\]
where $x^{+}=\max(0,x)$. For sample sizes $n=\{ 1000, 2000, 5000\}$, we obtain $b = \{300, 475, 1000 \} $. 

For the spacing parameter $m$ and $\tau_{n,0}$ in the feasible normalizing factor, we use the formulas described in Section 5.5 of \cite{zhang2018extremal}:
\begin{align*}
	\tau_{n,0} = \min \left( \frac{10}{n}, \frac{0.1b}{n}  \right)
	\quad \text{and} \quad
	m = 1 + \frac{10}{n \tau_{n,0} }.
\end{align*}

\subsection{Q-Q plots} \label{appendix:qqplot}

To empirically verify the asymptotic normality result of our extremal QTE estimator in Theorem~\ref{thm:hill_extreme_qte}, 
we show in Figure~\ref{fig:qqplots} its related normal Q-Q plot.
As comparison, we also present the normal Q-Q plots for the extrapolation estimator with the causal Pickands estimator and the Firpo--Zhang estimator.
The Q-Q plots of all settings are similar, thus we only present those of model $\text{H}_2$ with $n=5000$ and $p_n=5/(n\log(n))$ as an example.
From the plot, our proposed extremal QTE estimator is approximately normal, which empirically verifies Theorem~\ref{thm:hill_extreme_qte}.
The other two estimators, however, appear not to be asymptotically normal. 
This is expected for the Firpo--Zhang estimator because \cite{zhang2018extremal} showed that this estimator is not asymptotically normal in the extreme case.
\begin{figure}
	\centering
	\includegraphics[width=\textwidth]{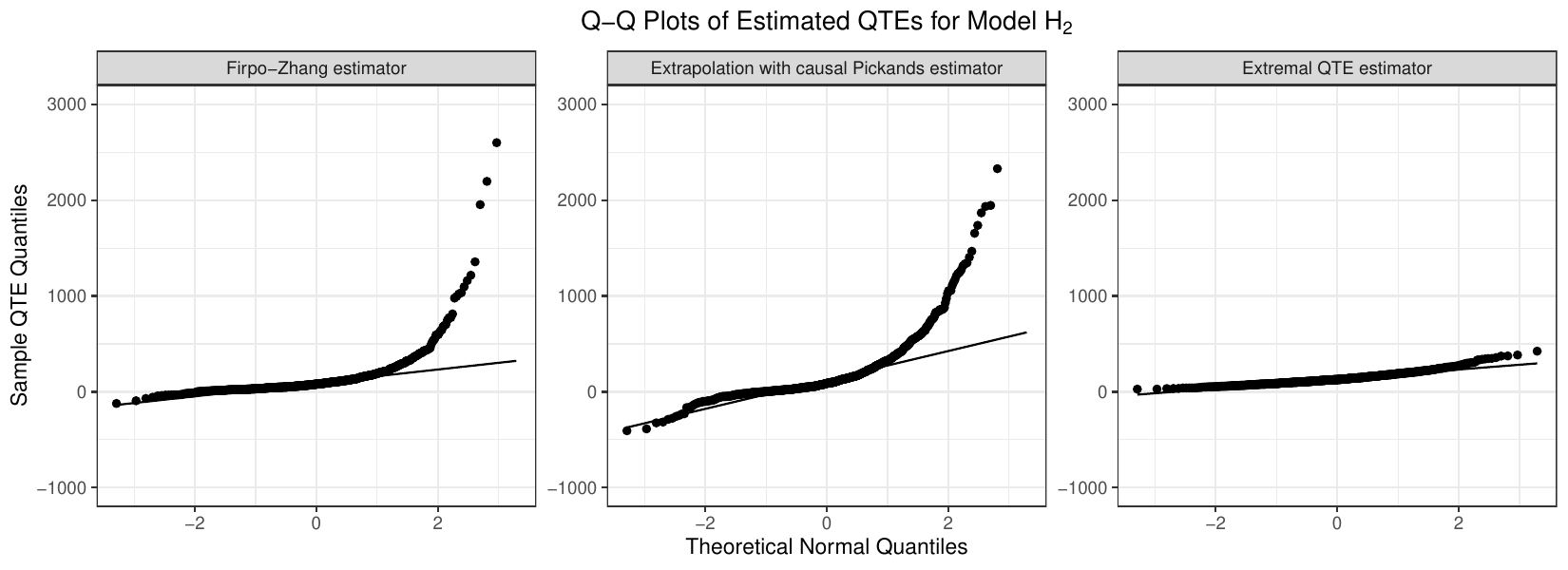}
	\caption[]
	{Normal Q-Q plots of QTE estimates from different methods for model $\text{H}_2$ with
		$n=5000$, \mbox{$p_n=5/(n\log(n))$}. The $x$-axis corresponds to the theoretical quantiles
		of the standard normal distribution, and the $y$-axis corresponds to the sample quantiles
		of the QTE estimators.
	}
	\label{fig:qqplots}
\end{figure}

\subsection{Dependency on $k$} \label{appendix:kdependency}
We implement simulations to investigate how the choice of the tuning parameter $k = n \tau_n$ affects the MSE and the coverage of the 
extrapolation based extremal QTE estimators. We also present the result of the Firpo-Zhang estimator for MSE and Zhang's $b$ out of $n$ bootstrap for coverage as comparison. The considered models $\text{H}_1$, $\text{H}_2$ and $\text{H}_3$ are the same as in Section~\ref{sec:simulations}.

\begin{figure}[h!]
	\centering
	\includegraphics[scale=0.46]{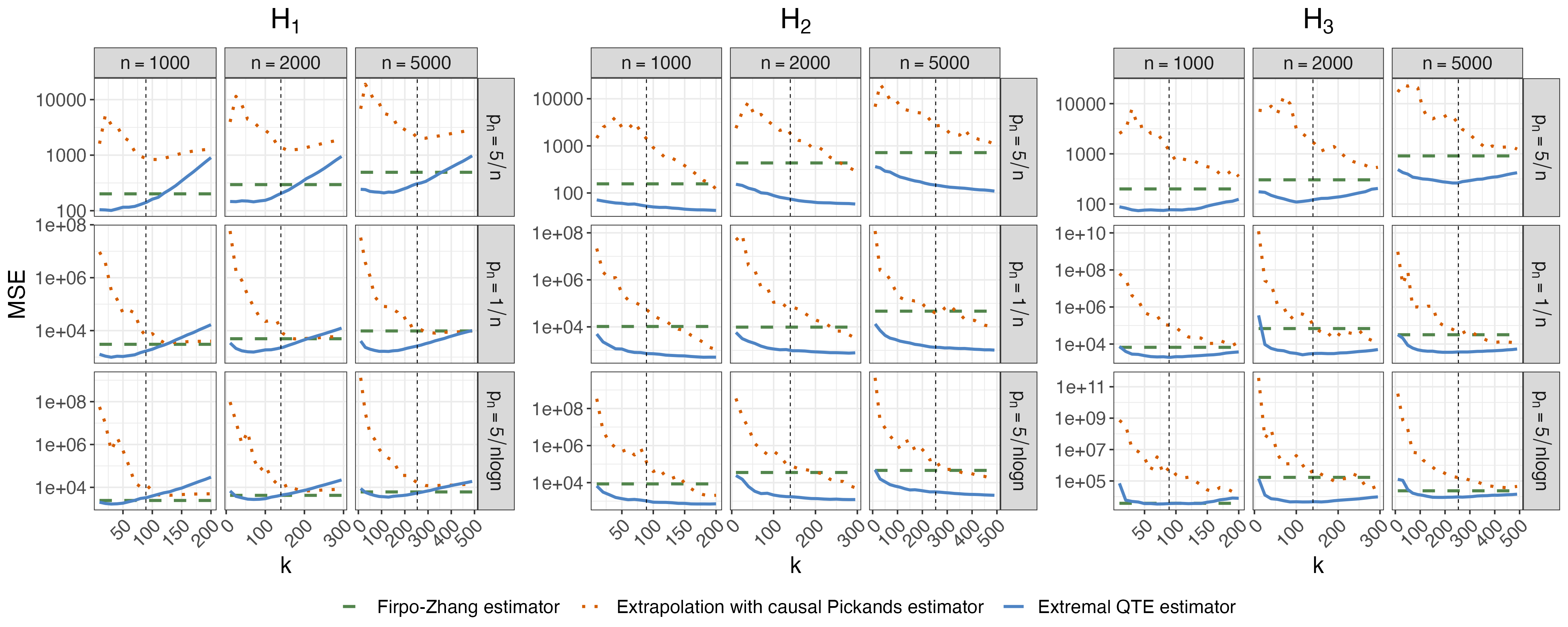}
	\caption[]
	{MSEs of the Firpo-Zhang estimator (horizontal dashed line), the quantile extrapolation method with causal Pickands estimator (dotted line), and our proposed extremal QTE estimator (solid line), with different values of the tuning parameter $k$. The vertical dashed line indicates our choice $k_n = n^{0.65}$ used in Section~\ref{sec:simulations}. }
	\label{fig:mse_kplot}
\end{figure}

Figure~\ref{fig:mse_kplot} shows the simulation results about MSE.
The MSE of the Firpo-Zhang estimator is a line because it does not depend on $k$.
From this figure, we can see that the value of $k$ has a big influence on the MSE of our extremal QTE estimator and the quantile extrapolation method with the causal Pickands estimator.
In particular, a clear bias-variance trade-off with respect to $k$ is shown in the plots related to $\text{H}_1$ and $\text{H}_3$.
We also note that for the more heavy-tailed models $\text{H}_2$ and $\text{H}_3$, our proposed method with the causal Hill estimator outperforms the Firpo-Zhang estimator for a wide range of values of $k$.

\begin{figure}[h!]
	\centering
	\includegraphics[scale=0.45]{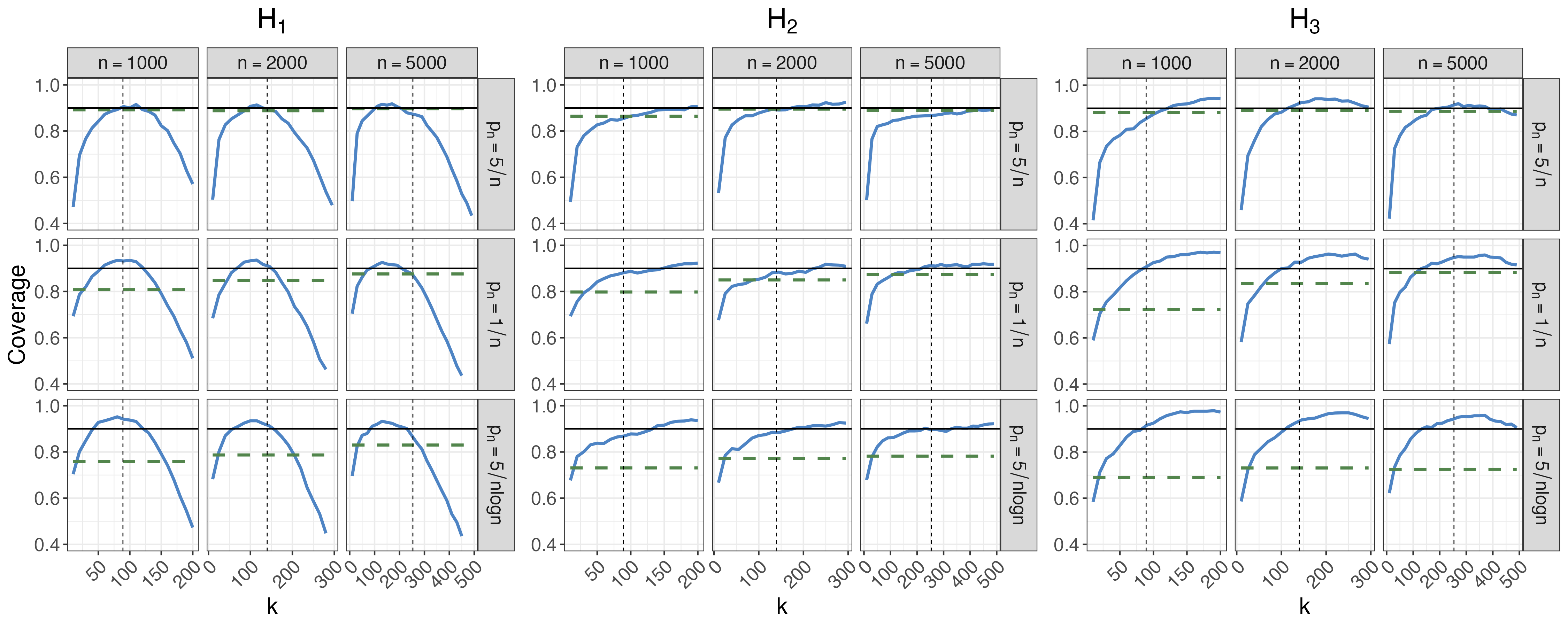}
	\caption[]
	{Coverage of Zhang's $b$ out of $n$ bootstrap method (denoted by the horizontal dashed green line) and our proposed confidence interval \eqref{ConfidenceInterval} (denoted by the solid blue line) for different values of the tuning parameter $k$. The target coverage level is $90\%$, indicated by the black horizontal lines. The vertical dashed line indicates our choice $k_n = n^{0.65}$ used in Section~\ref{sec:simulations}}
	\label{fig:coverage_kplot}
\end{figure}

Figure~\ref{fig:coverage_kplot} shows the simulation results about coverage.
The coverage of the Zhang's $b$ out of $n$ bootstrap method is a line because it does not depend on $k$.
We can see that there is always a range of $k$ where our proposed confidence interval \eqref{ConfidenceInterval} has good coverage.
The particular range, however, depends on the respective model.
We also note that our method works well in terms of coverage for a wide range of $k$ for models $\text{H}_2$.

The above observations generally agree with the observation from classical quantile extrapolation setting, see e.g. \cite{de2007extreme}.

\section{Supplementary material for real application} \label{appendix:application}

\subsection{The resulting model of PROP2} \label{appendix:applicationPROP2}
The resulting model of PROP2 is:
``college$\sim$1+race+region+age80+mhgc$\_$mi+fhgc$\_$mi+
\\ sasvab5+sasvab6+sgr9$\_$scosci$\_$gpa+sasvab2:mhgc$\_$mi+sasvab5:age80+sasvab5:sasvab2+
\\
sasvab6:sasvab2+sgr9$\_$lang$\_$gpa:age80+sgr9$\_$lang$\_$gpa:mhgc$\_$mi+sgr9$\_$scosci$\_$gpa:age80+
\\
race:age80+race:mhgc$\_$mi+region:age80+region:fhgc$\_$mi+region:sgr9$\_$scosci$\_$gpa''.

\subsection{Plots of the estimated EVIs and QTEs with different $k$} \label{appendix:applicationPlots}

Figure~\ref{fig:gamma_kplot_comb} and \ref{fig:college_extremal_qte_comb} shows the plots of the estimated EVIs and QTEs versus the tuning parameter $k$ for the real data analyzed in Section~\ref{sec:application}, respectively.
For Figure~\ref{fig:gamma_kplot_comb}, $\wh \gamma_0^H$ and $\wh \gamma_1^H$ are denoted by triangular and circle, respectively.
In both figures, the red and blue colors correspond to the results with the estimated propensity scores using PROP1 and PROP2 (see Section~\ref{sec:application} for the details of these two approaches), respectively.
\begin{figure}
	\centering
	\includegraphics[width=13cm]{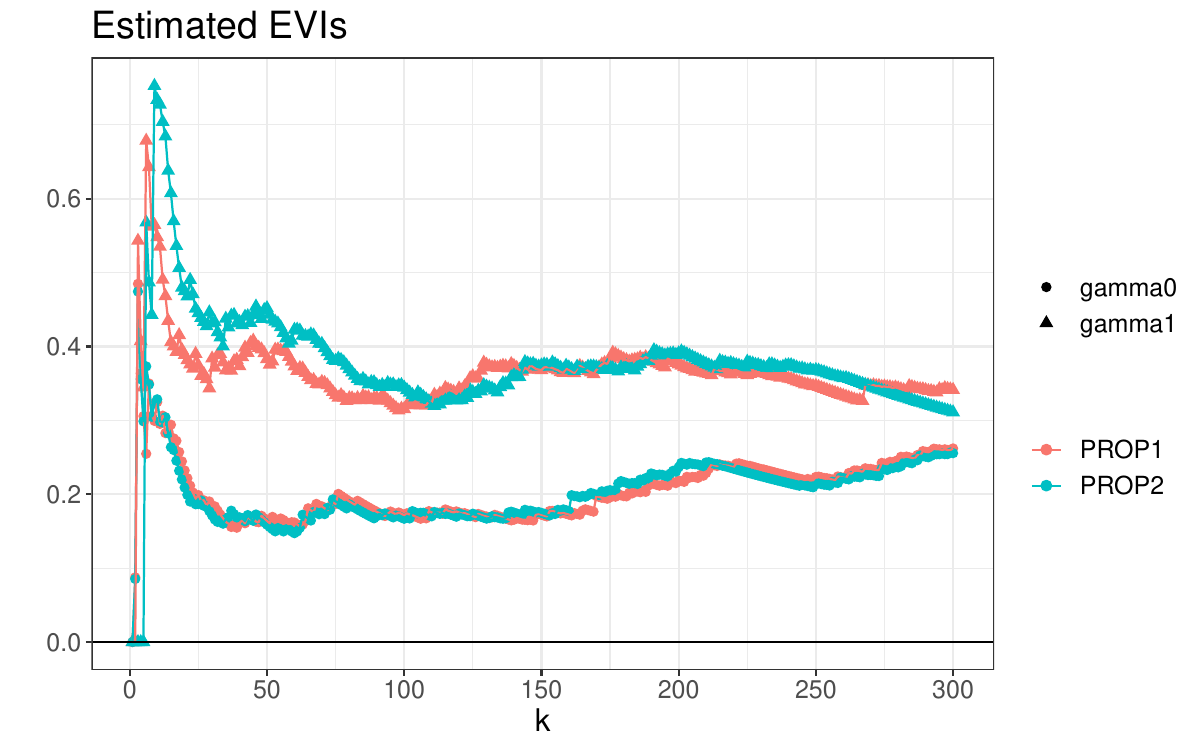}
	\caption[]
	{The EVI estimates $\wh \gamma_0^H$ and $\wh \gamma_1^H$ as a function of the tuning parameter $k$. The points corresponding to $\wh \gamma_0^H$ and $\wh \gamma_1^H$ are denoted by triangular and circle, respectively.
 The red and blue colors correspond to the results with the estimated propensity scores using PROP1 and PROP2, respectively.}
	\label{fig:gamma_kplot_comb}
\end{figure}

\begin{figure}
	\centering
	\includegraphics[width=14cm]{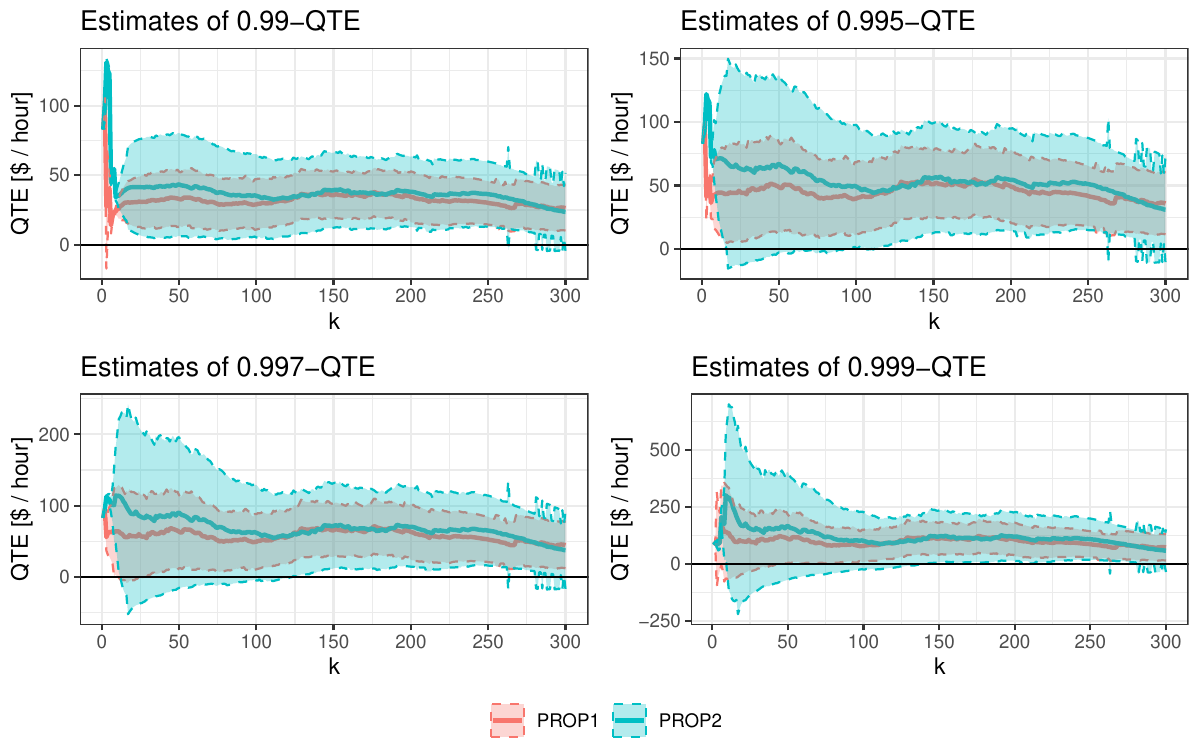}
	\caption[]
	{The extremal QTE estimates (solid lines) as a function of the tuning parameter $k$ for four quantile indices. The shadow indicates the $90\%$-confidence intervals. 
		The red and blue colors correspond to the results with the estimated propensity scores using PROP1 and 
		PROP2, respectively.}
	\label{fig:college_extremal_qte_comb}
\end{figure}

\section{Proofs} \label{appendix:proof}
We mention again that unless otherwise stated, $\tau_n$ denotes the intermediate quantile which satisfies $\tau_n \to 0$ and $k := n \tau_n \to \infty$. We will use notations $k$ and $\tau_n$ interchangeably for convenience.

\subsection{Proof of Lemma~\ref{lemma:hill_consistency}}

To prove Lemma~\ref{lemma:hill_consistency}, we first introduce 
Theorem~\ref{thm:interquantile}, Lemma~\ref{lemma:interquantile_normalizing},
Lemma~\ref{lemma:propensity_adjustment}, Lemma~\ref{lemma:propensity_consistency} and Lemma~\ref{lemma:weighted_cdf_diff}.
Theorem~\ref{thm:interquantile} is a special case of 
Theorem 3.1 in \cite{zhang2018extremal}, so we omit its proof.

\begin{theorem}
	Suppose that Assumptions \ref{ass:1}, \ref{ass:tail}, and \ref{ass:sieve} hold,
	and assume that $n \tau_n  \to \infty$ and $\tau_n \to 0$.
	Let
	\begin{equation*}
		\lambda_{j,n} := \sqrt{\frac{n}{\tau_n}} f_j(q_j(1-\tau_n))
	\end{equation*}
	and consider the random vector 
	\begin{equation*}
		\left(
		\begin{matrix}
			\wh{\Delta}_1^n(\tau_n) \\
			\wh{\Delta}_0^n(\tau_n)
		\end{matrix}
		\right)
		:= 
		\left(
		\begin{matrix}
			\lambda_{1,n}( \wh{q}_1(1-\tau_n) - q_1(1-\tau_n)) \\
			\lambda_{0,n}( \wh{q}_0(1-\tau_n) - q_0(1-\tau_n)) 
		\end{matrix}
		\right).
	\end{equation*}
	Then for $j=0,1$,
	\[
	\wh{\Delta}_j^n(\tau_n) =  \frac{1}{\sqrt{n}} \Sn \phi_{i,j,n} +  o_{p}(1),
	\]
	where 
	\begin{equation*}
		\begin{aligned}
			\phi_{i,1,n} &:= \frac{1}{\sqrt{ \tau_n}} \left( \frac{D_i}{\Pi(X_i)} T_{i,1,n} 
			- \frac{\ERW{T_{i,1,n}\mid X_i}}{\Pi(X_i)} (D_i - \Pi(X_i)) \right)
			\\
			\phi_{i,0,n} &:= \frac{1}{\sqrt{\tau_n}} \left( \frac{1-D_i}{1-\Pi(X_i)} T_{i,0,n} 
			+ \frac{\ERW{T_{i,0,n}\mid X_i}}{1-\Pi(X_i)} (D_i - \Pi(X_i)) \right)
		\end{aligned}
	\end{equation*}
	and
	\[
	T_{i,j,n} :=  \Ind{Y_i(j) > q_j(1-\tau_n)}- \tau_n.
	\]
	In particular, if Assumption \ref{ass:zhangTechnical} holds,     
	then
	\begin{equation*}
		(\wh{\Delta}_1^n(\tau_n), \wh{\Delta}_0^n(\tau_n)) 
		\tendsto{D} N,
	\end{equation*}
	where $N$ is bivariate Gaussian vector with mean zero and covariance matrix 
	\begin{equation*}
		\mathcal{H} =  \left(
		\begin{matrix}                                
			H_1  & H_{10}   \\                                               
			H_{10} & H_0   \\                                               
		\end{matrix}
		\right)
	\end{equation*}
	with $H_1, H_0$ and $H_{10}$ defined as in
	Assumption~\ref{ass:zhangTechnical}.
	\label{thm:interquantile}
\end{theorem}

Lemma~\ref{lemma:interquantile_normalizing} shows that for a CDF $F_j$ with positive extreme value index, the normalization sequence $\lambda_{j,n}$ can be 
replaced by a simpler expression.

\begin{lemma} \label{lemma:interquantile_normalizing}
	For $j=0,1$,
	suppose Assumption~\ref{ass:tail} is met
	and $F_j$ has an extreme value index $\gamma_j >0$. Then
	\[
	\lim_{n \to \infty} \frac{\gamma_j \lambda_{j,n}}{\sqrt{k} q_j(1-\tau_n)^{-1}} = 1.
	\]
\end{lemma}

\begin{proof}[Proof of Lemma~\ref{lemma:interquantile_normalizing}]
	By Assumption~\ref{ass:tail}, $F_j$ satisfies the max-domain of attraction
	condition with a positive extreme value index $\gamma_j$ and its 
	density $f_j$ is monotone in the upper tail.
	Therefore, the von Mises condition
	\[
	\lim_{t \to \infty} \frac{t f_j(t)}{1-F_j(t)} = \frac{1}{\gamma_j}
	\]
	holds by Theorem 2.7.1  in \cite{dehaan1970regular}. So we have 
	\begin{align*}
		\lim_{n \to \infty} \frac{\lambda_{j,n}}{\sqrt{k} q_j(1-\tau_n)^{-1}} 
		& = \lim_{n \to \infty} \frac{ q_j(1-\tau_n) f_j(q_j(1-\tau_n))}{\tau_n}  \\
		& = \lim_{n \to \infty} \frac{ q_j(1-\tau_n) f_j(q_j(1-\tau_n))}{1-F_j(q_j(1-\tau_n))}  \\
		& = \lim_{t \to \infty} \frac{t f_j(t)}{1-F_j(t)} \\
		& = \frac{1}{\gamma_j},
	\end{align*}
	where the second last equality is obtained by setting $t = q_j(1-\tau_n)$.
\end{proof}

Lemma~\ref{IdentifiabilityPropensityScore} is a classical result in causal inference literature, and we omit its proof.
\begin{lemma} \label{IdentifiabilityPropensityScore}
	Let $g$ be a measurable function such that $\ERW{|g(Y(1))|}$ and $\ERW{|g(Y(0))|}$ are finite. Suppose $(Y(1), Y(0)) \indep D \ | \ X$
	and there exists $c > 0$ such that $c < \Pi(X) < 1-c$ almost surely.
	Then we have  
	\begin{align*}
		\ERW{ g(Y) \frac{D}{\Pi(X)} } & = \ERW{ g(Y(1)) } \quad \text{and}
		\quad \ERW{ g(Y) \frac{1-D}{1-\Pi(X)} }  = \ERW{ g(Y(0)) }.
	\end{align*}
	\label{lemma:propensity_adjustment}
\end{lemma}

Lemma~\ref{ConvergenceRatePropensityScore} gives the convergence rate of the estimated propensity score by using the sieve method.
\begin{lemma} \label{ConvergenceRatePropensityScore}
	Suppose Assumptions \ref{ass:1} and \ref{ass:sieve} are met.
	Then we have
	\[
	\sup_{x \in \supp(X)} | \wh \Pi(x) - \Pi(x) | = o_p(k^{-1/4}).
	\]
	In particular, this implies 
	\[
	\sup_{x \in \supp(X)} \bigg| \frac{1}{\wh \Pi(x)} - \frac{1}{\Pi(x)} \bigg| = o_p(1).
	\]
	\label{lemma:propensity_consistency}
\end{lemma}

\begin{proof}[Proof of Lemma~\ref{ConvergenceRatePropensityScore}]
	By Assumption \ref{ass:sieve} $iv)$
	$\frac{\zeta(h_n)^2 h_n}{\sqrt{n}} \to 0$, 
	we have $\frac{\zeta(h_n)^4}{n} \to 0$
	since $h_n \to \infty$.
	Therefore,  we can apply Lemma 1 and 2 in \cite{hirano2003efficient} to obtain
	\[
	\sup_{x \in \supp(X)} | \wh \Pi(x) - \Pi(x) | = O_p(\zeta(h_n) \sqrt{\frac{h_n}{n}} 
	+ \zeta(h_n) h_n^{-s/2r)}).
	\]
	The condition $ \frac{\zeta(h_n)^2 h_n}{\sqrt{n}} \to 0$
	implies  $\zeta(h_n) \sqrt{\frac{h_n}{n}} = o(n^{-1/4})$,
	and the Assumption \ref{ass:sieve} $iv)$
	$n \tau_n \zeta(h_n)^6 h_n^{-s/r} \to 0$ implies
	$ \zeta(h_n) h_n^{-s/2r} = o((n \tau_n)^{-1/2}) = o(k^{-1/2}) $
	as $\zeta(h_n) \ge 1$.
	In addition, $k=n \tau_n \to \infty$ and $\tau_n \to 0$ implies $n^{-1/4} = o(k^{-1/4})$.
	Combining the above rates, we have 
	\[
	\sup_{x \in \supp(X)} | \wh \Pi(x) - \Pi(x) | 
	= O_p( o(k^{-1/4}) + o(k^{-1/2}) )
	= o_p(k^{-1/4}).
	\]
	In particular, we have 
	$ 
	\sup_{x \in \supp(X)} | \wh \Pi(x) - \Pi(x) | = o_p(1).
	$
	The second part of the lemma then follows from the assumption
	that $\Pi(x)$ is continuous and bounded away from zero, 
	which allows us to apply the continuous mapping theorem (see Theorem 7.25 in \cite{kosorok2007introduction}).

\end{proof}

Lemma~\ref{lemma:weighted_cdf_diff} shows that the following two terms converge to zero in probability. This lemma will be used many times in the remaining proofs, so we prove it here.

\begin{lemma}
	Suppose Assumptions \ref{ass:1}, \ref{ass:tail} \ref{ass:sieve} and \ref{ass:zhangTechnical} hold.
	Then
	\begin{align*}
		\frac{1}{n \tau_n}   \Sn \frac{D_i}{\wh{\Pi}(X_i)}  
		\left( \Ind{Y_i > \wh{q}_1(1- \tau_n)} - \Ind{Y_i > {q}_1(1- \tau_n)} \right)
		& \tendsto{P} 0
		\\ \text{and} \quad 
		\frac{1}{n \tau_n }  \Sn \frac{1-D_i}{1-\wh{\Pi}(X_i)}  
		\left( \Ind{Y_i > \wh{q}_0(1- \tau_n)} - \Ind{Y_i > {q}_0(1- \tau_n)} \right)
		& \tendsto{P} 0.
	\end{align*}
	\label{lemma:weighted_cdf_diff}
\end{lemma}

\begin{proof}[Proof of Lemma~\ref{lemma:weighted_cdf_diff}]
	We show the first part of the lemma, and the second part follows analogously.
	For simplicity of notation, we denote $t_n = 1- \tau_n$. So 
	\begin{align}
		&\left| \frac{1}{n \tau_n}   \Sn \frac{D_i}{\wh{\Pi}(X_i)}  
		\left( \Ind{Y_i > \wh{q}_1(1- \tau_n)} - \Ind{Y_i > {q}_1(1- \tau_n)} \right) \right| \\
		= &   
		\left| \frac{1}{n \tau_n}   \Sn \frac{D_i}{\wh{\Pi}(X_i)}  
		\left( \Ind{Y_i > \wh{q}_1{(t_n})} - \Ind{Y_i > {q}_1{(t_n})} \right) \right|  \nonumber\\
		\le & 
		\left| \frac{1}{n \tau_n}   \Sn \frac{D_i}{\wh{\Pi}(X_i)}  
		\left( 1- t_n - \Ind{Y_i > \wh{q}_1{(t_n})} \right) \right| 
		\label{eq:var_est_3_term} \\
		& + \left| \frac{1}{n \tau_n}   \Sn \frac{D_i}{\wh{\Pi}(X_i)}  
		\left( 1 -t_n - \Ind{Y_i > q_1{(t_n})} \right) \right|
		\label{eq:var_est_4_term}
	\end{align}
	by the triangle inequality.
	Now we show that both terms \eqref{eq:var_est_3_term} and \eqref{eq:var_est_4_term} converge to zero in probability, which then proves the original claim. 
	
	For the first term \eqref{eq:var_est_3_term}, we need the subgradient condition for $\wh{q}_1$ (defined by \eqref{eq:firpo_estimator}).
	Specifically, since 
	\[
	\mathcal{L}_n(q) =  \Sn \frac{D_i}{\wh{\Pi}(X_i) }
	(Y_i - q) (t_n - \Ind{Y_i \le q})  \\
	\]
	is convex, a necessary condition for  $\wh{q}_1(t_n) = \argmin_{q \in \R} \mathcal{L}_n(q)$
	is that $0 \in \partial\mathcal{L}_n(\wh{q}_1(t_n))$ where $ \partial\mathcal{L}_n(q) $ denotes the set of subgradients of $\mathcal{L}_n$ at $q$
	(for details see e.g.  \cite{bubeck2015convex}).
	
	Because $Y_i$ is a continuous random variable for $i \in \{1, \ldots, n\}$, there exists at most one $Y_i = \wh{q}_1(t_n)$ almost surely.
	In the first case where $Y_i \neq \wh{q}_1(t_n)$ for all $ i \in \{1, \ldots, n\}$, the subgradient condition implies
	\begin{align*}
		0 = \partial\mathcal{L}_n(\wh{q}_1(t_n)) =
		\sum_{ i =1}^n \frac{D_i}{\wh{\Pi}(X_i) } \{\Ind{Y_i \le \wh{q}_1(t_n)}  - t_n\}.
	\end{align*}
	In the second case where there exists some $i_0 \in \{1, \ldots, n\}$ such that $Y_{i_0}=\wh{q}_1(t_n)$, we have 
	\begin{align*}
		\partial\mathcal{L}_n(\wh{q}_1(t_n)) =
		\sum_{ i \in \{1, \ldots, n \} \setminus \{i_0\}} \frac{D_i}{\wh{\Pi}(X_i) } \{\Ind{Y_i \le \wh{q}_1(t_n)}  - t_n\}  
		+ \sum_{i =i_0}   \frac{D_i}{\wh{\Pi}(X_i) } [-t_n, 1-t_n],
	\end{align*}
	where $[-t_n, 1-t_n]$ is an interval and the set addition is understood elementwise. 
	The subgradient condition then implies that there exists some $t \in [-t_n, 1-t_n]$ such that 
	\begin{align*}
		0 & = \sum_{ i \in \{1, \ldots, n \} \setminus \{i_0\}} \frac{D_i}{\wh{\Pi}(X_i) } (\Ind{Y_i \le \wh{q}_1(t_n)}  - t_n)  
		+   \frac{D_{i_0}}{\wh{\Pi}(X_{i_0}) } t   \\
		& = 
		\Sn \frac{D_i}{\wh{\Pi}(X_i) } (\Ind{Y_i \le \wh{q}_1(t_n)}  - t_n)  
		+   \frac{D_{i_0}}{\wh{\Pi}(X_{i_0}) } (t - 1 + t_n).
	\end{align*}
	Hence 
	\begin{align*}
		\left|
		\Sn \frac{D_i}{\wh{\Pi}(X_i) } (\Ind{Y_i \le \wh{q}_1(t_n)}  - t_n)  
		\right|
		=
		\left|
		\frac{D_{i_0}}{\wh{\Pi}(X_{i_0}) } (t - 1 + t_n).
		\right|
		\leq
		\sup_x \left| \frac{1}{\wh{\Pi}(x)} \right|
	\end{align*}
	since $| D_{i_0} (t - 1 + t_n) | \le 1$.
	
	Combining the above two cases, we have that almost surely
	\begin{align*}
		\left| \frac{1}{n \tau_n}   \Sn \frac{D_i}{\wh{\Pi}(X_i)}  
		\left( 1- t_n - \Ind{Y_i > \wh{q}_1{(t_n})} \right) \right| 
		& = 
		\left| \frac{1}{n \tau_n}   \Sn \frac{D_i}{\wh{\Pi}(X_i)}  
		\left( t_n - \Ind{Y_i \le \wh{q}_1{(t_n})} \right) \right| 
		\\ & \le 
		\frac{1}{n \tau_n} \sup_x \left| \frac{1}{\wh{\Pi}(x)} \right|.
	\end{align*}
	By Assumption \ref{ass:1}, we have that $\sup_x \frac{1}{\Pi(x)} < \frac{1}{c} $. So by Lemma
	\ref{lemma:propensity_consistency} we have 
	\begin{align}
		\sup_x \left| \frac{1}{\wh{\Pi}(x)} \right| \le 
		\sup_x \left| \frac{1}{\wh{\Pi}(x)} - \frac{1}{\Pi(x)} \right| 
		+ \sup_x \left| \frac{1}{\Pi(x)} \right| = o_p(1) + \frac{1}{c}  = O_p(1).
		\label{eq:prob_bounded_prop}
	\end{align}
	Because $n\tau_n \to \infty$, we have 
	\[
	\left| \frac{1}{n \tau_n}   \Sn \frac{D_i}{\wh{\Pi}(X_i)}  
	\left( 1- t_n - \Ind{Y_i > \wh{q}_1{(t_n})} \right) \right|  \tendsto{P} 0.
	\]
	
	For the second term (\ref{eq:var_est_4_term}), the proof
	of Theorem 3.1 in \cite{zhang2018extremal} showed that the term
	\[
	\frac{1}{\sqrt{n \tau_n}}   \Sn \frac{D_i}{\wh{\Pi}(X_i)}  
	\left( 1-t_n  - \Ind{Y_i > q_1{(t_n})} \right)
	\]
	converges in distribution to a normal random variable
	(in particular, \cite{zhang2018extremal} proved the corresponding result for lower quantiles, but the same holds for upper quantiles).
	Hence, we have 
	\begin{align*}
		&\left| \frac{1}{n \tau_n}   \Sn \frac{D_i}{\wh{\Pi}(X_i)}  
		\left( 1-t_n  - \Ind{Y_i > q_1{(t_n})} \right) \right| \\
		=& 
		\frac{1}{\sqrt{n \tau_n}} \left| \frac{1}{\sqrt{n \tau_n}}   \Sn \frac{D_i}{\wh{\Pi}(X_i)}  
		\left( 1-t_n  - \Ind{Y_i > q_1{(t_n})} \right) \right| \\
		=& \frac{1}{\sqrt{n \tau_n}} O_{p}(1)
		\tendsto{P} 0.
	\end{align*}

\end{proof}

Now we give the proof of Lemma~\ref{lemma:hill_consistency}.
\begin{proof}[Proof of Lemma~\ref{lemma:hill_consistency}]
	We show the claim for $j=1$, and the case of $j=0$ can be proved analogously.
	First, we expand $\wh \gamma_1^H$ (defined by \eqref{hill_est}) as 
	\[
	\wh \gamma_1^H = 
	G_n^1 + G_n^2 + G_n^3 + G_n^4
	\]
	where 
	\begin{align*}
		G_n^1  &= \frac{1}{k} \Sn (\log(Y_i) - \log( q_1(1-\tau_n))) 
		\frac{D_i}{ \Pi(X_i)} \Ind{Y_i >  q_1(1-\tau_n)}, \\
		G_n^2 &=  (\log(q_1(1-\tau_n)) - \log( \wh q_1(1-\tau_n)) )
		\frac{1}{k} \Sn 
		\frac{D_i}{ \Pi(X_i)} \Ind{Y_i >  q_1(1-\tau_n)}, 
		\\
		G_n^3 &=  \frac{1}{k} \Sn (\log(Y_i) - \log( \wh q_1(1-\tau_n))) 
		\frac{D_i}{  \wh \Pi(X_i)} \left( \Ind{Y_i > \wh q_1(1-\tau_n)} -
		\Ind{Y_i >  q_1(1-\tau_n)} 
		\right), \\
		G_n^4 &= \frac{1}{k} \Sn (\log(Y_i) - \log( \wh q_1(1-\tau_n))) 
		D_i  \left( \frac{1}{  \wh \Pi(X_i)}- \frac{1}{ \Pi(X_i)}  \right)
		\Ind{Y_i >  q_1(1-\tau_n)} .
	\end{align*}
	In the following, we will show that $G_n^1 \tendsto{P} \gamma_1$ and that
	$G_n^2, G_n^3, G_n^4$ converge to zero in probability, which then proves the original claim.
	
	Now we prove that $G_n^1 \tendsto{P} \gamma_1$.
	This part of the proof is similar to the proof of Theorem 3.2.2 in \cite{de2007extreme}, which shows the consistency of the classical Hill estimator.
	First, we have
	\[
	G_n^1 = \frac{1}{k} \Sn (\log(Y_i(1)) - \log( q_1(1-\tau_n)) ) \frac{D_i}{ \Pi(X_i)} \Ind{Y_i(1) >  q_1(1-\tau_n)}
	\]
	because $Y_i=Y_i(1)D_i + Y_i(0)(1-D_i)$.
	
	Since $F_1$ is the CDF of $Y_i(1)$, 
	$F_1(Y_i(1))$ is uniformly distributed on $[0,1]$.
	Let $ Z_i(1) = 1/( 1-F_1(Y_i(1)) )$, 
	so its CDF is $1-1/z$ for $z\ge 1$, which then implies that 
	$\log(Z_i(1))$ has a standard exponential distribution.
	Let $U_1 = (1/(1-F_1))^\leftarrow$ be the tail function of $Y_i(1)$.
	Then we have that $Y_i(1) = U_1(Z_i(1))$ and $\Ind{Y_i(1) >  q_1(1-\tau_n)}=\Ind{Z_i(1) > \tau_n^{-1}}$ almost surely.
	Since $q_1(1-\tau_n) = U_1(\tau_n^{-1})$, we have that almost surely
	\begin{align*}
		G_n^1  = \frac{1}{k} \Sn\left( \log(U_1(Z_i(1))) - \log\left( U_1 (\tau_n^{-1})\right) \right) 
		\frac{D_i}{ \Pi(X_i)} \Ind{Z_i(1) > \tau_n^{-1}}.
	\end{align*}
	
	By Assumption \ref{ass:tail} $iii)$, $F_1$ satisfies the max-domain of attraction condition with extreme value index $\gamma_1 > 0$, so by Theorem 1.1.6 and Corollary 1.2.10 in \cite{de2007extreme}, we have that for all $x>0$,
	\[
	\lim_{t \to \infty} \frac{U_1(tx)}{U_1(t)} = x^{\gamma_1}.
	\]
	Then, by the statement 5 of the Proposition B.1.9 in \cite{de2007extreme}, we have that for any 
	$\varepsilon, \varepsilon' >0$ such that $\varepsilon < 1,\ \varepsilon' < \gamma_1$, 
	there exists some $t_0$ such that for $x \ge 1$, $t\ge t_0$,
	\[
	(1-\varepsilon) x^{\gamma_1- \varepsilon'}
	< 
	\frac{U_1(tx)}{U_1(t)}
	<
	(1+\varepsilon) x^{\gamma_1+ \varepsilon'},
	\]
	which is equivalent to
	\[
	\log(1-\varepsilon) + (\gamma_1 - \varepsilon') \log(x)
	< \log(U_1(tx)) - \log(U_1(t)) < 
	\log(1+\varepsilon) + (\gamma_1 + \varepsilon') \log(x).
	\]
	For large enough $n$ and for $i \in \{1,\dots,n\}$ such that $Z_i(1) > \tau_n^{-1}$, we can set $t = \tau_n^{-1} $ and $x = Z_i(1)\tau_n$ to obtain
	\[
	\log(1-\varepsilon) + (\gamma_1 - \varepsilon') \log(Z_i(1)\tau_n)
	< \log(U_1(Z_i(1))) - \log(U_1(\tau_n^{-1})) < 
	\log(1+\varepsilon) + (\gamma_1 + \varepsilon') \log(Z_i(1)\tau_n).
	\]
	Multiplying by $\frac{1}{k}\frac{D_i}{ \Pi(X_i)}$ on both sides of the above inequality and summing up all $i \in \{1,\dots,n\}$ with $Z_i(1) > \tau_n^{-1}$ gives us that almost surely, 
	$G_n^1$ lies in the interval $[a, b]$ with
	\begin{align*}
		&a= \log(1 - \varepsilon) \frac{1}{k} \Sn  
		\frac{D_i}{ \Pi(X_i)} \Ind{Z_i(1) > \tau_n^{-1}} 
		+ (\gamma_1 - \varepsilon') \frac{1}{k} \Sn  
		\log(Z_i(1)\tau_n)
		\frac{D_i}{ \Pi(X_i)} \Ind{Z_i(1) > \tau_n^{-1}} 
		\bigg] \quad \text{and} \\
		&b= \log(1 + \varepsilon) \frac{1}{k} \Sn  
		\frac{D_i}{ \Pi(X_i)} \Ind{Z_i(1) > \tau_n^{-1}} 
		+ (\gamma_1 + \varepsilon') \frac{1}{k} \Sn  
		\log(Z_i(1)\tau_n)
		\frac{D_i}{ \Pi(X_i)} \Ind{Z_i(1) > \tau_n^{-1}} 
		\bigg].
	\end{align*}
	Since $\varepsilon$ and $\varepsilon'$ can be arbitrarily small, to prove $G_n^1 \tendsto{P} \gamma_1$, it is enough to show
	\begin{align*}
		(i) \quad & \frac{1}{k} \Sn  \frac{D_i}{ \Pi(X_i)} \Ind{Z_i(1) > \tau_n^{-1}} \tendsto{P} 1 \quad \text{and} \\
		(ii) \quad & \frac{1}{k} \Sn
		\log(Z_i(1) \tau_n )
		\frac{D_i}{ \Pi(X_i)} \Ind{Z_i(1) > \tau_n^{-1}} 
		\tendsto{P} 1.
	\end{align*} 
	
	For (i), 
	let $b_n = k$ and $S_n = \Sn  \frac{D_i}{ \Pi(X_i)} \Ind{Z_i(1) > \tau_n^{-1}} $. We have
	\[
	\ERW{ S_n } = n \ERW{ \frac{D_i}{ \Pi(X_i)} \Ind{Z_i(1) > \tau_n^{-1}} } 
	=n \ERW{ \frac{D_i}{ \Pi(X_i)} \Ind{Y_i >  q_1(1-\tau_n)} }
	= n P(Y_i(1) > q_1(1-\tau_n))
	= k,
	\]
	where at the second last equality we used Lemma~\ref{lemma:propensity_adjustment}, and
	\begin{align*}
		\frac{\Var(S_n)}{b_n^2} &= \frac{n}{k^2} \Var \left( \frac{D_i}{ \Pi(X_i)} \Ind{Z_i(1) > \tau_n^{-1}} \right) \\
		&= \frac{n}{k^2} \left[   \ERW{ \frac{D_i}{ \Pi(X_i)^2} \Ind{Z_i(1) > \tau_n^{-1}} }  -  \ERW{ \frac{D_i}{ \Pi(X_i)} \Ind{Z_i(1) > \tau_n^{-1}} }^2 \right] \\
		&< \frac{n}{k^2} \left[  \frac{1}{c} \tau_n -  \tau_n^2 \right] \to 0.
	\end{align*}
	Thus, by the weak law for triangular array (see Theorem 2.2.4 in \cite{durrett2013probability}), we have
	\[
	\frac{ S_n - k }{b_n} \tendsto{P} 0,
	\]
	or equivalently,
	\begin{equation}
		\frac{1}{k} \Sn  \frac{D_i}{ \Pi(X_i)} \Ind{Z_i(1) > \tau_n^{-1}} \tendsto{P} 1.
		\label{eq:hill_Gn2_lln}
	\end{equation}

	For (ii), let $b_n=k$ and $S_n=\Sn \log(Z_i(1) \tau_n ) \frac{D_i}{ \Pi(X_i)} \Ind{Z_i(1) > \tau_n^{-1}} $. By the fact that $\log(Z_i(1))$ has a standard exponential distribution, we have
	\begin{align*}
		\ERW{ S_n } &= n \ERW{ \log(Z_i(1) \tau_n ) \frac{D_i}{ \Pi(X_i)} \Ind{Z_i(1) > \tau_n^{-1}} } \\
		&= n \ERW{ \log(Z_i(1) \tau_n ) \Ind{Z_i(1) > \tau_n^{-1}} } \\
		&= n \int_{ \log(\tau_n^{-1}) }^\infty ( z + \log(\tau_n) ) e^{-z} \ dz \\
		&= k.
	\end{align*}
	Similarly, we have 
	\begin{align*}
		\frac{\Var(S_n)}{b_n^2} = \frac{n}{k^2} \Var \left( \log(Z_i(1) \tau_n ) \frac{D_i}{ \Pi(X_i)} \Ind{Z_i(1) > \tau_n^{-1}} \right)
		< \frac{n}{k^2} \left[ \frac{2}{c} \tau_n - \tau_n^2 \right] \to 0.
	\end{align*}
	Thus, by the weak law for triangular array, we have
	\begin{align}
		\frac{1}{k} \Sn
		\log(Z_i(1) \tau_n )
		\frac{D_i}{ \Pi(X_i)} \Ind{Z_i(1) > \tau_n^{-1}} 
		\tendsto{P} 1.
		\label{eq:hill_Gn1_term2}
	\end{align}
	This concludes that $G_n^1 \tendsto{P} \gamma_1$.

	Now we prove that $G_n^2, G_n^3, G_n^4$ converge to zero in probability.
	For $G_n^2$, let 
	\begin{align*}
		\Delta_n := \log(\wh q_1(1-\tau_n))  - \log(q_1(1-\tau_n))
		= 
		\log\left(\frac{\wh q_1(1-\tau_n)}{q_1(1-\tau_n)}\right),
	\end{align*}
	so
	\[
	|G_n^2|  = |\Delta_n|  
	\frac{1}{k} \Sn 
	\frac{D_i}{ \Pi(X_i)} \Ind{Y_i >  q_1(1-\tau_n)}.
	\] 
	Given the previous results \eqref{eq:hill_Gn2_lln} and the fact that $\Ind{Y_i(1) >  q_1(1-\tau_n)}=\Ind{Z_i(1) > \tau_n^{-1}}$ almost surely, it is sufficient to show that $\Delta_n = o_p(1)$.
	Consider
	\[
	\sqrt{k} \left( \frac{\wh q_1(1-\tau_n)}{q_1(1-\tau_n)} - 1 \right)
	= 
	\frac{\sqrt{k}q_1(1-\tau_n)^{-1}}{ \gamma_1 \lambda_{1,n} } \gamma_1 \lambda_{1,n} \left( \wh q_1(1-\tau_n) - q_1(1-\tau_n) \right),
	\]
	where $\lambda_{1,n} = \sqrt{\frac{n}{\tau_n}} f_1(q_1(1-\tau_n))$ is defined in Theorem~\ref{thm:interquantile}.
	By Lemma~\ref{lemma:interquantile_normalizing}, we have 
	\[
	\frac{\sqrt{k}q_1(1-\tau_n)^{-1}}{ \gamma_1 \lambda_{1,n} }  \to 1,
	\]
	and
	Theorem~\ref{thm:interquantile} implies that
	\begin{align*}
		\lambda_{1,n} (\wh{q}_1(1-\tau_n) - q_1(1-\tau_n))
		= O_p(1).
	\end{align*}
	Therefore,
	\[
	\sqrt{k} \left( \frac{\wh q_1(1-\tau_n)}{q_1(1-\tau_n)} - 1 \right) = O_p(1),
	\] 
	and consequently 
	\[
	\left( \frac{\wh q_1(1-\tau_n)}{q_1(1-\tau_n)} - 1 \right) \tendsto{P} 0.
	\]
	By the continuous mapping theorem, 
	\begin{equation}\label{equa: delta_n=op1}
		\Delta_n = o_p(1).
	\end{equation}
	Thus $G_n^2 \tendsto{P} 0$.
	
	For $G_n^3$, note that
	\begin{align*}
		0 \le & \left ( 
		\log(Y_i) - \log( \wh q_1(1-\tau_n))
		\right )
		\left (
		\Ind{Y_i > \wh q_1(1-\tau_n)} - \Ind{Y_i >  q_1(1-\tau_n)} 
		\right ) \\
		\le &
		\left ( 
		\log( q_1(1-\tau_n)) - \log( \wh q_1(1-\tau_n))
		\right )
		\left (
		\Ind{Y_i > \wh q_1(1-\tau_n)} - \Ind{Y_i >  q_1(1-\tau_n)} 
		\right ) \\
		= & \Delta_n
		\left (
		\Ind{Y_i > \wh q_1(1-\tau_n)} - \Ind{Y_i >  q_1(1-\tau_n)} 
		\right ),
	\end{align*}
	thus
	\begin{align*}
		G_n^3 
		&=  \frac{1}{k} \Sn 
		(\log(Y_i) - \log( \wh q_1(1-\tau_n))) 
		\frac{D_i}{  \wh \Pi(X_i)} 
		\left( 
		\Ind{Y_i > \wh q_1(1-\tau_n)} -\Ind{Y_i >  q_1(1-\tau_n)}
		\right) \\
		&\le  \Delta_n \frac{1}{k} \Sn 
		\frac{D_i}{  \wh \Pi(X_i)} 
		\left( 
		\Ind{Y_i > \wh q_1(1-\tau_n)} -\Ind{Y_i >  q_1(1-\tau_n)}
		\right) \\
		& = o_p(1).
	\end{align*}
	The last equality follows from Lemma~\ref{lemma:weighted_cdf_diff} and the result \eqref{equa: delta_n=op1} that $\Delta_n = o_p(1)$.
	Since $G_n^3 \ge 0$, we have $G_n^3 \tendsto{P} 0$.
	
	For $G_n^4$, we have that 
	\begin{align*}
		|G_n^4| 
		& \leq
		\left|
		\frac{1}{k} \Sn (\log(Y_i) - \log(  q_1(1-\tau_n))) 
		D_i \Ind{Y_i >  q_1(1-\tau_n)} 
		\left(
		\frac{1}{  \wh \Pi(X_i)}- \frac{1}{ \Pi(X_i)} 
		\right) 
		\right|\\
		& \quad +
		\left|
		\frac{1}{k} \Sn ( \log( q_1(1-\tau_n)) - \log( \wh q_1(1-\tau_n))) 
		D_i \Ind{Y_i >  q_1(1-\tau_n)} 
		\left(
		\frac{1}{  \wh \Pi(X_i)}- \frac{1}{ \Pi(X_i)} 
		\right)
		\right| \\
		& \leq 
		\sup_x \left| \frac{1}{\wh \Pi(x)}- \frac{1}{\Pi(x)} \right|
		\Bigg(
		\frac{1}{k} \Sn
		(\log(Y_i) - \log(  q_1(1-\tau_n))) 
		D_i \Ind{Y_i >  q_1(1-\tau_n)} \\
		& \quad + 
		|\Delta_n|
		\frac{1}{k} \Sn 
		D_i \Ind{Y_i >  q_1(1-\tau_n)} 
		\Bigg) \\
		& \le \sup_x \left| \frac{1}{\wh \Pi(x)}- \frac{1}{\Pi(x)} \right|
		\left( G_n^{1} + | G_n^{2} | \right).
	\end{align*}
	We have shown that $G_n^{1} \tendsto{P} \gamma_1$ 
	and $G_n^{2} = o_p(1)$, so $G_n^{4} = o_p(1)$ by Lemma~\ref{lemma:propensity_consistency}.
	
\end{proof}

\subsection{Proof of Theorem~\ref{thm:hill}} \label{appendix:ProofThm2}
To prove Theorem~\ref{thm:hill}, we first introduce Lemma~\ref{lemma:s_order}, which shows that the term $\ERW{ S_{i,j,n}^p }$ is of order $O(\tau_n)$.
\begin{lemma} 
	For
	\[
	S_{i,j,n} := \gamma_j \log \left( \frac{\tau_n}{1-F_j(Y_i(j))} \right)  \Ind{ Y_i(j) > q_j(1-\tau_n)}
	\]
	defined in Theorem~\ref{thm:hill} and for all $p \in \mathbb{N}$,  we have 
	\[
	\ERW{ S_{i,j,n}^p } = O(\tau_n)
	\]
	and
	\[
	\ERW{ S_{i,j,n}^p | X_i } = O_p(\tau_n).
	\]
	\label{lemma:s_order}
\end{lemma}

\begin{proof}[Proof of Lemma~\ref{lemma:s_order}]
	
	We have already seen in the proof of Lemma~\ref{lemma:hill_consistency} that 
	$\log \left( \frac{1}{1-F_j(Y_i(j))} \right)$ follows a standard exponential distribution.
	Therefore, for $p \in \mathbb{N}$, we have
	\begin{align*}
		\ERW{S_{i,j,n}^p}
		& = \gamma_j^p \int_{-\log(\tau_n)}^\infty (z + \log(\tau_n))^p e^{-z} \ dz
		\\
		& = \tau_n \gamma_j^p \int_{0}^\infty z^p e^{-z} \ dz
		\\
		& = \tau_n \gamma_j^p p! < \infty.
	\end{align*}
	Thus the first claim follows.
	Note that $S_{i,j,n} \ge 0$ for positive $\gamma_1$, thus the second claim follows from the Markov inequality.
\end{proof}

Now we prove Theorem~\ref{thm:hill}.
\begin{proof}[Proof of Theorem~\ref{thm:hill}] 
	We show the claim for $j=1$, and the case $j=0$ can be proved analogously.
	As in the proof of Lemma~\ref{lemma:hill_consistency}, we expend
	\[
	\wh \gamma_1^H = 
	G_n^1 + G_n^2 + G_n^3 + G_n^4
	\]
	where 
	\begin{align*}
		G_n^1  &:= \frac{1}{k} \Sn (\log(Y_i) - \log( q_1(1-\tau_n))) 
		\frac{D_i}{ \Pi(X_i)} \Ind{Y_i >  q_1(1-\tau_n)} \\
		G_n^2 &:=  (\log(q_1(1-\tau_n)) - \log( \wh q_1(1-\tau_n)) )
		\frac{1}{k} \Sn 
		\frac{D_i}{ \Pi(X_i)} \Ind{Y_i >  q_1(1-\tau_n)} 
		\\
		G_n^3 &:=  \frac{1}{k} \Sn (\log(Y_i) - \log( \wh q_1(1-\tau_n))) 
		\frac{D_i}{  \wh \Pi(X_i)} \left( \Ind{Y_i > \wh q_1(1-\tau_n)} -
		\Ind{Y_i >  q_1(1-\tau_n)} 
		\right) \\
		G_n^4 &:= \frac{1}{k} \Sn (\log(Y_i) - \log( \wh q_1(1-\tau_n))) 
		D_i \Ind{Y_i >  q_1(1-\tau_n)} 
		\left(
		\frac{1}{  \wh \Pi(X_i)}- \frac{1}{ \Pi(X_i)} 
		\right). 
	\end{align*}
	In the following, we will show that
	\begin{align*}
		\sqrt{k} G_n^1 &= \frac{\lambda_1}{1-\rho_1} + \frac{1}{\sqrt{k}} \Sn  \frac{D_i}{ \Pi(X_i)} S_{i,1,n} + o_p(1) \\
		\sqrt{k} G_n^2 &=  -\frac{1}{\sqrt{n}} \Sn \gamma_1 \phi_{i,1,n}  +  o_{p}(1) \\  
		\sqrt{k} G_n^3 &= o_p(1) \\
		\sqrt{k} G_n^4 &= -  \frac{1}{\sqrt{k}} \Sn \frac{\ERW{S_{i,1,n} \mid X_i}}{\Pi(X_i)} (D_i - \Pi(X_i)) + o_p(1),
	\end{align*}
	which then implies the original claim that
	\[
	\sqrt{k} (\wh \gamma^H_1 - \gamma_1) = \frac{\lambda_1}{1-\rho_1} + \frac{1}{\sqrt{n}} \Sn \left( \psi_{i,1,n} -\gamma_j\phi_{i,1,n} \right) + o_p(1).
	\]

	For $G_n^1$, we proceed similarly as in the proof of Theorem 3.2.5 in \cite{de2007extreme}.
	As shown in the proof of \mbox{Lemma~\ref{lemma:hill_consistency}}, we have that almost surely,
	\[
	G_n^1 = \frac{1}{k} \Sn\left( \log(U_1(Z_i(1))) - \log\left( U_1 (n/k)\right) \right) 
	\frac{D_i}{ \Pi(X_i)} \Ind{Z_i(1) >  n/k},
	\]
	where $Z_i(1) = 1/(1-F_1(Y_i(1)))$ and $U_1 = (1/(1-F_1))^\leftarrow$.
	By Assumption \ref{ass:second_order_regular_var}, we have for all $x >0$,
	\[
	\lim_{t \to \infty} \frac{x^{-\gamma_1} \frac{U_1(tx)}{U_1(t)} - 1}{A_1(t)}
	= \frac{x^{\rho_1} - 1}{\rho_1},
	\]
	with $\gamma_1 > 0$, $\rho_1 < 0$ and $\lim_{t \to \infty} A_1(t)=0$.
	Equivalently, we have
	\[
	\lim_{t \to \infty}
	\frac{ \log U_1(tx) - \log U_1(t) - \gamma_1  \log(x)}{A_1(t)}
	= \frac{x^{\rho_1}-1}{\rho_1}.
	\]
	By the proof of Theorem 3.2.5 in \cite{de2007extreme},
	there exists a function $A$ such that $\lim_{t \to \infty} A(t) / A_1(t) = 1$
	and for any $\epsilon > 0$,
	there exists $t_0 > 0$ 
	such that for all
	$t \ge t_0$, $x \ge 1$,
	\begin{equation*}
		\left| \frac{ \log U_1(tx) - \log U_1(t) - \gamma_1  \log(x)}{A(t)} 
		- \frac{x^{\rho_1}-1}{\rho_1}
		\right| \le \epsilon x^{\rho_1 + \epsilon}.
		\label{eq:log_uniform_second}
	\end{equation*}
	
	For large enough $n$ 
	and for $i \in \{1,\dots,n\}$ such that $Z_i(1) > n/k$, we can set $t = n/k $ and $x = Z_i(1) \cdot k/n$. Multiplying by $\sqrt{k} \frac{D_i}{k \Pi(X_i)}$ on both sides of the above inequality and summing up all $i \in \{1,\dots,n\}$ with $Z_i(1) > n/k$ gives us that almost surely, 
	\begin{equation}
		-\epsilon \sqrt{k}G_n^{1,3} \le \sqrt{k}G_n^1 - \sqrt{k}G_n^{1,1} - \sqrt{k}G_n^{1,2}  \le \epsilon \sqrt{k}G_n^{1,3},
		\label{Gn1_inequality}
	\end{equation}
	where 
	\begin{align*}
		G_n^{1,1} &:= \frac{\gamma_1}{k} \Sn \log\left(Z_i(1) \frac{k}{n}\right)
		\frac{D_i}{ \Pi(X_i)} \Ind{Z_i(1) >  n/k}  \\
		G_n^{1,2} &:=  A \nk \frac{1}{k} \Sn \frac{\left(Z_i(1) \frac{k}{n}\right)^{\rho_1} - 1}{\rho_1}
		\frac{D_i}{ \Pi(X_i)} \Ind{Z_i(1) >  n/k} \\
		G_n^{1,3} &:=   A \nk \frac{1}{k} \Sn \left(Z_i(1) \frac{k}{n}\right)^{\rho_1+\epsilon}
		\frac{D_i}{ \Pi(X_i)} \Ind{Z_i(1) >  n/k}. 
	\end{align*}
	
	Let $0 < \epsilon < -\rho_1$. We first show that $\sqrt{k} G_n^{1,3}$ converge in probability to some constant. Since $\sqrt{k} A_1 \nk \to \lambda_1$  by assumption and $A \nk / A_1 \nk \to 1$, it is enough to show that
	\[
	\frac{1}{k} \Sn \left(Z_i(1) \frac{k}{n}\right)^{\rho_1+\epsilon} \frac{D_i}{ \Pi(X_i)} \Ind{Z_i(1) >  n/k}
	\]
	converge in probability to some constant.
	
	Let $b_n=k$ and $S_n = \Sn \left(Z_i(1) \frac{k}{n}\right)^{\rho_1+\epsilon} \frac{D_i}{ \Pi(X_i)} \Ind{Z_i(1) >  n/k}$. By the fact that
	$Z_i(1)$ has probability density function $1/z^2$ on $z\ge 1$, we can calculate
	\[
	\ERW{S_n} 
	= \frac{k}{1- \rho_1 - \epsilon}
	\]
	and
	\[
	\frac{\Var(S_n)}{b_n^2} < \frac{n}{k^2} [ \frac{1}{c(1-2(\rho_1+\epsilon))} \cdot \frac{k}{n} - \frac{k^2}{n^2(1- \rho_1 - \epsilon)^2}] \to 0.
	\]
	Thus, by the weak law for triangular array, we have
	\[
	\frac{1}{k} \Sn \left(Z_i(1) \frac{k}{n}\right)^{\rho_1+\epsilon} \frac{D_i}{ \Pi(X_i)} \Ind{Z_i(1) >  n/k} \tendsto{P}  \frac{1}{1-\rho_1-\epsilon},
	\]
	which implies
	\[
	\sqrt{k} G_n^{1,3} \tendsto{P} \frac{\lambda_1}{1-\rho_1-\epsilon}.
	\]

	Because $\epsilon$ can be arbitrarily close to zero, by inequality \eqref{Gn1_inequality}, we have that almost surely,
	\[
	\sqrt{k}G_n^1 = \sqrt{k}G_n^{1,1} + \sqrt{k}G_n^{1,2} + o_p(1).
	\]
	
	Similarly, by using the weak law for triangular array, one can obtain that 
	\[
	\sqrt{k}G_n^{1,2} \tendsto{P} \frac{\lambda_1}{1-\rho_1}.
	\]
	Hence, we conclude that almost surely,
	\begin{equation}
		\sqrt{k} G_n^1 
		=
		\frac{\lambda_1}{1-\rho_1} + \sqrt{k} G_n^{1,1} + o_p(1)
		= \frac{\lambda_1}{1-\rho_1} + \frac{1}{\sqrt{k}} \Sn  \frac{D_i}{ \Pi(X_i)} S_{i,1,n} + o_p(1).
		\label{eq:hill_Gn1_influence}
	\end{equation}
	
	For $G_n^2$, similarly as in the proof of Lemma~\ref{lemma:hill_consistency}, we have
	
	\begin{align*}
		\sqrt{k} \left( \frac{\wh q_1(1-\tau_n)}{q_1(1-\tau_n)} - 1 \right)
		= 
		\gamma_1 \lambda_{1,n} (\wh{q}_1(1-\tau_n) - q_1(1-\tau_n))
		= 
		\frac{1}{\sqrt{n}} \Sn \gamma_1 \phi_{i,1,n}  +  o_{p}(1),
	\end{align*}
	where $\lambda_{1,n}$ and $\phi_{i,1,n}$  are defined as in Theorem~\ref{thm:interquantile}, and we applied Theorem~\ref{thm:interquantile} to obtain the last equality.
	By applying the delta method, we have
	\begin{equation}
		\sqrt{k} (\log (\wh q_1(1-\tau_n)) - \log (q_1(1-\tau_n)) )
		= 
		\frac{1}{\sqrt{n}} \Sn \gamma_1 \phi_{i,1,n}  +  o_{p}(1),
		\label{eq:hill_log_distribution}
	\end{equation}
	Combining with result \eqref{eq:hill_Gn2_lln}, we obtain
	\begin{equation}
		\sqrt{k} G_n^2 =  -
		\frac{1}{\sqrt{n}} \Sn \gamma_1 \phi_{i,1,n}  +  o_{p}(1).
		\label{eq:hill_Gn2_influence}
	\end{equation}
	Note that by Theorem~\ref{thm:interquantile}, we have $\sqrt{k} G_n^2 = O_{p}(1)$.
	
	For $G_n^3$, similarly as in the proof of Lemma~\ref{lemma:hill_consistency}, we have
	\[
	0 \le
	\sqrt{k} G_n^3 \le  
	\sqrt{k} ( \log(q_1(1-\tau_n)) - \log(\wh q_1(1-\tau_n)) )
	\frac{1}{k} \Sn \frac{D_i}{  \wh \Pi(X_i)} 
	\left( 
	\Ind{Y_i > \wh q_1(1-\tau_n)} -\Ind{Y_i >  q_1(1-\tau_n)}
	\right).
	\]
	Theorem~\ref{thm:interquantile} and the result \eqref{eq:hill_log_distribution} then imply that
	\[
	\sqrt{k}( \log(q_1(1-\tau_n)) - \log(\wh q_1(1-\tau_n))) = O_p(1).
	\]
	Thus, by Lemma~\ref{lemma:weighted_cdf_diff} we have 
	\begin{equation}
		\sqrt{k} G_n^3 = o_p(1). 
		\label{eq:gn3_to_0}
	\end{equation}
	
	For $G_n^4$, we expand 
	\begin{align*}
		\sqrt{k} G_n^4 =& 
		\sqrt{k} \frac{1}{k} \Sn (\log(Y_i) - \log( q_1(1-\tau_n))) D_i \Ind{Y_i >  q_1(1-\tau_n)} 
		\left( \frac{1}{ \wh \Pi(X_i)} - \frac{ 1}{\Pi(X_i)} \right) \\
		& + \sqrt{k} (\log( q_1(1-\tau_n)) - \log( \wh q_1(1-\tau_n)) )
		\frac{1}{k} \Sn 
		D_i \Ind{Y_i >  q_1(1-\tau_n)} 
		\left( \frac{1}{ \wh \Pi(X_i)} - \frac{ 1}{\Pi(X_i)} \right).
	\end{align*}
	For the second term, we have
	\begin{align*}
		& \left|\sqrt{k} (\log(q_1(1-\tau_n)) - \log( \wh q_1(1-\tau_n)) ) 
		\frac{1}{k} \Sn 
		D_i \Ind{Y_i >  q_1(1-\tau_n)} 
		\left( \frac{1}{ \wh \Pi(X_i)} - \frac{ 1}{\Pi(X_i)} \right) \right| \\
		\le &  \sup_{x} \left| 
		\frac{1}{ \wh \Pi(x)} - \frac{ 1}{\Pi(x)}  \right|
		|\sqrt{k} \log(\wh q_1(1-\tau_n)) - \log( q_1(1-\tau_n))| 
		\frac{1}{k} \Sn 
		D_i \Ind{Y_i >  q_1(1-\tau_n)} 
		\\
		= &  \sup_{x} \left| 
		\frac{1}{ \wh \Pi(x)} - \frac{ 1}{\Pi(x)}  \right|
		|\sqrt{k} G_n^2 | \\
		= & o_{p}(1)
	\end{align*}
	by that fact that $\sqrt{k} G_n^2 = O_{p}(1)$ (see the comment below \eqref{eq:hill_Gn2_influence}) and Lemma~\ref{lemma:propensity_consistency}.
	Hence, we obtain 
	\[
	\sqrt{k} G_n^4  = 
	\frac{1}{\sqrt{k}} \Sn (\log(Y_i) - \log( q_1(1-\tau_n))) D_i \Ind{Y_i >  q_1(1-\tau_n)} 
	\left( \frac{1}{ \wh \Pi(X_i)} - \frac{ 1}{\Pi(X_i)} \right) 
	+ o_p(1).
	\]
	
	Denote 
	\[
	\wt G_n^4  := 
	\frac{\gamma_1}{k} \Sn \log(Z_i(1) \tau_n) D_i \Ind{Y_i >  q_1(1-\tau_n)} 
	\left( \frac{1}{ \wh \Pi(X_i)} - \frac{ 1}{\Pi(X_i)} \right),
	\]
	then we have
	\begin{align*}
		& \left|
		\frac{1}{\sqrt{k}} \Sn (\log(Y_i) - \log( q_1(1-\tau_n))) D_i \Ind{Y_i >  q_1(1-\tau_n)} 
		\left( \frac{1}{ \wh \Pi(X_i)} - \frac{ 1}{\Pi(X_i)} \right) 
		- \sqrt{k} \wt G_n^4  \right|
		\\ 
		\le & 
		\sup_x \left| 
		\frac{1}{ \wh \Pi(x)} - \frac{ 1}{\Pi(x)}  \right|
		\frac{1}{\sqrt{k}} \Sn \left| \log(Y_i) - \log( q_1(1-\tau_n)) - \gamma_1 \log(Z_i(1) \tau_n) \right| D_i \Ind{Y_i >  q_1(1-\tau_n)} 
		\\ 
		\le & 
		\sup_x \left| 
		\frac{1}{ \wh \Pi(x)} - \frac{ 1}{\Pi(x)}  \right|
		\frac{1}{\sqrt{k}} \Sn  \Bigg( \left| \log(Y_i) - \log( q_1(1-\tau_n)) - \gamma_1 \log(Z_i(1) \tau_n ) - A \nk\frac{\left(Z_i(1) \frac{k}{n}\right)^{\rho_1} - 1}{\rho_1} \right|
		\\ 
		& \qquad + \left| A \nk \right| \frac{\left(Z_i(1) \frac{k}{n}\right)^{\rho_1} - 1}{\rho_1} \Bigg)
		\frac{D_i}{\Pi(X_i)} \Ind{Y_i >  q_1(1-\tau_n)} 
		\\
		= & 
		\sup_x \left| 
		\frac{1}{ \wh \Pi(x)} - \frac{ 1}{\Pi(x)}  \right|
		\left( \sqrt{k} |G_n^{1,2}| + o_p(1) \right)
		= o_p(1),
	\end{align*}
	where for the second last equality we used that
	\begin{align*}
		&\frac{1}{\sqrt{k}} \Sn 
		\left| \log(Y_i) - \log( q_1(1-\tau_n)) - \gamma_1 \log(Z_i(1) \tau_n ) - A \nk\frac{\left(Z_i(1) \frac{k}{n}\right)^{\rho_1} - 1}{\rho_1} \right|
		\frac{D_i}{\Pi(X_i)} \Ind{Y_i >  q_1(1-\tau_n)} \\ &= o_p(1) 
	\end{align*}
	which can be shown by using a similar argument as on page 43.
	Thus, we have
	\[
	\sqrt{k} G_n^4 = \sqrt{k} \wt G_n^4 + o_p(1) = 
	\frac{1}{\sqrt{k}} \Sn D_i S_{i,1,n} 
	\left( \frac{1}{ \wh \Pi(X_i)} - \frac{ 1}{\Pi(X_i)} \right) + o_p(1)
	\]
	where $S_{i,1,n} = 
	\gamma_1 \log(Z_i(1) \tau_n)\Ind{Y_i >  q_1(1-\tau_n)} $.
	
	In order to derive the influence function 
	which arises from using the estimated propensity score,
	we follow similar steps as in the proof of Theorem 3.1 of \cite{zhang2018extremal}.
	First, we rewrite $\sqrt{k} G_n^{4} = G_n^{4,1} - G_n^{4,2} + o_p(1)$ with 
	\begin{align*}
		G_n^{4,1} &:=  \frac{1}{\sqrt{k}} \Sn D_i S_{i,1,n} 
		\frac{( \wh \Pi(X_i) - \Pi(X_i))^2}{\wh \Pi(X_i) \Pi(X_i)^2},
		\\
		G_n^{4,2} &:= 
		\frac{1}{\sqrt{k}} \Sn D_i S_{i,1,n}
		\frac{ \wh \Pi(X_i) - \Pi(X_i)}{\Pi(X_i)^2}.
	\end{align*}
	
	
	For $G_n^{4,1}$, note that $S_{i,1,n} \ge 0$, so we have
	\begin{align*}
		0 \le G_n^{4,1} &\le \frac{1}{c^2} \sup_x | \wh \Pi(x) - \Pi(x)|^2 \sup_x \left| \frac{1}{ \wh \Pi(x)} \right| \frac{1}{\sqrt{k}} \Sn S_{i,1,n} \\
		& = o_p(k^{-1/2}) O_p(1) \frac{1}{\sqrt{k}} O_p(k)  = o_p(1),
	\end{align*}
	where in the second inequality we used Assumption \ref{ass:1} $iii.)$ and $D_i \le 1$, 
	and in the second last equality we used Lemma~\ref{lemma:propensity_consistency},
	result \eqref{eq:prob_bounded_prop} and
	$\Sn S_{i,1,n} = O_p(k)$ which can be obtained by the Markov inequality and Lemma~\ref{lemma:s_order}.
	Thus $\sqrt{k} G_n^{4} = - G_n^{4,2} + o_p(1)$.
	
	For $G_n^{4,2}$, we expand \mbox{$G_n^{4,2} = G_n^{4,3} + G_n^{4,4}$} where
	\begin{align*}
		G_n^{4,3} & := \frac{n}{\sqrt{k}} 
		\int_{\supp(X)} \frac{1}{\Pi(x)} ( \wh \Pi(x) - \Pi(x) ) 
		\ERW{S_{1,1,n} | x} dF_X(x)
		\\ G_n^{4,4} & := 
		\frac{1}{\sqrt{k}} \Sn \left( D_i S_{i,1,n} 
		\frac{ \wh \Pi(X_i) - \Pi(X_i)}{\Pi(X_i)^2} - 
		\int_{\supp(X)} \frac{1}{\Pi(x)} ( \wh \Pi(x) - \Pi(x) ) 
		\ERW{S_{i,1,n} | x} dF_X(x) \right)
	\end{align*}
	and $F_X$ denotes the CDF of $X$. 
	
	First, we show $G_n^{4,4} = o_p(1)$.
	For this we consider
	\[
	\pi_n := \argmin_{\pi \in \R^{h_n}} 
	\ERW{ \Pi(X) \log(L(H_{h_n}(X)^T \pi)) + 
		(1-\Pi(X)) \log(1-L(H_{h_n}(X)^T \pi))}
	\]
	and the pseudo true propensity score
	$\Pi_n(x) = L(H_{h_n}(x)^T \pi _n)$, where
	$H_{h_n}$ is the vector consisting of $h_n$ sieve basis
	functions and $L$ is the sigmoid function (see Section \ref{appendix:sieve}
	for more details).
	We rewrite $G_n^{4,4} = G_n^{4,5} + G_n^{4,6}$ with
	\begin{align*}
		G_n^{4,5} & := 
		\frac{1}{\sqrt{k}} \Sn \left( D_i S_{i,1,n} 
		\frac{ \wh \Pi(X_i) - \Pi_n(X_i)}{\Pi(X_i)^2} - 
		\int_{\supp(X)} \frac{1}{\Pi(x)} ( \wh \Pi(x) - \Pi_n(x) ) 
		\ERW{S_{i,1,n} | x} dF_X(x) \right),
		\\
		G_n^{4,6} & := 
		\frac{1}{\sqrt{k}} \Sn \left( D_i S_{i,1,n} 
		\frac{ \Pi_n(X_i) - \Pi(X_i)}{\Pi(X_i)^2} - 
		\int_{\supp(X)} \frac{1}{\Pi(x)} (  \Pi_n(x) - \Pi(x) ) 
		\ERW{S_{i,1,n} | x} dF_X(x) \right),
	\end{align*}
	and we show that both terms converge to $0$ in probability.
	For $G_n^{4,6}$, note that
	\[
	\ERW{ D_i S_{i,1,n} 
		\frac{ \Pi_n(X_i) - \Pi(X_i)}{\Pi(X_i)^2}} = 
	\int_{\supp(X)} \frac{1}{\Pi(x)} (  \Pi_n(x) - \Pi(x) ) 
	\ERW{S_{i,1,n} | x} dF_X(x), 
	\]
	and thus $\ERW{G_n^{4,6}}=0$. We also have 
	\begin{align*}
		\Var (G_n^{4,6})
		& =
		\frac{1}{k} \Sn \Var \bigg( D_i S_{i,1,n} 
		\frac{ \Pi_n(X_i) - \Pi(X_i)}{\Pi(X_i)^2} \bigg)
		\\ & \le
		\frac{n}{k} \ERW{ D_i S_{i,1,n}^2 \left( \frac{\Pi_n(X_i) - \Pi(X_i)}{ \Pi(X_i)^2 } \right)^2
		}
		\\ & \le \frac{n}{k} \frac{\sup_x | \Pi_n(x) - \Pi(x) |^2}{c^4} \ERW{ S_{i,1,n}^2}
		\\ & = O(\zeta(h_n)^2 h_n^{-s/r} ) \to 0.
	\end{align*}
	where for the last equality we applied Lemma~\ref{lemma:s_order} which gives $\ERW{ S_{i,1,n}^2} = O( \tau_n )$ and the Lemma 1 of \cite{hirano2003efficient} which gives $\sup_x | \Pi_n(x) - \Pi(x) | = O(\zeta(h_n) h_n^{-s/2r})$ with $\zeta(h_n) = \sup_{x} \Vert H_{h_n}(x) \Vert$ under Assumption \ref{ass:sieve}. The convergence to $0$ can be obtained because
	Assumption \ref{ass:sieve} $iv)$ $n \tau_n \zeta(h_n)^6 h_n^{-s/r} \to 0$ implies
	$ \zeta(h_n)^2 h_n^{-s/r} = o((n \tau_n)^{-1}) = o(k^{-1}) $ as $\zeta(h_n) \ge 1$.
	Therefore we have $G_n^{4,6} = o_p(1)$.
	
	For $G_n^{4,5}$, we use the Taylor expansion to get 
	\[G_n^{4,5} = G_n^{4,5,1}( \wh \pi_n - \pi_n)+ \frac{1}{2}( \wh \pi_n - \pi_n)^T(G_n^{4,5,2} - G_n^{4,5,3})( \wh \pi_n - \pi_n)\]
	with
	\begin{align*} 
		G_n^{4,5,1} & := \frac{1}{\sqrt{k}} \Sn 
		\bigg( \frac{D_i S_{i,1,n} }{\Pi(X_i)^2} L'(H_{h_n}(X_i)^T \pi_n) H_{h_n}(X_i)^T \\
		& \qquad \qquad
		- \int_{\supp(X)} \frac{\ERW{S_{i,1,n}|x} }{\Pi(x)} L'(H_{h_n}(x) \pi_n) H_{h_n}(x)^T dF_X(x) \bigg)
		\\
		G_n^{4,5,2} & := \frac{1}{\sqrt{k}} \Sn
		\frac{D_i S_{i,1,n} }{\Pi(X_i)^2}L''(H_{h_n}(X_i)^T \wt{\pi}_n) H_{h_n}(X_i) H_{h_n}(X_i)^T
		\\
		G_n^{4,5,3} & := \frac{1}{\sqrt{k}} \Sn \int_{\supp(X)} 
		\frac{\ERW{S_{i,1,n}|x} }{\Pi(x)}L''(H_{h_n}(x)^T \wt {\pi}_n) H_{h_n}(x) H_{h_n}(x)^T \ dF_X(x).
	\end{align*}
	where $\wt \pi_n$ is random, lies on the line
	between $\pi_n$ and $\wh \pi_n$, and depends on $X_i$ resp. $x$.
	For $G_n^{4,5,1}$, we have
	\begin{align*}
		\ERW{ \Vert G_n^{4,5,1}\Vert^2 }
		&\le \frac{n}{k} \ERW{ \left\Vert \frac{D_i S_{i,1,n} }{\Pi(X_i)^2} L'(H_{h_n}(X_i)^T \pi_n) H_{h_n}(X_i)^T \right\Vert^2 } \\
		&< \frac{n}{k} \frac{1}{c^4} \zeta(h_n)^2 \ERW{ S_{i,1,n}^2}
		= O(\zeta(h_n)^2),
	\end{align*}
	where the first inequality holds because the summands of $G_n^{4,5,1}$ 
	are i.i.d. and with mean $0$, the second inequality holds because 
	$|L'| < 1$, $1/\Pi(X_i)^4 < 1/c^4$ and $\zeta(h_n) = \sup_x \Vert H_{h_n}(x) \Vert$, and 
	for the last equality we used Lemma~\ref{lemma:s_order}. Thus $\ERW{ \Vert G_n^{4,5,1}\Vert } \le \left(\ERW{ \Vert G_n^{4,5,1} \Vert^2} \right)^{1/2}
	= O(\zeta(h_n))$.
	
	Note that the summands of $G_n^{4,5,2}$ and $G_n^{4,5,3}$ are no longer independent because of $\wt \pi_n$. Therefore, for $G_n^{4,5,2}$ and $G_n^{4,5,3}$, we apply the triangle inequality to obtain
	\begin{align*}
		\ERW{ \Vert G_n^{4,5,2} \Vert}
		&\le \frac{n}{\sqrt{k}} \frac{1}{c^2} \zeta(h_n)^2  
		\ERW{ S_{i,1,n}} = O(\sqrt{k} \zeta(h_n)^2), \\
		\ERW{ \Vert G_n^{4,5,3} \Vert}
		&\le \frac{n}{\sqrt{k}} \frac{1}{c^2} \zeta(h_n)^2  
		\ERW{ S_{i,1,n}} = O(\sqrt{k} \zeta(h_n)^2),
	\end{align*}
	where we used similar arguments as for $G_n^{4,5,1}$ and the fact that $| L'' | < 1$.
	By the Markov inequality, we have
	\begin{align*}
		\Vert G_n^{4,5,1}\Vert = O_p(\zeta(h_n))
		\quad \text{and} \quad
		\Vert G_n^{4,5,2}\Vert = \Vert G_n^{4,5,3}\Vert =O_p(\sqrt{k} \zeta(h_n)^2).
	\end{align*}
	
	Under the Assumption \ref{ass:sieve},
	we have by the Lemma 2 in \cite{hirano2003efficient} that
	$\Vert \wh \pi_n - \pi_n \Vert = O_p(\sqrt{h_n/n})$. 
	Hence, by the Cauchy–Schwarz inequality, we have
	\begin{align*} 
		|G_n^{4,5}| 
		&\le \Vert G_n^{4,5,1} \Vert \Vert  \wh \pi_n - \pi_n\Vert +
		\frac{1}{2}\Vert \wh \pi_n - \pi_n\Vert^2 (\Vert G_n^{4,5,2}\Vert  +\Vert  G_n^{4,5,3} \Vert) \\
		&= O_p( \zeta(h_n) \sqrt{\frac{h_n}{n}}) + O_p( \sqrt{k} \zeta(h_n)^2 \frac{h_n}{n})
	\end{align*}
	By Assumption \ref{ass:sieve} iv) $\frac{1}{\sqrt{n}} \zeta(h_n)^2 h_n \to 0$, we have $\zeta(h_n) \sqrt{\frac{h_n}{n}} \to 0$ and $ \sqrt{k} \zeta(h_n)^2  \frac{h_n}{n} \to 0$, thus $G_n^{4,5} = o_p(1)$, which concludes that $G_n^{4,4} = o_p(1)$.
	
	At this point, we have $\sqrt{k} G_n^{4} = - G_n^{4,3} + o_p(1)$. 
	For $G_n^{4,3}$, we proceed in a similar manner as for $G_n^{4,4}$.
	First, we decompose $G_n^{4,3} = G_n^{4,7} + G_n^{4,8}$
	with
	\begin{align*}
		G_n^{4,7} & := \frac{n}{\sqrt{k}} \int_{\supp(X)}
		\frac{1}{\Pi(x)} ( \wh \Pi(x) - \Pi_n(x) ) 
		\ERW{S_{1,1,n} | x} dF_X(x),
		\\
		G_n^{4,8} & := \frac{n}{\sqrt{k}} \int_{\supp(X)}
		\frac{1}{\Pi(x)} (  \Pi_n(x) - \Pi(x) ) 
		\ERW{S_{1,1,n} | x} dF_X(x).
	\end{align*}
	For $G_n^{4,8}$, we have
	\[
	|G_n^{4,8}| 
	\le \frac{n}{\sqrt{k}} \frac{1}{c} \sup_x | \Pi_n(x) - \Pi(x) | \ERW{S_{i,1,n}} 
	=O(\sqrt{k} \zeta(h_n) h_n^{-s/2r}) = o(1),
	\]
	where the second last equality holds because $\sup_x | \Pi_n(x) - \Pi(x) | = O(\zeta(h_n) h_n^{-s/2r})$ and $\ERW{S_{i,1,n}} = O(\tau_n)$, 
	and the last equality holds because Assumption \ref{ass:sieve} $iv)$ $n \tau_n \zeta(h_n)^6 h_n^{-s/r} \to 0$ implies
	$ \zeta(h_n)^2 h_n^{-s/r} = o(k^{-1}) $.
	
	Hence, we have $\sqrt{k} G_n^{4} = - G_n^{4,7} + o_p(1)$.
	For $G_n^{4,7}$, we use the mean value theorem for $\wh \Pi(x) - \Pi_n(x)$ to get
	\[
	G_n^{4,7} = \frac{n}{\sqrt{k}} \int_{\supp(X)}
	\frac{\ERW{S_{1,1,n} | x}}{\Pi(x)} 
	L'(H_{h_n}(x)^T \wt \pi_n) H_{h_n}(x)^T 
	dF_X(x) \ (\wh \pi_n - \pi_n)
	\]
	where $\wt \pi_n$ lies on the line between
	$\pi_n$ and $\wh \pi_n$. Since
	$\wh \pi_n$ is the solution of the optimization problem
	\eqref{eq:pihat}, the first order condition 
	\[
	0 = \frac{1}{n} \Sn (D_i - \wh \Pi(X_i)) H_{h_n}(X_i)
	\] 
	can be derived by differentiation of the objective function. Extending
	\[ 
	0=\frac{1}{n} \Sn (D_i - \wh \Pi(X_i)) H_{h_n}(X_i) = 
	\frac{1}{n} \Sn (D_i -  \Pi_n(X_i)) H_{h_n}(X_i) 
	- 
	\frac{1}{n} \Sn (\wh \Pi(X_i) - \Pi_n(X_i)) H_{h_n}(X_i)  
	\]
	and applying the mean value theorem for $\wh \Pi(X_i) - \Pi_n(X_i)$ leads to  
	\[ \wh \pi_n - \pi_n = \frac{1}{n} \Sn \wt \Sigma_n^{-1} (D_i - \Pi_n(X_i)) H_{h_n}(X_i) \]
	where
	\[
	\wt \Sigma_n = \frac{1}{n} \Sn L'(H_{h_n}(X_i)^T \wt \pi_n) H_{h_n}(X_i) H_{h_n}(X_i)^T.
	\]
	We define
	\begin{align*}
		\wt \Psi_{h_n} & := \sqrt{\frac{n}{k}} \int _{\supp(X)}
		\frac{\ERW{S_{i,1,n} | x}}{\Pi(x)} 
		L'(H_{h_n}(x)^T \wt \pi_n) H_{h_n}(x) dF_X(x) 
		\\
		\Psi_{h_n} & := \sqrt{\frac{n}{k}} \int _{\supp(X)}
		\frac{\ERW{S_{i,1,n} | x}}{\Pi(x)} 
		L'(H_{h_n}(x)^T \pi_n) H_{h_n}(x) dF_X(x) 
		\\
		\Sigma_n &:= \ERW{ H_{h_n}(X) H_{h_n}(X)^T L'(H_{h_n}(X)^T \pi_n)}
		\\
		V_n & := \frac{1}{\sqrt{n}} \Sn H_{h_n}(X_i) (D_i - \Pi_n(X_i)),
	\end{align*}
	which allows us to write
	\begin{align*}
		G_n^{4,7} &= 
		\wt \Psi_{h_n}^T \wt \Sigma_n^{-1} V_n
		= 
		\Psi_{h_n}^T \Sigma_n^{-1} V_n
		+ \underbrace{(\wt \Psi_{h_n}^T- \Psi_{h_n}^T)   \wt \Sigma_n^{-1} V_n}_{=: G_n^{4,7,1}}
		+ 
		\underbrace{\Psi_{h_n}^T (\wt \Sigma_n^{-1}- \Sigma_n^{-1}  )V_n}_{=: G_n^{4,7,2}}.
	\end{align*}
	Now we show that both terms $G_n^{4,7,1}$ and $G_n^{4,7,2}$ are $o_p(1)$.
	
	For $G_n^{4,7,1}$,
	using the mean value theorem for $L'(H_{h_n}(x)^T \wt \pi_n)-L'(H_{h_n}(x)^T \pi_n) $ and the fact that $| L'' | < 1$, we have
	\[
	\Vert \wt \Psi_{h_n}- \Psi_{h_n} \Vert
	\le 
	\frac{1}{c} \sqrt{\frac{n}{k}} \zeta(h_n)^2 \ERW{S_{i,1,n}} \Vert \wh \pi_n - \pi_n \Vert
	= O(\sqrt{\tau_n} \zeta(h_n)^2 \sqrt{\frac{h_n}{n}}),
	\]
	where in the last equality we used $\Vert \wh \pi_n - \pi_n \Vert = O_p(\sqrt{h_n/n}) $
	and $\ERW{S_{i,1,n}} = O(\tau_n)$. 
	In addition, \cite{hirano2003efficient} showed in the proof 
	of their Theorem 1 that $\Vert \wt \Sigma_n^{-1} \Vert = O_p(1)$
	and $\ERW{ \Vert V_n \Vert ^2 } = O(\zeta(h_n)^2)$, and the last result implies that 
	$\Vert V_n \Vert = O_p(\zeta(h_n))$ by the Markov inequality. Therefore, by the submultiplicativity of Frobenius norm we have 
	\[ |G_n^{4,7,1}| \le  
	\Vert \wt \Psi_{h_n}- \Psi_{h_n} \Vert
	\Vert \wt \Sigma_n^{-1} \Vert
	\Vert V_n \Vert 
	= O_p(\sqrt{\tau_n} \zeta(h_n)^3 \sqrt{\frac{h_n}{n}}) = o_p(1), \]
	where the last equality is implied by Assumption \ref{ass:sieve} $iv)$ $\frac{ \tau_n \zeta(h_n)^{10} h_n}{n} \to 0$.
	
	For $G_n^{4,7,2}$, we rewrite
	$G_n^{4,7,2} = \Psi_{h_n}^T \wt \Sigma_n^{-1}( \Sigma_n- \wt \Sigma_n  )\Sigma_n^{-1} V_n$.
	From the proof of Theorem 3.1 in \cite{zhang2018extremal}, we know that
	$ \Vert ( \wt \Sigma_n - \Sigma_n) \Sigma_n^{-1}V_n \Vert
	= O_p( \zeta(h_n)^4 \sqrt{\frac{h_n}{n}} + \frac{1}{\sqrt{n}} \zeta(h_n)^3).$
	In addition, we have $\Vert \Psi_{h_n} \Vert \le \frac{1}{c} \sqrt{\frac{n}{k}} \zeta(h_n) \ERW{ S_{i,1,n}}
	= O_p(\sqrt{\tau_n} \zeta(h_n))$.
	Therefore, 
	\[  |G_n^{4,7,1}| \le
	\Vert \Psi_{h_n}^T \Vert
	\Vert \wt \Sigma_n^{-1} \Vert
	\Vert (\Sigma_n- \wt \Sigma_n)\Sigma_n^{-1} V_n \Vert 
	= O_p(\sqrt{\frac{\tau_n}{n}}( \zeta(h_n)^5 \sqrt{h_n} +  \zeta(h_n)^4)) 
	=  o_p(1),
	\]
	where the last equality is implied by Assumption \ref{ass:sieve} $iv)$ $\frac{ \tau_n \zeta(h_n)^{10} h_n}{n} \to 0$.
	
	Now we have $\sqrt{k} G_n^{4} = -G_n^{4,7} + o_p(1) = - \Psi_{h_n}^T \Sigma_n^{-1} V_n + o_p(1)$.
	Let 
	\begin{align*}
		\delta_0(x) &:= 
		\sqrt{ \Pi(x)( 1- \Pi(x) )}
		\frac{ \ERW{ S_{i,1,n} | x}}{\sqrt{\tau_n} \Pi(x)}
		\\
		\delta_{h_n}(x) 
		& := \sqrt{\Pi_n(x)(1-\Pi_n(x))} \Psi_{h_n}^T \Sigma^{-1}_n H_{h_n}(x),
	\end{align*}
	we rewrite 
	\begin{align*}
		\Psi_{h_n}^T \Sigma_n^{-1} V_n 
		&= \frac{1}{\sqrt{n}} \Sn \delta_{0}(X_i) \frac{D_i - \Pi(X_i)}{\sqrt{\Pi(X_i)(1- \Pi(X_i))}} + G_n^{4,7,3} + G_n^{4,7,4},
	\end{align*}
	where
	\begin{align*}
		G_n^{4,7,3} & := 
		\frac{1}{\sqrt{n}} \Sn (\delta_{h_n}(X_i)-\delta_0(X_i)) \frac{D_i - \Pi(X_i)}{\sqrt{\Pi(X_i)(1- \Pi(X_i))}}
		\\
		G_n^{4,7,4} & := 
		\frac{1}{\sqrt{n}} \Sn \delta_{h_n}(X_i) \left(  \frac{D_i - \Pi_n(X_i)}{\sqrt{\Pi_n(X_i)(1- \Pi_n(X_i))}}- \frac{D_i - \Pi(X_i)}{\sqrt{\Pi(X_i)(1- \Pi(X_i))}} \right).
	\end{align*}
	Now we show that both terms $G_n^{4,7,3}$ and $G_n^{4,7,4}$ are $o_p(1)$.
	
	For $G_n^{4,7,4}$, from previous proofs we know that $\Vert \Psi_{h_n} \Vert = O_p(\sqrt{\tau_n} \zeta(h_n))$. In addition,
	\cite{hirano2003efficient} showed in the proof of their Theorem 1 that $\Vert  \Sigma_n^{-1} \Vert = O(1)$. Together with the fact that 
	$|\Pi_n| < 1$, we have $\sup_x \Vert \delta_{h_n}(x) \Vert = O_p( \sqrt{\tau_n} \zeta(h_n)^2)$.
	Since $\Pi$ is bounded away from 0 and 1 by Assumption \ref{ass:1}, using the mean value theorem and the triangular inequality give us
	\[
	|G_n^{4,7,4} |  = O(1)
	\sqrt{n} (\sup_x \Vert \delta_{h_n}(x) \Vert )
	(\sup_x | \Pi_n(x) - \Pi(x) | )
	= O_p(\sqrt{n \tau_n} \zeta(h_n)^3 h_n^{-s/2r})=o_p(1),
	\]
	where the last equality follows from Assumption \ref{ass:sieve} $iv)$ $n \tau_n \zeta(h_n)^6 h_n^{-s/r} \to 0 $.
	
	For $G_n^{4,7,3}$, 
	first note that we can view $\sqrt{\tau_n} \delta_{h_n}(x)$ as the least squares approximation of 
	$ \sqrt{\tau_n} \delta_0(x)$ using the approximation functions
	$\sqrt{L'(H_{h_n}(x)^T \pi_n)} H_{h_n}(x)$. 
	By Assumption \ref{ass:hill_sieve} $i)$ that
	$\ERW{ S_{i,1,n} | x}$ is $t$-times
	continuously differentiable
	with all derivatives bounded by $N_n$
	on $\supp(X)$. Thus, similarly as
	the Lemma 1 in \cite{hirano2003efficient} which follows
	from the Theorem 8 on page 90 in \cite{lorentz1966approximation}, it holds that
	\[
	\sup_{x \in \supp(X)} | \delta_0(x) - \delta_{h_n}(x) | 
	= O( N_n h_n^{-t /2r} / \sqrt{\tau_n}).
	\]
	Hence, by the triangular inequality we have
	\[
	|G_n^{4,7,3}| = O_p( \sqrt{\frac{n}{\tau_n}} N_n h_n^{- t /2r}) = o_p(1),
	\]
	where the last equality follows from Assumption \ref{ass:hill_sieve} $ii)$.
	
	Therefore, we have
	\begin{align*}
		G_n^{4,7} = \Psi_{h_n}^T \Sigma_n^{-1} V_n &=
		\frac{1}{\sqrt{n}} \Sn \delta_{0}(X_i) \frac{D_i - \Pi(X_i)}{\sqrt{\Pi(X_i)(1- \Pi(X_i))}} 
		+ o_p(1)
		\\ 
		& =  \frac{1}{\sqrt{n \tau_n}} \Sn \frac{\ERW{S_{i,1,n} \mid X_i}}{\Pi(X_i)} (D_i - \Pi(X_i))
		+ o_p(1),
	\end{align*}
	and consequently
	\begin{equation}
		\sqrt{k} G_n^4 = -  \frac{1}{\sqrt{n \tau_n}} \Sn \frac{\ERW{S_{i,1,n} \mid X_i}}{\Pi(X_i)} (D_i - \Pi(X_i))
		+ o_p(1).
		\label{eq:hill_prop_asymp_lin}
	\end{equation}
	
	Combining equations \eqref{eq:hill_Gn1_influence},
	\eqref{eq:hill_Gn2_influence}, \eqref{eq:gn3_to_0} and \eqref{eq:hill_prop_asymp_lin},
	we conclude that
	\[
	\sqrt{k} (\wh \gamma^H_1 - \gamma_1) =  
	\frac{\lambda_1}{1-\rho_1} + 
	\frac{1}{\sqrt{n}} \Sn \left( \psi_{i,1,n} + \gamma_1\phi_{i,1,n}  \right) + o_p(1).
	\]
\end{proof}

\subsection{Proof of Theorem~\ref{thm:hill_normality}}

\begin{proof}[Proof of Theorem~\ref{thm:hill_normality}]
	
	Let $V_{i,n} := (\psi_{i,1,n}, \psi_{i,0,n}, \phi_{i,1,n}, \phi_{i,0,n})^T$
	and let $\Sigma_n$ be its covariance matrix. 
	For $j=0,1$, we write $\psi_{i,j,n} = \nu_{i,j,n} + \eta_{i,j,n}$ with
	\begin{align*}
		\nu_{i,1,n} & = 
		\frac{1}{\sqrt{\tau_n}}  \left(
		\frac{ D_i}{\Pi(X_i)} S_{i,1,n}  
		- \gamma_1 \tau_n  \right),
		\qquad \quad
		\nu_{i,0,n}  = 
		\frac{1}{\sqrt{\tau_n}}  \left(
		\frac{ 1-D_i}{1-\Pi(X_i)} S_{i,0,n}  
		- \gamma_0 \tau_n  \right),
		\\
		\eta_{i,1,n} & = - \frac{1}{\sqrt{\tau_n}}
		\frac{\ERW{S_{i,1,n} \mid X_i}}{\Pi(X_i)} (D_i - \Pi(X_i)),
		\quad
		\eta_{i,0,n}  =  \frac{1}{\sqrt{\tau_n}}
		\frac{\ERW{S_{i,0,n} \mid X_i}}{1-\Pi(X_i)} (D_i - \Pi(X_i)).
	\end{align*}
	
	We will apply the multidimensional Lindeberg CLT to prove the claim. 
	We first show that the expectation of $V_{i,n}$ is a zero vector. For $\psi_{i,1,n}$, we have 
	\[
	\ERW{\eta_{i,1,n}} 
	= - \frac{1}{\sqrt{\tau_n}} \ERW{\ERW{\frac{\ERW{S_{i,1,n} \mid X_i}}{\Pi(X_i)} (D_i - \Pi(X_i)) \bigg| X_i}} 
	= 0
	\]
	and
	\[
	\ERW{ \nu_{i,1,n} }
	=\frac{1}{\sqrt{\tau_n}} \ERW{\frac{ D_i}{\Pi(X_i)} S_{i,1,n} } - \gamma_1\sqrt{\tau_n}
	=\frac{1}{\sqrt{\tau_n}} \ERW{S_{i,1,n} } - \gamma_1\sqrt{\tau_n}
	=0,
	\]
	where the second equality follows from Lemma~\ref{IdentifiabilityPropensityScore} and the last equality holds 
	as $\ERW{S_{i,1,n} } = \gamma_1\tau_n$ which can be seen from the proof of Lemma~\ref{lemma:s_order}.
	Thus we have $\ERW{\psi_{i,1,n}}=0$. Analogously, one can show that $\ERW{\psi_{i,0,n}}=\ERW{\phi_{i,1,n}}=\ERW{\phi_{i,0,n}}=0$.
	
	Now we show that the covariance matrix $\Sigma_n$ converge to $\Sigma$ via showing entry-wise convergence.
	We first compute
	\begin{align*}
		&  \quad \ERW{ \nu_{i,1,n} \phi_{i,1,n} }
		= \frac{1}{\tau_n} \E \bigg[ \left( 
		\frac{D_i}{\Pi(X_i)}S_{i,1,n}
		- \gamma_1 \tau_n \right)
		\cdot  
		\bigg( \frac{D_i}{\Pi(X_i)} 
		T_{i,1,n} - \frac{\ERW{T_{i,1,n} \mid X_i}}{\Pi(X_i)} (D_i - \Pi(X_i))\bigg)
		\bigg] \\
		&=
		\frac{1}{\tau_n} \ERW{  \frac{D_i}{\Pi(X_i)^2} S_{i,1,n}}
		-  \frac{1}{\tau_n} \ERW{ \frac{D_i}{\Pi(X_i)^2}S_{i,1,n} (1 - \Pi(X_i)) P(Y_i(1) > q_1(1-\tau_n)\mid X_i)  }  
		-  \ERW{  \frac{D_i}{\Pi(X_i)} S_{i,1,n} } \\
		&= 
		\frac{1}{\tau_n} \ERW{  \frac{ \ERW{S_{i,1,n}|X_i} } {\Pi(X_i)} }
		-  \frac{1}{\tau_n} \ERW{ \frac{ \ERW{S_{i,1,n}|X_i} }{\Pi(X_i)} (1 - \Pi(X_i)) P(Y_i(1) > q_1(1-\tau_n)\mid X_i) }  
		-  \ERW{ S_{i,1,n} } \\
		&=
		\frac{1}{\tau_n} \ERW{ \frac{ \ERW{S_{i,1,n}|X_i} }{\Pi(X_i)} (1- (1 - \Pi(X_i)) P(Y_i(1) > q_1(1-\tau_n)\mid X_i)) }  
		+ o(1),
	\end{align*}
	where the second last equality is obtained by applying the law of iterated expectation and the last equality follows from Lemma~\ref{lemma:s_order}.
	Similarly, we have
	\begin{align*}
		&\ERW{ \eta_{i,1,n} \phi_{i,1,n} } = 0, \\ 
		&\ERW{  \nu_{i,1,n} \eta_{i,1,n}} = - \frac{1}{\tau_n} \E \bigg[ \frac{1-\Pi(X_i)}{\Pi(X_i)} \ERW{ S_{i,1,n} \mid X_i}^2 \bigg], \\ 
		&\ERW{  \nu_{i,1,n}^2} =  \frac{\gamma_1^2}{\tau_n} \E 
		\bigg[ \frac{1}{\Pi(X_i)} \ERW{ \log(Z_i(1) \tau_n)^2  \Ind{ Y_i(1) > q_1(1-\tau_n)} \mid X_i}  \bigg] + o(1), \\ 
		&\ERW{  \eta_{i,1,n}^2} =  \frac{1}{\tau_n} \E \bigg[ \frac{1-\Pi(X_i)}{\Pi(X_i)}  \ERW{S_{i,1,n}\mid X_i}^2 \bigg]. \\ 
	\end{align*}
	Thus
	\begin{align*}
		&\ERW{ \psi_{i,1,n}^2 } = \frac{1}{\tau_n} \E \bigg[ \frac{\gamma_1^2}{\Pi(X_i)} \ERW{ \log(Z_i(1) \tau_n)^2 
			\Ind{ Y_i(1) > q_1(1-\tau_n)} \mid X_i}  -   \frac{1-\Pi(X_i)}{\Pi(X_i)} \ERW{S_{i,1,n}\mid X_i}^2 \bigg] + o(1), \\
		& \ERW{ \psi_{i,1,n} \phi_{i,1,n} } = \frac{1}{\tau_n} \E \bigg[ \frac{1}{\Pi(X_i)} \ERW{ 
			S_{i,1,n} \mid X_i}  \cdot \bigg(1- (1-\Pi(X_i))  P(Y_i(1) > q_1(1-\tau_n)\mid X_i) \bigg) 
		\bigg] + o(1).
	\end{align*}
	Analogously, we have 
	\begin{align*}
		&\ERW{ \psi_{i,0,n}^2 } = \frac{1}{\tau_n} \E \bigg[ \frac{\gamma_0^2}{1-\Pi(X_i)} \ERW{ \log(Z_i(0) \tau_n)^2 
			\Ind{ Y_i(0) > q_0(1-\tau_n)} \mid X_i}  -   \frac{\Pi(X_i)}{1-\Pi(X_i)} \ERW{S_{i,0,n}\mid X_i}^2 \bigg] + o(1), \\
		& \ERW{ \psi_{i,0,n} \phi_{i,0,n} } = \frac{1}{\tau_n} \E \bigg[ \frac{1}{1-\Pi(X_i)} \ERW{ 
			S_{i,0,n} \mid X_i}  \cdot \bigg(1- \Pi(X_i)  P(Y_i(0) > q_0(1-\tau_n)\mid X_i) \bigg) 
		\bigg] + o(1).
	\end{align*}
	
	For the intersection terms between two potential outcome distributions, we have
	\begin{align*}
		\ERW{ \nu_{i,0,n}\eta_{i,1,n}} 
		& =\ERW{ \nu_{i,1,n}\eta_{i,0,n}}
		= 
		\frac{1}{\tau_n} \E \bigg[ \ERW{ 
			S_{i,1,n} \mid X_i}  
		\ERW{S_{i,0,n}\mid X_i}
		\bigg],
		\\
		\ERW{ \eta_{i,1,n} \eta_{i,0,n}}
		& =  -
		\frac{1}{\tau_n} \E \bigg[ 
		\ERW{S_{i,1,n}\mid X_i}
		\ERW{S_{i,0,n}\mid X_i}
		\bigg],
		\\
		\ERW{ \nu_{i,1,n} \nu_{i,0,n} } 
		&= o(1),
		\\
		\ERW{ \eta_{i,1,n} \phi_{i,0,n}} 
		& = \ERW{ \eta_{i,0,n} \phi_{i,1,n}}  = 0,
		\\
		\ERW{ \nu_{i,1,n} \phi_{i,0,n}}
		& =
		\frac{1}{\tau_n} \E \bigg[\ERW{  S_{i,1,n} \mid X_i} P(Y_i(0) > q_0(1-\tau_n)\mid X_i) \bigg] + o(1),
		\\
		\ERW{ \nu_{i,0,n} \phi_{i,1,n}}
		& =
		\frac{1}{\tau_n} \E \bigg[\ERW{  S_{i,0,n} \mid X_i} P(Y_i(1) > q_1(1-\tau_n)\mid X_i) \bigg] + o(1).
	\end{align*}
	Thus
	\begin{align*}
		&\ERW{ \psi_{i,1,n} \psi_{i,0,n} } = \frac{1}{\tau_n} \E \bigg[ \ERW{ S_{i,1,n} \mid X_i} \ERW{S_{i,0,n}\mid X_i} \bigg]+ o(1), \\
		&\ERW{ \psi_{i,1,n} \phi_{i,0,n} } = \frac{1}{\tau_n} \E \bigg[\ERW{  S_{i,1,n} \mid X_i} P(Y_i(0) > q_0(1-\tau_n)\mid X_i) \bigg] + o(1), \\
		&\ERW{ \psi_{i,0,n} \phi_{i,1,n} } = \frac{1}{\tau_n} \E \bigg[\ERW{  S_{i,0,n} \mid X_i} P(Y_i(1) > q_1(1-\tau_n)\mid X_i) \bigg] + o(1).\\
	\end{align*}
	At last, we have
	\begin{align*}
		&\ERW{ \phi_{i,1,n}^2 } = \frac{1}{\tau_n} \E \biggl[ \frac{P(Y_i(1) > q_1(1- \tau_n) \mid X_i)}{\Pi(X_i)}
		- \frac{1-\Pi(X_i)}{\Pi(X_i)} P(Y_i(1) > q_1(1- \tau_n)\mid X_i)^2 \biggr] + o(1), \\
		&\ERW{ \phi_{i,0,n}^2} = \frac{1}{\tau_n} \E \biggl[ \frac{P(Y_i(0) > q_0(1- \tau_n) \mid X_i)}{1-\Pi(X_i)}
		- \frac{\Pi(X_i)}{1-\Pi(X_i)} P(Y_i(0) > q_0(1- \tau_n)\mid X_i)^2 \biggr] + o(1), \\
		&\ERW{ \phi_{i,1,n} \phi_{i,0,n} } = \frac{1}{\tau_n} \ERW{ P(Y_i(1) > q_1(1-\tau_n)\mid X_i) P(Y_i(0) > q_0(1-\tau_n)\mid X_i) } + o(1).\\
	\end{align*}
	Therefore, under Assumption \ref{ass:hill_Technical} and \ref{ass:zhangTechnical}, we have 
	\begin{align*}
		&\ERW{ \psi_{i,1,n}^2 } \to G_1, \quad
		\ERW{ \psi_{i,0,n}^2 } \to G_0, \quad
		\ERW{ \psi_{i,1,n} \psi_{i,0,n} } \to G_{10}, \\
		&\ERW{ \psi_{i,1,n} \phi_{i,1,n} } \to J_1, \quad
		\ERW{ \psi_{i,0,n} \phi_{i,0,n} } \to J_0, \quad
		\ERW{ \psi_{i,1,n} \phi_{i,0,n} } \to J_{10}, \quad
		\ERW{ \psi_{i,0,n} \phi_{i,1,n} } \to J_{01}, \\
		&\ERW{ \phi_{i,1,n}^2 } \to H_1, \quad
		\ERW{ \phi_{i,0,n}^2 } \to H_0, \quad
		\ERW{ \phi_{i,1,n} \phi_{i,0,n} } \to H_{10}.
	\end{align*}
	
	The above results implies that
	\[
	\ERW {\frac{V_{i,n}}{\sqrt{n}}} = 0 
	\quad \text{and} \quad
	\lim_{n \to \infty} \Sn \Cov \left( \frac{V_{i,n}}{\sqrt{n}}  \right) = \lim_{n \to \infty} \Sigma_n = \Sigma.
	\]
	Now we verify the Lindeberg condition, that is, for any $\epsilon > 0$,
	\[
	\lim_{n \to \infty} \Sn \ERW{\left\Vert \frac{V_{i,n}}{\sqrt{n}} \right\Vert ^2  \Ind{\left\Vert \frac{V_{i,n}}{\sqrt{n}} \right\Vert > \epsilon} } 
	= \lim_{n \to \infty} \ERW{\Vert V_{i,n} \Vert ^2  \Ind{\Vert V_{i,n}\Vert > \sqrt{n} \epsilon} } 
	\to 0.
	\]
	
	Since on $\Vert V_{i,n} \Vert > \sqrt{n} \epsilon$ we have $ (\Vert V_{i,n} \Vert / \sqrt{n} \epsilon)^2 > 1 $, we can bound
	\[
	\ERW{\Vert V_{i,n} \Vert ^2  \Ind{\Vert V_{i,n}\Vert > \sqrt{n} \epsilon} } 
	\le \frac{1}{n \epsilon^2} \ERW{\Vert V_{i,n} \Vert^4}.
	\]
	Thus, it is enough to show the Lyapunov type condition $\frac{1}{n \epsilon^2} \ERW{\Vert V_{i,n} \Vert^4} \to 0$.
	
	Recall that $T_{i,j,n} =  \Ind{Y_i(j) > q_j(1-\tau_n)}- \tau_n$
	and 
	$S_{i,j,n} = \gamma_j \log \left( \frac{\tau_n}{1-F_j(Y_i(j))} \right)  \Ind{ Y_i(j) > q_j(1-\tau_n)}$.
	For any $p \in \mathbb{N}$, we know from Lemma~\ref{lemma:s_order}  that
	\[
	\ERW{S_{i,1,n}^p } = O(\tau_n)
	\quad \text{and} \quad
	\ERW{S_{i,1,n}^p \mid  X_i} = O_p(\tau_n).
	\]
	In addition, it is easy to see that 
	\[
	\ERW{T_{i,1,n}^p } = O(\tau_n) 
	\quad \text{and} \quad
	\ERW{T_{i,1,n}^p \mid  X_i} = O_p(\tau_n).
	\]
	Hence, for any $p,q \in \mathbb{N}$, by the Cauchy-Schwarz inequality we have 
	\begin{align*}
		\ERW{ S_{i,1,n}^p \ERW{S_{i,1,n} \mid  X_i}^q } 
		&\le \sqrt{ \ERW{ S_{i,1,n}^{2p}} \ERW{\ERW{S_{i,1,n} \mid  X_i}^{2q} }  } \\
		&\le \sqrt{ \ERW{ S_{i,1,n}^{2p}} \ERW{\ERW{S_{i,1,n}^{2q} \mid  X_i} }  }
		= O(\tau_n).
	\end{align*}
	Similar inequalities hold for all combinations of $S_{i,1,n}$, $\ERW{S_{i,1,n} \mid  X_i}$, $T_{i,1,n}$ and $\ERW{T_{i,1,n} \mid  X_i}$.
	
	Since $\Pi(X_i)$ is bounded away from $0$ and $1$ by Assumption~\ref{ass:1}, for
	\begin{align*}
		\frac{1}{n \epsilon^2} \ERW{\Vert V_{i,n} \Vert^4} 
		&= \frac{1}{n \epsilon^2} \E \big[ \psi_{i,1,n}^4 + \psi_{i,0,n}^4 + \phi_{i,1,n}^4 + \phi_{i,0,n}^4 + 2\psi_{i,1,n}^2\psi_{i,0,n}^2 + 2\psi_{i,1,n}^2\phi_{i,1,n}^2 \\
		&\qquad+ 2\psi_{i,1,n}^2\phi_{i,0,n}^2 + 2\psi_{i,0,n}^2\phi_{i,1,n}^2 + 2\psi_{i,0,n}^2\phi_{i,0,n}^2 + 2\phi_{i,1,n}^2\phi_{i,0,n}^2 \big],
	\end{align*}
	we can see by expanding all above terms that its rate is of order
	\[
	\frac{1}{n} \frac{1}{\tau_n^2} O(\tau_n) = O(k^{-1}) = o(1),
	\]
	which shows that the Lyapunov type condition hold.

	Therefore, we can apply the multidimensional Lindeberg CLT to obtain
	\[
	\frac{1}{\sqrt{n}} \Sn V_{i,n} \tendsto{D} \wt N
	\]
	where $ \wt N $ is a $4$-dimensional Gaussian vector with mean zero and covariance matrix $\Sigma$. The claim then follows from applying the continuous mapping theorem.
	
\end{proof}

\subsection{Proof of Lemma~\ref{lemma:hill_extreme_quantiles}}

\begin{proof}[Proof of Lemma~\ref{lemma:hill_extreme_quantiles}]
	The proof is similar to the proof of Theorem 4.3.8 in \cite{de2007extreme}.
	Denote $d_n := \tau_n/ p_n  = k/ (n p_n)$, we have
	$\log(d_n) / \sqrt{k} \to 0$
	by the assumption that $ \log(n p_n) = o(\sqrt{k})$.
	Using $q_j(1-\tau) = U_j(1/\tau)$, we can expand
	\begin{align*}
		\frac{\sqrt{k}}{\log(d_n)} \left( \frac{\wh{Q}_j(1-p_n) }{ q_j(1-p_n)} -1  \right)
		&= \frac{d_n^{\gamma_j} U_j\nk}{U_j(\frac1{p_n})}
		\Biggl( \frac{d_n^{\wh \gamma_j^H - \gamma_j}}{\log(d_n)} \sqrt{k} \left( \frac{\wh q_j(1-\tau_n)}{U_j \nk} -1 \right) \\
		& \quad + \frac{\sqrt{k}}{\log(d_n)} \left( d_n^{\wh \gamma_j^H - \gamma_j} - 1 \right) \\
		& \quad - \frac{\sqrt{k} A_j \nk}{\log(d_n)} 
		\frac{ \frac{ U_j(1/p_n) d_n^{-\gamma_j}}{U_j(n/k)} - 1}{A_j \nk}
		\Biggr).
	\end{align*}
	
	By using the same arguments as in the proof of Lemma~\ref{lemma:hill_consistency}, we have
	\[
	\sqrt{k} \left( \frac{\wh q_j(1-\tau_n)}{U_j \nk} - 1 \right) = O_p(1)
	\]
	By applying the Theorem 2.3.9 in \cite{de2007extreme}, we have
	\begin{equation*}
		\lim_{n \to \infty} \frac{ \frac{ U_j(1/p_n) d_n^{-\gamma_j}}{U_j(n/k)} - 1}{A_j \nk}
		= - \frac1{\rho_j},
	\end{equation*}
	which then implies
	\[
	\frac{ U_j(1/p_n) d_n^{-\gamma_j}}{U_j(n/k)} \to 1
	\quad \text{or} \quad
	\frac{d_n^{\gamma_j} U_j\nk}{U_j(\frac1{p_n})} \to 1
	\]
	as $A_j \nk \to 0$.
	By assumption that $\sqrt{k} A_j \nk \to \lambda_j$, we have
	\[
	\frac{\sqrt{k} A_j \nk}{\log(d_n)}  \to 0.
	\]
	Note that 
	\[
	\frac{\sqrt{k}}{\log(d_n)} \left( d_n^{\wh \gamma_j^H - \gamma_j} - 1 \right) 
	= \frac{\sqrt{k} (\wh \gamma_j^H - \gamma_j)}{\log(d_n)} \int_1^{d_n} \frac{ e^{(\wh \gamma_j^H - \gamma_j) \log s}}{s} \ ds,
	\]
	so
	\begin{align}
		\left| 
		\frac{\sqrt{k}}{\log(d_n)} \left( d_n^{\wh \gamma_j^H - \gamma_j} - 1 \right) 
		- \sqrt{k} (\wh \gamma_j^H - \gamma_j) 
		\right|
		=
		| \sqrt{k} (\wh \gamma_j^H - \gamma_j) |
		\left| 
		\frac{1}{\log(d_n)} \int_1^{d_n} \frac{ e^{(\wh \gamma_j^H - \gamma_j) \log s}}{s} \ ds - 1
		\right|.
		\label{ProofOfLemma2InterRestult}
	\end{align}
	Since for all $s \in [1, d_n]$, we have 
	\[
	-|\wh \gamma_j^H - \gamma_j| \log d_n \le
	(\wh \gamma_j^H - \gamma_j) \log s \le
	|\wh \gamma_j^H - \gamma_j| \log d_n,
	\]
	thus
	\begin{align*}
		e^{-|\wh \gamma_j^H - \gamma_j| \log d_n} \le
		\frac{1}{\log(d_n)} \int_1^{d_n} \frac{ e^{(\wh \gamma_j^H - \gamma_j) \log s}}{s} \ ds \le
		e^{|\wh \gamma_j^H - \gamma_j| \log d_n}.
	\end{align*}
	By Theorem~\ref{thm:hill_normality} and the fact that $\log(d_n) / \sqrt{k} \to 0$, we have 
	\[
	\pm |\wh \gamma_j^H - \gamma_j| \log d_n
	= \pm \sqrt{k} |\wh \gamma_j^H - \gamma_j| \frac{\log(d_n)}{\sqrt{k}}
	\tendsto{P} 0,
	\]
	thus
	\[
	\frac{1}{\log(d_n)} \int_1^{d_n} \frac{ e^{(\wh \gamma_j^H - \gamma_j) \log s}}{s} \ ds \tendsto{P} 1.
	\]
	Combing the equality \eqref{ProofOfLemma2InterRestult} and the fact that $\sqrt{k} (\wh \gamma_j^H - \gamma_j) = O_p(1)$, we have
	\begin{align*}
		\frac{\sqrt{k}}{\log(d_n)} \left( d_n^{\wh \gamma_j^H - \gamma_j} - 1 \right) 
		= \sqrt{k} (\wh \gamma_j^H - \gamma_j) + o_p(1),
	\end{align*}
	which implies
	\[
	\frac{d_n^{\wh \gamma_j^H - \gamma_j}}{\log(d_n)} \tendsto{P} 0.
	\]
	
	Combining the above results, we conclude that
	\[
	\frac{\sqrt{k}}{\log(d_n)} \left( \frac{\wh{Q}_j(1-p_n) }{ q_j(1-p_n)} -1  \right)
	= \sqrt{k} (\wh \gamma_j^H - \gamma_j) + o_p(1),
	\]
	which implies
	\[
	\frac{\wh{Q}_j(1-p_n) }{ q_j(1-p_n)} \tendsto{P} 1
	\]
	as ${\sqrt{k}}/{\log(d_n)} \to \infty$.
	
\end{proof}

\subsection{Proof of Theorem~\ref{thm:hill_extreme_qte}}
\begin{proof}[Proof of Theorem~\ref{thm:hill_extreme_qte}]
	Let $d_n := \tau_n / p_n$.
	First, we expand 
	\begin{align*}
		& \wh \beta_n \left( \wh{\delta}(1-p_n) - \delta(1-p_n) \right) \\
		= & \wh{\beta}_n \left( \wh{Q}_1(1-p_n)-q_1(1-p_n) - ( \wh{Q}_0(1-p_n) - q_0(1-p_n)) \right)\\ 
		= &
		\frac{\wh Q_1(1-p_n)}{\max\{ \wh Q_1(1-p_n), \wh Q_0(1-p_n) \}}
		\frac{\sqrt{k}}{\log(d_n)} \frac{\wh{Q}_1(1-p_n) - q_1(1-p_n)}{ \wh Q_1(1-p_n)} \\
		& -
		\frac{\wh Q_0(1-p_n)}{\max\{ \wh Q_1(1-p_n), \wh Q_0(1-p_n) \}}
		\frac{\sqrt{k}}{\log(d_n)} \frac{\wh{Q}_0(1-p_n) - q_0(1-p_n)}{ \wh Q_0(1-p_n)}.
	\end{align*}
	By Lemma~\ref{lemma:hill_extreme_quantiles} and Theorem~\ref{thm:hill_normality}, we have that for $j=0,1$,
	\begin{align*}
		&\frac{\sqrt{k}}{\log(d_n)} \frac{\wh{Q}_j(1-p_n) - q_j(1-p_n) }{ \wh Q_j(1-p_n) } 
		= \sqrt{k} (\wh \gamma_j^H - \gamma_j) + o_p(1) = O_p(1) \\
		\text{and} \quad
		&\wh{Q}_j(1-p_n) \tendsto{P} q_j(1-p_n),
	\end{align*}
	thus
	\[
	\wh{\beta}_n \left( \wh{\delta}(1-p_n) - \delta(1-p_n) \right)
	= 
	\min \{ 1, \kappa \} \sqrt{k}
	(\wh \gamma_1^H - \gamma_1)
	- 
	\min \left\{ 1, \frac{1}{\kappa} \right\} 
	\sqrt{k} (\wh \gamma_0^H - \gamma_0)
	+ o_p(1).
	\]
	By applying the continuous mapping theorem and Theorem~\ref{thm:hill_normality}, we have
	\[
	\wh{\beta}_n \left( \wh{\delta}(1-p_n) - \delta(1-p_n) \right)
	\tendsto{D} \N( \mu, \sigma^2),
	\]
	with $\mu = v_{\kappa}^T w_{\lambda, \rho}$ and $\sigma^2 = v_{\kappa}^T B \Sigma B^T v_{\kappa}$, where the notations are defined as in the description of the theorem.

\end{proof}

\subsection{Proof of Theorem~\ref{thm:hill_variance_estimation}}
We first introduce Lemma~\ref{lemma:easy_sigma_hill} which shows that under Assumption~\ref{ass:squared_hill_tozero},
the covariance matrix $\Sigma$ simplifies to $\wt  \Sigma$.
In particular, the components $G_{10}$, $J_{10}, J_{01}$ and
$H_{10}$ become zero, which implies that the extrapolated quantiles $\wh Q_1^H(1-p_n)$ and $\wh Q_0^H(1-p_n)$ are asymptotically independent.
\begin{lemma}
	\label{lemma:easy_sigma_hill}
	Suppose that Assumptions \ref{ass:1} and \ref{ass:squared_hill_tozero} hold, then $\Sigma =\wt  \Sigma $.
\end{lemma}

\begin{proof}[Proof of Lemma~\ref{lemma:easy_sigma_hill}] 
	It is sufficient to show that the entries of $\Sigma$ and $\wt  \Sigma $ are equal.
	For $H_1, H_0, G_1, G_0, J_1$ and $J_0$, the equalities are direct consequences of Assumption \ref{ass:squared_hill_tozero} and Assumption \ref{ass:1} that $\Pi(x)$ is bounded away from $0$ and $1$.
	For other terms the equalities can be proved by using the Cauchy--Schwarz inequality.
\end{proof}

Now we give the proof of Theorem~\ref{thm:hill_variance_estimation}.
\begin{proof}[Proof of Theorem~\ref{thm:hill_variance_estimation}]
	Recall that $\wh \sigma^2 = \wh v_{\kappa}^T \wh B \wh \Sigma \wh B^T \wh v_{\kappa}$, where
	\begin{equation*}
		\wh \Sigma =  \left(
		\begin{matrix}                                
			\wh G_{1} &  0 & \wh J_{1} & 0 \\
			0 &  \wh G_{0} & 0  & \wh J_{0} \\
			\wh J_{1} & 0 & \wh H_{1} & 0 \\
			0 & \wh J_{0} & 0 & \wh H_{0} \\
		\end{matrix}
		\right) , \
		\wh B =  \left(
		\begin{matrix}                                
			1 & 0 & -\wh \gamma_1^H & 0 \\
			0 & 1 & 0 & -\wh \gamma_0^H \\
		\end{matrix}
		\right), \ 
		\wh v_{\kappa} = \left(
		\begin{matrix}
			\min\{1, \wh \kappa \} \\
			-\min\{1, \frac{1}{\wh \kappa} \}
		\end{matrix}
		\right)
	\end{equation*}
	with $\wh H_1 , \wh H_0, \wh G_1, \wh G_0, \wh J_{1}, \wh J_{0}$ are defined in equation \eqref{covariance_est_terms}. 
	By Lemma~\ref{lemma:hill_consistency}, we have $\wh \gamma_j^H \tendsto{P} \gamma_j$ for $j=0,1$.
	By Lemma~\ref{lemma:hill_extreme_quantiles} and Assumption~\ref{ass:normal_conv}, we have $\wh \kappa \tendsto{P} \kappa$.
	By Lemma~\ref{lemma:easy_sigma_hill}, only the covariance terms $H_1, H_0$, $G_1, G_0$, $J_1$ and $J_0$ in $\Sigma$
	are nonzero.
	Therefore, to prove the consistency of the estimated variance $\wh \sigma^2$, it is suffice to show that for $j=0,1$,
	\begin{align*}
		(i)\ \wh H_j \tendsto{P} H_j, \qquad  
		(ii)\ \wh G_j \tendsto{P} G_j, \qquad 
		(iii)\ \wh J_j \tendsto{P} J_j,
	\end{align*}
	where $G_1, G_0, J_{1}, J_{0}$ are defined
	as in Assumption \ref{ass:hill_Technical} 
	and $H_1, H_0$ as in Assumption \ref{ass:zhangTechnical}. 
	We only show the result for $j=1$, the case of $j=0$ can be shown analogously. 
	We proceed in a similar manner as in the proof of Lemma~\ref{lemma:hill_consistency}.
	
	For (i), we expand
	\begin{align*}
		\wh H_1 = H_{1}^{n,1 } + H_{1}^{n, 2} + H_{1}^{n, 3} 
	\end{align*}
	with
	\begin{align*}
		& H_{1}^{n,1 } = \frac{1}{n \tau_n}   \Sn \frac{D_i}{\Pi(X_i)^2} \Ind{Y_i > q_1(1-\tau_n)},\\
		& H_{1}^{n,2 } =  \frac{1}{n \tau_n}   \Sn D_i \left( \frac{1}{\wh{\Pi}(X_i)^2} -  \frac{1}{\Pi(X_i)^2} \right) \Ind{Y_i > q_1(1-\tau_n)},  \\
		& H_{1}^{n,3 } = \frac{1}{n \tau_n}   \Sn \frac{D_i}{\wh{\Pi}(X_i)^2}  \left( \Ind{Y_i > \wh{q}_1{(1-\tau_n})} - \Ind{Y_i > {q}_1{(1-\tau_n})} \right).
	\end{align*}
	We will show that $H_1^{n,1} \tendsto{P} H_1$ and that $H_1^{n,2}, H_1^{n,3}$ converge to zero in probability.
	
	For $H_{1}^{n,1 } $, we have
	\begin{align*}
		\ERW{H_{1}^{n,1 }} = \frac{1}{n \tau_n}   \Sn \ERW{ \frac{D_i}{\Pi(X_i)^2} \Ind{Y_i > q_1(1-\tau_n)} } 
		=  \frac{1}{\tau_n}   \ERW{ \frac{P(Y_i(1) > q_1(1-\tau_n) \mid X_i)}{\Pi(X_i)} } \to H_1
	\end{align*}
	by Assumption \ref{ass:hill_Technical}, Assumption \ref{ass:squared_hill_tozero} and Lemma~\ref{lemma:easy_sigma_hill},
	and
	\begin{align*} 
		\VAR{H_{1}^{n,1 }} =& \frac{1}{n \tau_n^2}  \VAR{ \frac{D_i}{\Pi(X_i)^2} \Ind{Y_i(1) > q_1(1-\tau_n)} } \\
		\le &
		\frac{1}{n \tau_n^2}  \ERW{ \left( \frac{D_i}{\Pi(X_i)^2} \Ind{Y_i(1) > q_1(1-\tau_n)} \right)^2 } \\
		\le &
		\frac{1}{c^2} \frac{1}{n \tau_n}  \frac{1}{\tau_n}\ERW{ \frac{P(Y_i(1) > q_1(1-\tau_n) \mid X_i )}{\Pi(X_i)} } \to 0
	\end{align*}
	since $\frac{1}{n \tau_n} \to 0$ and $\frac{1}{\tau_n}\ERW{ \frac{P(Y_i(1) > q_1(1-\tau_n) \mid X_i )}{\Pi(X_i)} } \to H_1 < \infty$, 
	thus $H_{1}^{n,1 } \tendsto{P} H_1$.
	
	For $H_{1}^{n,2 } $, we have
	\begin{align*}
		|H_{1}^{n,2 }| \le \sup_x \left| \frac{1}{\wh{\Pi}(x)^2} - \frac{1}{\Pi(x)^2} \right|  \frac{1}{n \tau_n} \Sn D_i \Ind{Y_i > q_1(1-\tau_n)}.
	\end{align*}
	Note that
	\begin{equation*}
		\sup_x \left| \frac{1}{\wh{\Pi}(x)^2} - \frac{1}{\Pi(x)^2} \right| = o_p(1)
	\end{equation*}
	by a similar proof as the proof of Lemma~\ref{lemma:propensity_consistency} and 
	\begin{align*}
		\frac{1}{n \tau_n} \Sn D_i \Ind{Y_i > q_1(1-\tau_n)} = O_{p}(1)
	\end{align*}
	by the fact that $\ERW{ \frac{1}{n \tau_n} \Sn D_i  \Ind{Y_i > q_1(1-\tau_n)}} = 1$ and the Markov inequality, we have $H_{1}^{n,2 } = o_{p}(1)$.
	
	For $H_{1}^{n,3 } $, note that the terms
	$ \Ind{Y_i > \wh{q}_1{(1-\tau_n})} - \Ind{Y_i > {q}_1{(1-\tau_n})} $
	have the same signs for all $i=1,\dots,n$. Thus
	\begin{align*}
		|H_{1}^{n,3 }| 
		& =  \left| \frac{1}{n \tau_n}   \Sn \frac{D_i}{\wh{\Pi}(X_i)^2}  
		\left( \Ind{Y_i > \wh{q}_1{(1-\tau_n})} - \Ind{Y_i > {q}_1{(1-\tau_n})} \right) \right|  \nonumber \\ 
		& \le 
		\sup_x \left| \frac{1}{\wh{\Pi}(x)} \right| \left| \frac{1}{n \tau_n}   \Sn \frac{D_i}{\wh{\Pi}(X_i)}  
		\left( \Ind{Y_i > \wh{q}_1{(1-\tau_n})} - \Ind{Y_i > {q}_1{(1-\tau_n})} \right) \right|.
	\end{align*}
	We have seen in \eqref{eq:prob_bounded_prop} that $\sup_x \left| \frac{1}{\wh{\Pi}(x)} \right| = O_p(1)$, together with 
	Lemma~\ref{lemma:weighted_cdf_diff} we conclude that $H_{1}^{n,3} = o_{p}(1)$.
	Combining the above results we have $\wh H_1 \tendsto{P} H_1$.
	
	For (ii), we expand
	\begin{align*}
		\wh G_1 = G_{1}^{n,1 } + G_{1}^{n, 2} + G_{1}^{n, 3} + G_1^{n,4}
	\end{align*}
	with
	\begin{align*}
		G_1^{n,1} & = \frac{1}{k} \Sn (\log(Y_i) - \log( q_1(1-\tau_n)))^2 
		\frac{D_i}{ \Pi(X_i)^2} \Ind{Y_i >  q_1(1-\tau_n)}, \\
		G_1^{n,2} & = 
		\frac{\Delta_n}{k}  \Sn
		\big( 2(\log(Y_i) - \log( q_1(1-\tau_n))) 
		+ 
		\Delta_n \big) 
		\frac{D_i}{ \Pi(X_i)^2} \Ind{Y_i >  q_1(1-\tau_n)}, \\
		G_1^{n,3} &=  \frac{1}{k} \Sn (\log(Y_i) - \log( \wh q_1(1-\tau_n)))^2 
		\frac{D_i}{  \wh \Pi(X_i)^2} \left( \Ind{Y_i > \wh q_1(1-\tau_n)} -
		\Ind{Y_i >  q_1(1-\tau_n)} 
		\right), \\
		G_1^{n,4} &= \frac{1}{k} \Sn (\log(Y_i) - \log( \wh q_1(1-\tau_n))) ^2
		D_i \Ind{Y_i >  q_1(1-\tau_n)} 
		\left(
		\frac{1}{  \wh \Pi(X_i)^2}- \frac{1}{ \Pi(X_i)^2} 
		\right),
	\end{align*}
	where $\Delta_n = \log( q_1(1-\tau_n)) -\log( \wh q_1(1-\tau_n))$.
	We will show that $G_1^{n,1} \tendsto{P} G_1$ and 
	that $G_1^{n,2}, G_1^{n,3}, G_1^{n,4}$ converge to zero in probability.
	
	For $G_1^{n,1}$, let $Z_i(1) = 1/(1-F_1(Y_i(1)))$. As in the proof
	of Lemma~\ref{lemma:hill_consistency}, we have $Y_i(1) = U_1(Z_i(1))$
	almost surely. Because $q_1(1-\tau_n) = U_1(\tau_n^{-1})$, we have that almost surely
	\[
	G_1^{n,1}  = \frac{1}{k} \Sn (\log(U_1(Z_i(1))) - \log( U_1(n/k) ))^2 
	\frac{D_i}{ \Pi(X_i)^2} \Ind{Y_i >  q_1(1-\tau_n)}.
	\]
	
	Since $F_1$ satisfies the max-domain of attraction condition with a positive extreme value index $\gamma_1$, we can apply the Corollary 1.2.10 and the statement 5 of the Proposition B.1.9 in \cite{de2007extreme} to obtain that 
	for any  $\varepsilon, \varepsilon' >0$ such that $\varepsilon < 1$ and $\varepsilon' < \gamma_1$,
	there exists $t_0$ such that for any $x > 1$ and $t\ge t_0$, we have
	\[
	(1-\varepsilon) x^{\gamma_1- \varepsilon'}
	< 
	\frac{U_1(tx)}{U_1(t)}
	<
	(1+\varepsilon) x^{\gamma_1+ \varepsilon'}.
	\]
	Since $\varepsilon$ and $ \varepsilon'$ can be arbitrary small, we can take them small enough such that $(1-\varepsilon) x^{\gamma_1- \varepsilon'} > 1$. 
	Hence, we can first take logarithm and then take square on the above inequality to obtain
	\begin{align*}
		& \log(1-\varepsilon)^2 +2\log(1-\varepsilon) (\gamma_1 - \varepsilon') \log(x)
		+ (\gamma_1 - \varepsilon')^2 \log(x)^2 \\
		<& \left(\log(U_1(tx)) - \log(U_1(t))\right)^2 \\ 
		<& \log(1+\varepsilon)^2 +2\log(1+\varepsilon) (\gamma_1 + \varepsilon') \log(x)
		+(\gamma_1 + \varepsilon')^2 \log(x)^2.
	\end{align*}
	
	For large enough $n$ and for $i \in \{1,\dots,n\}$ such that $Z_i(1) > n/k$, we can set $t = n/k $ and $x = Z_i(1) \cdot k/n$. 
	Multiplying by $\frac{D_i}{k \Pi(X_i)^2}$ on both sides of the above inequality and summing up all $i \in \{1,\dots,n\}$ with $Z_i(1) > n/k$ gives us that almost surely, 
	$G_1^{n,1}$ lies in the interval $[a, b]$ with
	\begin{align*}
		a =& \log(1 - \varepsilon)^2 \frac{1}{k} \Sn  
		\frac{D_i}{ \Pi(X_i)^2} \Ind{Y_i >  q_1(1-\tau_n)} 
		+  2(\gamma_1 - \varepsilon')\log(1- \varepsilon) \frac{1}{k} \Sn  
		\log\left(Z_i(1) \frac{k}{n} \right)
		\frac{D_i}{ \Pi(X_i)^2} \Ind{Y_i >  q_1(1-\tau_n)} \\
		&+
		(\gamma_1 - \varepsilon')^2 \frac{1}{k} \Sn  
		\log^2\left(Z_i(1) \frac{k}{n} \right)
		\frac{D_i}{ \Pi(X_i)^2} \Ind{Y_i >  q_1(1-\tau_n)}, \\
		b =& \log(1+ \varepsilon)^2 \frac{1}{k} \Sn  
		\frac{D_i}{ \Pi(X_i)^2} \Ind{Y_i >  q_1(1-\tau_n)} 
		+  2(\gamma_1 + \varepsilon')\log(1+ \varepsilon) \frac{1}{k} \Sn  
		\log\left(Z_i(1) \frac{k}{n} \right)
		\frac{D_i}{ \Pi(X_i)^2} \Ind{Y_i >  q_1(1-\tau_n)} \\
		&+
		(\gamma_1 + \varepsilon')^2 \frac{1}{k} \Sn  
		\log^2\left(Z_i(1) \frac{k}{n} \right)
		\frac{D_i}{ \Pi(X_i)^2} \Ind{Y_i >  q_1(1-\tau_n)}.
	\end{align*}
	We have that
	\[
	\frac{1}{k} \Sn \frac{D_i}{ \Pi(X_i)^2} \Ind{Y_i >  q_1(1-\tau_n)}
	\le
	\frac{1}{c k} \Sn \frac{D_i}{ \Pi(X_i)} \Ind{Y_i >  q_1(1-\tau_n)}
	=O_{p}(1)
	\]
	by Assumption \ref{ass:1} ii) and the result \eqref{eq:hill_Gn2_lln}, and
	\[ 
	\frac{1}{k} \Sn  \log\left(Z_i(1) \frac{k}{n} \right)
	\frac{D_i}{ \Pi(X_i)^2} \Ind{Y_i >  q_1(1-\tau_n)} 
	\le
	\frac{1}{ck} \Sn  \log\left(Z_i(1) \frac{k}{n} \right)
	\frac{D_i}{ \Pi(X_i)} \Ind{Y_i >  q_1(1-\tau_n)} 
	= O_p(1).
	\]
	by Assumption \ref{ass:1} ii) and the result \eqref{eq:hill_Gn1_term2}. 
	Hence, since $\varepsilon$ and $\varepsilon'$ can
	be arbitrarily small, to prove $G_1^{n,1} \tendsto{P} G_1$, it is enough to show that
	\[
	\frac{\gamma_1^2}{k} \Sn  
	\log^2\left(Z_i(1) \frac{k}{n} \right)
	\frac{D_i}{ \Pi(X_i)^2} \Ind{Y_i >  q_1(1-\tau_n)} 
	\tendsto{P} G_1.
	\]
	We have
	\begin{align*}
		& \ERW{ \frac{\gamma_1^2}{k} \Sn  
			\log^2\left(Z_i(1) \frac{k}{n} \right)
			\frac{D_i}{ \Pi(X_i)^2} \Ind{Y_i >  q_1(1-\tau_n)}} \\
		=&\frac{\gamma_1^2}{\tau_n} 
		\ERW{ \frac{1}{ \Pi(X_i)} \ERW{ \log^2\left(Z_i(1) \tau_n \right)
				\Ind{Y_i(1) >  q_1(1-\tau_n)} \mid X_i }  } \to G_1
	\end{align*}
	by Lemma~\ref{lemma:easy_sigma_hill}, and
	\begin{align*}
		& \Var  \left( 
		\frac{\gamma_1^2}{k} \Sn  
		\log^2\left(Z_i(1) \frac{k}{n} \right)
		\frac{D_i}{ \Pi(X_i)^2} \Ind{Y_i >  q_1(1-\tau_n)} 
		\right)
		\\ = & 
		\frac{n}{k^2}
		\Var  \left( 
		\gamma_1^2 \log^2\left(Z_1(1) \frac{k}{n} \right)
		\frac{D_1}{ \Pi(X_1)^2} \Ind{Y_1 >  q_1(1-\tau_n)} 
		\right)
		\\ \le &  
		\frac{n}{c^3 k^2}
		\ERW{ \gamma_1^4 \log^4\left(Z_1(1) \frac{k}{n} \right)
			\frac{D_1}{ \Pi(X_1)} \Ind{Y_1 >  q_1(1-\tau_n)} } \\
		= & 
		\frac{n}{c^3 k^2}
		\ERW{
			\gamma_1^4 \log^4\left(Z_1(1) \frac{k}{n} \right)
			\Ind{Y_1(1) >  q_1(1-\tau_n)} 
		} \\
		= & 
		\frac{n}{c^3 k^2} O(\tau_n) \to 0,
	\end{align*}
	where we apply Lemma  \ref{lemma:propensity_adjustment} in the second last equality and apply Lemma~\ref{lemma:s_order} in the last equality.
	This shows that 
	\[
	\frac{\gamma_1^2}{k} \Sn  
	\log^2\left(Z_i(1) \frac{k}{n} \right)
	\frac{D_i}{ \Pi(X_i)^2} \Ind{Y_i >  q_1(1-\tau_n)} 
	\tendsto{P} G_1.
	\]
	and therefore we can conclude that $G_1^{n,1} \tendsto{P} G_1$.
	
	For $G_1^{n,2}$, we have
	\begin{align*}
		G_1^{n,2} =& 
		2\Delta_n
		\frac{1}{k} \Sn 
		\left(  \log(Y_i) - \log( q_1(1-\tau_n))
		\right)
		\frac{D_i}{ \Pi(X_i)^2} \Ind{Y_i >  q_1(1-\tau_n)}
		+ 
		\Delta_n^2
		\frac{1}{k} \Sn 
		\frac{D_i}{ \Pi(X_i)^2} \Ind{Y_i >  q_1(1-\tau_n)} \\
		\le &
		\frac{2\Delta_n}{c}
		\frac{1}{k} \Sn 
		\left(  \log(Y_i) - \log( q_1(1-\tau_n))
		\right)
		\frac{D_i}{ \Pi(X_i)} \Ind{Y_i >  q_1(1-\tau_n)}
		+ 
		\frac{\Delta_n^2}{c}
		\frac{1}{k} \Sn 
		\frac{D_i}{ \Pi(X_i)} \Ind{Y_i >  q_1(1-\tau_n)}\\
		= & o_p(1),
	\end{align*}
	where the inequality follows from the Assumption \ref{ass:1} ii), and the last result follows from results 
	\eqref{eq:hill_Gn2_lln}, \eqref{equa: delta_n=op1}, and the result that 
	\[
	\frac{1}{k} \Sn \left(  \log(Y_i) - \log( q_1(1-\tau_n)) \right) \frac{D_i}{ \Pi(X_i)} \Ind{Y_i >  q_1(1-\tau_n)} \tendsto{P} \gamma_1,
	\]
	as we proved that $G_n^1 \tendsto{P} \gamma_1$ in the proof of Lemma~\ref{lemma:hill_consistency}.

	For $G_1^{n,3}$, since the terms
	$ \Ind{Y_i > \wh{q}_1{(1-\tau_n})} - \Ind{Y_i > {q}_1{(1-\tau_n})} $
	have the same signs for all $i=1,\dots,n$.
	Thus
	\[
	|G_1^{n,3}|
	\le | \Delta_n |^2 
	\sup_x \left| \frac{1}{\wh{\Pi}(x)} \right|
	\left| \frac{1}{k} \Sn 
	\frac{D_i}{  \wh \Pi(X_i)} \left( \Ind{Y_i > \wh q_1(1-\tau_n)} -
	\Ind{Y_i >  q_1(1-\tau_n)} 
	\right)  \right| = o_{p}(1)
	\]
	by the result \eqref{equa: delta_n=op1}, \eqref{eq:prob_bounded_prop} and Lemma~\ref{lemma:weighted_cdf_diff}.
	This proves that $G_1^{n,3} = o_p(1)$.
	
	For $G_1^{n,4}$, similarly as in the proof of Lemma~\ref{lemma:hill_consistency},
	we have
	\[
	|G_1^{n,4}| \le \sup_x \left| \frac{1}{\wh \Pi(x)^2}- \frac{1}{\Pi(x)^2} \right|
	\left( G_1^{n,1} + | G_1^{n,2} | \right).
	\]
	We have shown that $G_1^{n,1} \tendsto{P} G_1$ and $G_1^{n,2} = o_p(1)$, so $G_1^{n,4} = o_p(1)$ follows from 
	Lemma~\ref{lemma:propensity_consistency}. 
	Combining all the previous results we can conclude that $\wh G_1 \tendsto{P} G_1$.
	
	For (iii), we expand
	\begin{align*}
		\wh J_1 = J_{1}^{n,1 } + J_{1}^{n, 2} + J_{1}^{n, 3} + J_1^{n,4}
	\end{align*}
	with
	\begin{align*}
		J_1^{n,1} & = \frac{1}{k} \Sn (\log(Y_i) - \log( q_1(1-\tau_n))) 
		\frac{D_i}{ \Pi(X_i)^2} \Ind{Y_i >  q_1(1-\tau_n)}, \\
		J_1^{n,2} & = 
		\frac{\Delta_n}{k}  \Sn
		\frac{D_i}{ \Pi(X_i)^2} \Ind{Y_i >  q_1(1-\tau_n)}, \\
		J_1^{n,3} &=  \frac{1}{k} \Sn (\log(Y_i) - \log( \wh q_1(1-\tau_n))) 
		\frac{D_i}{  \wh \Pi(X_i)^2} \left( \Ind{Y_i > \wh q_1(1-\tau_n)} -
		\Ind{Y_i >  q_1(1-\tau_n)} 
		\right), \\
		J_1^{n,4} &= \frac{1}{k} \Sn (\log(Y_i) - \log( \wh q_1(1-\tau_n)))
		D_i \Ind{Y_i >  q_1(1-\tau_n)} 
		\left(
		\frac{1}{  \wh \Pi(X_i)^2}- \frac{1}{ \Pi(X_i)^2} 
		\right),
	\end{align*}
	where $\Delta_n = \log( q_1(1-\tau_n)) -\log( \wh q_1(1-\tau_n))$. The next step is to show that $J_1^{n,1} \tendsto{P} J_1$ and 
	that $J_1^{n,2}, J_1^{n,3}, J_1^{n,4}$ converge to zero in probability, which is enough to prove $\wh J_1 \tendsto{P} J_1$. 
	The remaining proof is similar to the one for $\wh G_1$, so we omit it.
\end{proof}

\subsection{Proof of Lemma~\ref{lemma:hill_conservative}}
\begin{proof}[Proof of Lemma~\ref{lemma:hill_conservative}] 
	It is sufficient to show that  $\wt \Sigma - \Sigma$ is a positive semi-definite matrix.
	Note that $\wt \Sigma - \Sigma$ can be written as the limit
	\[
	\wt \Sigma - \Sigma = \lim_{n \to \infty} \frac{\ERW{ \Delta \Sigma_n }}{\tau_n},
	\]
	where 
	\[
	\Delta \Sigma_n := v_n w_n^T
	\]
	with
	\begin{align*}
		v_n := \left( 
		\begin{matrix}
			\frac{1-\Pi(X)}{\Pi(X)} \ERW{S_{1,n}\mid X} \\ 
			-\ERW{S_{0,n}\mid X} \\
			\frac{1-\Pi(X)}{\Pi(X)} P(Y(1) > q_1(1-\tau_n)\mid X) \\
			-P(Y(0) > q_0(1-\tau_n) \mid X)
		\end{matrix}
		\right)
		\quad \text{and} \quad
		w_n := \left( 
		\begin{matrix}
			\ERW{S_{1,n}\mid X} \\ 
			-\frac{\Pi(X)}{1-\Pi(X)} \ERW{S_{0,n}\mid X} \\
			P(Y(1) > q_1(1-\tau_n)\mid X) \\
			-\frac{\Pi(X)}{1-\Pi(X)}  P(Y(0) > q_0(1-\tau_n) \mid X)
		\end{matrix}
		\right).
	\end{align*}
	Therefore, $\Delta \Sigma_n$ is of rank 1 and has at most one non-zero
	eigenvalue.
	In addition, since all entries on the diagonal of $\Delta \Sigma_n$ are non-negative, we have
	$\text{trace}(\Delta \Sigma_n)\ge 0$, which implies that
	$\Delta \Sigma_n$ is positive semi-definite.
	Linearity and monotonicity of the expectation then yield that 
	$\ERW{\Delta \Sigma_n}/ \tau_n$
	is positive semi-definite, which implies the positive semi-definiteness
	of the limit $\wt \Sigma - \Sigma$.

\end{proof}

\end{document}